\documentclass[a4paper,english]{article}
\usepackage{graphicx} 
\usepackage{babel}
\usepackage[margin=2cm]{geometry}
\usepackage{amsmath}
\usepackage{amsfonts}
\usepackage{amssymb}
\usepackage{amsthm}
\usepackage{mathtools}
\usepackage[citestyle = authoryear, style = authoryear, maxbibnames = 10, maxcitenames = 2, dashed = false, url=false, doi=false,uniquelist=false]{biblatex}
\usepackage{csquotes}
\usepackage{xargs}
\usepackage{xifthen}
\usepackage{dsfont}
\usepackage{hyperref}
\usepackage{algorithm}
\usepackage{algorithmic}
\usepackage{subcaption}
\usepackage{nth}
\usepackage{eurosym}

\usepackage{tikz}
\usetikzlibrary{positioning, shapes.multipart}
\usepackage{pgfplots}
\pgfplotsset{compat=1.18}
\usepgfplotslibrary{groupplots}
\usetikzlibrary{fpu}
\usetikzlibrary{shapes.geometric, arrows}
\usepgfplotslibrary{dateplot}

\newcommand{\defeq}{\vcentcolon=}
\newcommand{\pmset}{\left\{-,+\right\}}
\newcommand{\stocks}{\mathbb{I}}
\newcommand{\diff}{\mathrm{d}}
\newcommandx{\E}[2][1={}, 2={}]{\ifthenelse{\isempty{#1}}{\mathbb{E}}{\mathbb{E}^{#1}}\ifthenelse{\isempty{#2}}{}{\left[#2\right]}}
\newcommandx{\Proba}[2][1={}, 2={}]{\ifthenelse{\isempty{#1}}{\mathbb{P}}{\mathbb{P}^{#1}}\ifthenelse{\isempty{#2}}{}{\left[#2\right]}}

\DeclareMathOperator{\tr}{Tr}
\DeclarePairedDelimiter{\floor}{\lfloor}{\rfloor}

\newtheorem{assumption}{Assumption}[subsection]
\newtheorem{lemma}{Lemma}[subsection]
\newtheorem{proposition}{Proposition}[subsection]
\newtheorem{corollary}{Corollary}[subsection]
\newtheorem{remark}{Remark}[subsection]
\newtheorem{theorem}{Theorem}[subsection]

\title{Multi-dimensional queue-reactive model and signal-driven models: a unified framework\footnote{This research benefited from the financial support of the chairs ``Deep finance \& Statistics'' and ``Machine Learning \& systematic methods in finance'' of \'Ecole polytechnique. The author thanks Qube RT for providing the data and the computing power to make this study possible, especially Charles-Henri Kempeners who guided the project and provided accurate insights, and Rémi Desgranges for his valuable advice. The author thanks his PhD advisors Sergio Pulido and Mathieu Rosenbaum for their guidance and careful reading of the paper.}}
\author{Emmanouil Sfendourakis\thanks{Centre de Mathématiques Appliquées (CMAP), CNRS, École polytechnique, Institut Polytechnique de
Paris. \textbf{Email:} emmanouil.sfendourakis@polytechnique.edu}}
\date{}

\begin{document}

\maketitle

\begin{abstract}
    \noindent We present a Markovian market model driven by a hidden Brownian efficient price. In particular, we extend the queue-reactive model, making its dynamics dependent on the efficient price. Our study focuses on two sub-models: a signal-driven price model where the mid-price jump rates depend on the efficient price and an observable signal, and the usual queue-reactive model dependent on the efficient price via the intensities of the order arrivals. This way, we are able to correlate the evolution of limit order books of different stocks. We prove the stability of the observed mid-price around the efficient price under natural assumptions. Precisely, we show that at the macroscopic scale, prices behave as diffusions. We also develop a maximum likelihood estimation procedure for the model, and test it numerically. Our model is them used to backest trading strategies in a liquidation context.
    \end{abstract}
    
    \textbf{Keywords:} Market microstructure, limit order book, queue-reactive model, stochastic stability, hidden process, MLE

\section{Introduction}

In the classical mathematical approach of financial markets, assets are traded continuously at prices that can take any real value, and anybody can trade at these prices at any moment. It is often assumed that they are diffusion processes driven by a Brownian motion. While this approximation is satisfactory at the low-frequency scale, it cannot describe accurately the behavior of traded prices at high frequency.

Indeed, most financial exchanges use a limit order book (LOB) mechanism, see \parencite{gould2013limit} for a thorough description. With this mechanism, trades only occur when someone is willing to pay the price set by a participant that placed a limit order. Prices are constrained to be discrete multiples of a tick size, for example one cent in the New York Stock Exchange or determined by various market quantities in Europe \parencite{amf}.

In the literature, many stochastic models have been presented to describe the dynamics of limit order books \parencite{smith2003statistical,farmer2005predictive,cont2010stochastic,abergel2013mathematical,cont2013price,abergel2015long,lakner2016high}, or simply discrete traded prices \parencite{bacry2013modelling}.

The dependence on the current state of the order book in the dynamics of the order flow has been shown empirically to be of great significance \parencite{muni2017modelling,lehalle2021optimal,morariu2022state,sfendourakis_lob_2023}. \textcite{huang2015simulating} introduce the \enquote{queue-reactive} approach to model the dynamics of the LOB and the price of an asset, introducing a Markovian dependence on the current state of the LOB for the arrival intensities of the orders. It has been extended in various ways. \textcite{wuQRHawkes} consider Hawkes dynamics for the order flow, \textcite{bodor2024novel} refine the model to accurately represent order sizes, and \textcite{bodor2025deep} use a deep-learning approach to model accurately high-dimensional dependencies.

It has been shown that in these models, in the macroscopic scale, the observed prices behave like Brownian motions with constant volatility \parencite{abergel2013mathematical,abergel2015long,huang2017ergodicity,mounjid2019asymptotic} or even rough volatility for Hawkes models with a more sophisticated scaling \parencite{jaisson2015limit}.

Some models directly consider a macroscopic diffusive efficient price, that drives the dynamics of the observed prices that only jump at discrete (random) times. It is the case of \parencite{frey1999risk,ceci2006risk,robert_new_2011,delattre2013estimating,derchu2020bayesian} among others. They are interested in either the estimation of the model parameters, especially the volatility of the latent process, or applications of filtering theory to some optimization problems.

In our paper, we combine both approaches to present a Markovian state-dependent price model driven by an efficient price. This includes an extension of the queue-reactive model, where the order book, up to a fixed price distance from the mid-price, of an asset is represented and may influence the dynamics of the order flow, along the Brownian efficient price. This way, we are able to consider correlated dynamics of the LOB and order flow of different assets, the correlation being originated exclusively from the efficient price. In short, the dynamics of the order flow of an asset depends on two states: its own order book, and its own efficient price, the second being possibly correlated to the efficient price of other assets.

We show that in the macroscopic scale, with a standard square-root scaling, under natural assumptions, the observed price becomes arbitrarily close in probability to the efficient price. To retrieve the macroscopic volatility we can therefore embed it in the efficient price instead of playing with $\theta^{reinit}$ as in \parencite{huang2015simulating}. This way, the LOB fulfills its role in price discovery.

We also propose, without any theoretical proof of convergence, a maximum likelihood estimator of our model, and an approximation algorithm to compute the likelihood which can be represented as the solution of a PDE. We then observe the convergence numerically.

Our model can also be used in the context of optimal execution across multiple assets. It exhibits transient market impact for large market orders. The agent can use it to design and backtest short-term strategies, for which the correlation between asset prices remains relevant.

In Section \ref{sec:intuitive_queue_reactive}, we present the models of interest. In Section \ref{sec:models}, we give the mathematical details and state the large-time scale convergence results. In Section \ref{sec:estimation}, we describe our estimation procedure, and validate it numerically. In Section \ref{sec:real_data}, we apply the estimation procedure to an example with real market data. In Section \ref{sec:liquidation}, we test the behavior of our model in optimal liquidation setting. In Section \ref{section:proofs_stability}, we prove the results stated in Section \ref{sec:models}. In Section \ref{sec:proof_likelihood}, we prove some of the statements of Section \ref{sec:estimation}.

\subsection{Notations}

Finite-dimensional spaces $\mathbb{R}^d$ are endowed with the max-norm $|y| = \max_i |y^i|.$ For two functions $f$ and $g$ defined on a space $E$ and valued in $\mathbb{R}^d$, we introduce the sup-norm $|f-g|_{\infty, E} = \sup_{t \in E} |f(t) - g(t)|$. For a topological space $\mathcal{X}$, we denote by $\mathbb{B}(\mathcal{X})$ its Borel $\sigma$-algebra. We denote by $(e_k)_{k \in \mathbb{M}}$ the canonical basis of $\mathbb{R}^{\mathbb{M}}$, $\mathbb{M}$ being an at most countable set. For $x=(x^j)_{j \in \mathbb{J}}$ and $i \in \mathbb{I}$, $e_i \otimes x =  (y_{i'}^{j})_{i' \in \mathbb{I}, j \in \mathbb{J}} \in (\mathbb{R}^{\mathbb{J}})^{\mathbb{I}}$ is the element where all $y_{i'}^{j} = 0$ if $i' \neq i$, otherwise $y_i^j = x^j$. The set of natural integers, including 0, is denoted by $\mathbb{N}$, and $\mathbb{N}^*=\mathbb{N}\setminus \{0\}$. $D(I, \mathbb{R}^d)$, denotes the Skorokhod space of càdlàg functions on an interval $I$, valued in $\mathbb{R}^d$.

We recall the multi-index notation. For $y=(y^i)_{i \in \mathbb{R}^d}$ and $\alpha =(\alpha_i)_{i \in \mathbb{N}^d} \in \mathbb{N}^d$, $y^{\alpha}=\prod_{i = 0}^d (y^i)^{\alpha_i}$. For $\alpha =(\alpha_i)_{i \in \mathbb{N}^d} \in \mathbb{N}^d$ and $\beta =(\beta_i)_{i \in \mathbb{N}^d} \in \mathbb{N}^d$, $\alpha \leqslant \beta$ if and only if for all $i \in \{1, \dots, d\}$, $\alpha_i \leqslant \beta_i$. Finally, $|\alpha| = \sum_{i=1}^d \alpha_i$.

\section{Presentation of the models}
\label{sec:intuitive_queue_reactive}

In this section, we present two models driven by an underlying efficient price: the signal-driven model and the queue-reactive model. We will formalize this construction rigorously in Section \ref{sec:models} within a more general framework.

We consider $d$ stocks. Each stock $i$ has a tick size $\delta^i > 0$: traded prices can only be multiples of $\delta^i$. We observe its reference price $P_t^i \in \delta^i (\frac{1}{2} + \mathbb{Z})$ lying on the inter-tick grid. This reference price can be a proxy of the mid-price of the asset. It is driven by an unobservable efficient prices $S_t = S_0 + \Sigma W_t$ where $\Sigma$ is a $d \times d$ matrix and $W$ is a $d$-dimensional Brownian motion. The process $S$ represents the \enquote{fair} price of the assets and takes continuous values.

At the macroscopic scale, the models presented below behave like $S$. The long-term correlations observed between asset prices are fully described by the correlation matrix $\Sigma$.

\subsection{Signal-driven model}

Fix a stock $i$. In this model, we are interested in the jumps of its reference price $P_t^i$. Specifically, we observe a counting process $(N^{i,-}, N^{i,+})$: $N^{i,-}$ jumps when $P^i$ decreases by one tick, $N^{i,+}$ jumps when $P^i$ increases by one tick.

Apart from the reference price, we observe a signal $X^i$ affecting the stock dynamics. It can be some quantity directly observable from order book data, such as the spread or the volume imbalance. The process $X^i$ evolves in a Markovian fashion.

The dynamics of $(N^{i,e})_{1 \leqslant i \leqslant d, e \in \pmset}$ are fully described by their intensity $t \to (\Lambda^{i,e}(X_{t-}^i, S_t^i - P_t^i))$. These intensities depend on the current signal $X^i$ and the relative position between the efficient price and the reference price, as in \parencite{derchu2020bayesian}. When the efficient price is higher (resp. lower) than the reference price, the latter will tend to go up (resp. down). This way, both prices are kept relatively close to each other.

\subsection{The queue-reactive model}

The limit order book (LOB) of the asset $i$ is represented as a $2K$-dimensional vector $q_t^i = (q_t^{i,s,j})_{s \in \{b,a\}, 1 \leqslant j \leqslant K} \in \mathbb{N}^{2K}$ around its reference price $P_t^i$, as in \parencite{huang2015simulating}. The quantity $q_t^{i,a,j}$ (resp. $q_t^{i,b,j}$) represents the pending volume on the ask (resp. bid) side at price $P_t^i + \delta^i(-\frac{1}{2}+j)$ (resp. $P_t^i - \delta^i(-\frac{1}{2}+j)$). We consider four types of orders:
\begin{itemize}
    \item Limit orders, adding volume to a specific pile.
    \item Cancellations, removing a pending order from a specific pile.
    \item Market orders, representing an actual transaction.
    \item Modifications, moving a pending limit order to a different price.
\end{itemize}
We denote by $\mathbb{T}$ the set of all order types (the four types mentioned above, together with the volumes and prices). Their effects are illustrated in Figure \ref{fig:illustration_order_types}. We track these orders with a counting process $N$ valued in $\mathbb{N}^{\mathbb{T}}$: $N^e_t$ is the number of orders of type $e \in \mathbb{T}$ that occurred until time $t$.

\begin{figure}
    \centering
    \begin{subfigure}[b]{0.45\textwidth}
        \centering
        \resizebox{\textwidth}{!}{
            \begin{tikzpicture}
                \draw[->] (0,0) -- (7,0) node[right] {Price};
                \draw[->] (0,0) -- (0,3) node[above] {Volume};
            
                \fill[blue!50] (0.75,0) rectangle (1.25,1.5) node[midway, above] {1.5};
                \node[blue!50] at (1.0, 1.8) {$q^{b,3}$};
                \fill[blue!50] (1.75,0) rectangle (2.25,2) node[midway, above] {2};
                \node[blue!50] at (2.0, 2.3) {$q^{b,2}$};
                \fill[blue!50] (2.75,0) rectangle (3.25,1) node[midway, above] {1};
                \node[blue!50] at (3.0, 1.3) {$q^{b,1}$};
            
                \fill[red!50] (4.75,0) rectangle (5.25,0.7);
                \node[red!50] at (5.0, 1.0) {$q^{a,2}$};
                \fill[red!50] (5.75,0) rectangle (6.25,2.5);
                \node[red!50] at (6.0, 2.8) {$q^{a,3}$};
        
                \foreach \x in {1,2,3,4,5,6} {
                \draw (\x, 0.1) -- (\x, -0.1);
            }
            \draw[violet] (3.5, 0.3) -- (3.5, -0.3) node[below] {$P$};
            
            \end{tikzpicture}
        }
        \caption{Initial state.}
    \end{subfigure}
    \hfill
    \begin{subfigure}[b]{0.45\textwidth}
        \centering
        \resizebox{\textwidth}{!}{
            \begin{tikzpicture}
                \draw[->] (0,0) -- (7,0) node[right] {Price};
                \draw[->] (0,0) -- (0,3) node[above] {Volume};
            
                \fill[blue!50] (0.75,0) rectangle (1.25,1.5) node[midway, above] {1.5};
                \node[blue!50] at (1.0, 1.8) {$q^{b,3}$};
                \fill[blue!50] (1.75,0) rectangle (2.25,2) node[midway, above] {2};
                \node[blue!50] at (2.0, 2.3) {$q^{b,2}$};
                \fill[blue!50] (2.75,0) rectangle (3.25,1) node[midway, above] {1};
                \node[blue!50] at (3.0, 1.3) {$q^{b,1}$};
            
                \fill[red!50] (3.75,0) rectangle (4.25,0.5);
                \node[red!50] at (4.0, 0.8) {$v$};
                \fill[red!50] (4.75,0) rectangle (5.25,0.7);
                \node[red!50] at (5.0, 1.0) {$q^{a,2}$};
                \fill[red!50] (5.75,0) rectangle (6.25,2.5);
                \node[red!50] at (6.0, 2.8) {$q^{a,3}$};
        
                \foreach \x in {1,2,3,4,5,6} {
                \draw (\x, 0.1) -- (\x, -0.1);
            }
            \draw[violet] (3.5, 0.3) -- (3.5, -0.3) node[below] {$P$};
            
            \end{tikzpicture}
        }
        \caption{Limit order at the first ask pile.}
    \end{subfigure}

    \vspace{0.5cm}

    \begin{subfigure}[b]{0.45\textwidth}
        \centering
        \resizebox{\textwidth}{!}{
            \begin{tikzpicture}
                \draw[->] (0,0) -- (7,0) node[right] {Price};
                \draw[->] (0,0) -- (0,3) node[above] {Volume};
            
                \fill[blue!50] (0.75,0) rectangle (1.25,1.0) node[midway, above] {1.5};
                \node[blue!50] at (1.0, 1.3) {$q^{b,3}-v$};
                \fill[blue!50] (1.75,0) rectangle (2.25,2) node[midway, above] {2};
                \node[blue!50] at (2.0, 2.3) {$q^{b,2}$};
                \fill[blue!50] (2.75,0) rectangle (3.25,1) node[midway, above] {1};
                \node[blue!50] at (3.0, 1.3) {$q^{b,1}$};
            
                \fill[red!50] (4.75,0) rectangle (5.25,0.7);
                \node[red!50] at (5.0, 1.0) {$q^{a,2}$};
                \fill[red!50] (5.75,0) rectangle (6.25,2.5);
                \node[red!50] at (6.0, 2.8) {$q^{a,3}$};
        
                \foreach \x in {1,2,3,4,5,6} {
                \draw (\x, 0.1) -- (\x, -0.1);
            }
            \draw[violet] (3.5, 0.3) -- (3.5, -0.3) node[below] {$P$};
            
            \end{tikzpicture}
        }
        \caption{Cancellation on the third bid pile.}
    \end{subfigure}
    \hfill
    \begin{subfigure}[b]{0.45\textwidth}
        \centering
        \resizebox{\textwidth}{!}{
            \begin{tikzpicture}
                \draw[->] (0,0) -- (7,0) node[right] {Price};
                \draw[->] (0,0) -- (0,3) node[above] {Volume};
            
                \fill[blue!50] (0.75,0) rectangle (1.25,1.0) node[midway, above] {1.5};
                \node[blue!50] at (1.0, 1.3) {$q^{b,3}-v$};
                \fill[blue!50] (1.75,0) rectangle (2.25,2) node[midway, above] {2};
                \node[blue!50] at (2.0, 2.3) {$q^{b,2}$};
                \fill[blue!50] (2.75,0) rectangle (3.25,1.5) node[midway, above] {1};
                \node[blue!50] at (3.0, 1.8) {$q^{b,1}+v$};
            
                \fill[red!50] (4.75,0) rectangle (5.25,0.7);
                \node[red!50] at (5.0, 1.0) {$q^{a,2}$};
                \fill[red!50] (5.75,0) rectangle (6.25,2.5);
                \node[red!50] at (6.0, 2.8) {$q^{a,3}$};
        
                \foreach \x in {1,2,3,4,5,6} {
                \draw (\x, 0.1) -- (\x, -0.1);
            }
            \draw[violet] (3.5, 0.3) -- (3.5, -0.3) node[below] {$P$};
            
            \end{tikzpicture}
        }
        \caption{Modification of an order pending on the third bid pile to the first bid pile}
    \end{subfigure}

    \caption{Illustration of the LOB and order types. In this example, $q^{a,1}$ is initially zero.}
    \label{fig:illustration_order_types}
\end{figure}
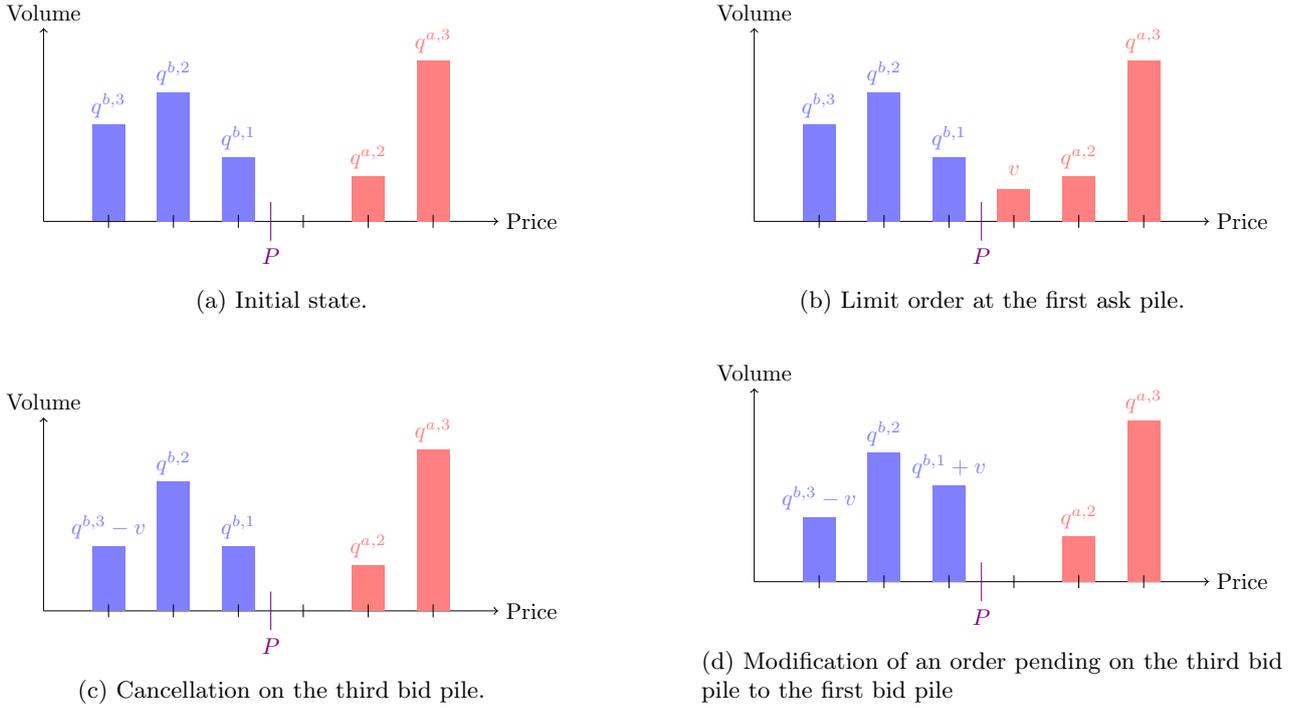

The dynamics of the reference price $P$ are completely determined by the order flow, as detailed in Section \ref{subsec:queue_reactive}, in a way that constrains $P_t$ to always be at most $\frac{\delta}{2}$ far from the mid-price.

We can now fully describe the dynamics of the order flow in our model: $N$ is a counting process with intensity $t \mapsto (\Lambda^{i,e}(q^i_{t-}, S_t - P_{t-}))_{1 \leqslant i \leqslant d, e \in \mathbb{T}}$. The intensities depend not only on the current state of the LOB $q^i_{t-}$ as in \parencite{huang2015simulating}, but also on the relative position of the efficient price and the reference price as in \parencite{derchu2020bayesian}. The efficient price can be seen as an aggregation of market participants' own idea of a fair price to trade. This way, when the efficient price is significantly higher than the best ask price, very few market bid orders will occur since the seller will get a disadvantageous price, and there will be many market orders on the ask side because this is a bargain for the buyer. These phenomena are modeled via the intensities $t \mapsto (\Lambda^{i,e}(q^i_{t-}, S_t - P_{t-}))$. This way, as we show in Section \ref{sec:models}, at the macroscopic scale, the reference price $P$ ends up behaving as the efficient price $S$.

\section{General mathematical framework}
\label{sec:models}

We consider $d$ stocks that are traded on a market. From now on, $\mathbb{I} \defeq \{1,\dots,d\}$. Let $(\Omega, \mathcal{F}, (\mathcal{F}_t)_{t \geqslant 0}, \mathbb{P})$ be a filtered probability space supporting an $\mathbb{R}^d$-valued continuous adapted process $S$ (the unobserved efficient price), an adapted càdlàg pure jump process $X$ valued in a subset $\mathcal{X}$ of a Euclidean space (the state process) and the observed reference price process $P_t$ valued in $\mathcal{P}\defeq \prod_{i\in \stocks} \mathcal{P}^i$, where
\begin{equation*}
    \mathcal{P}^i \defeq \frac{\delta^i}{2} + \delta^i \mathbb{Z},\quad i \in \mathbb{I}
\end{equation*}
is the tick grid for the $i$th stock, $\delta^i > 0$ being its tick size. We also assume there exists a counting process $N=(N^k)_{k \in \mathbb{M}}$ ($\mathbb{M}$ is an at most countable set). A jump of $N^k$ represents an event of type $k$ than we record. In Section \ref{sec:estimation}, we wish to estimate the intensity of these jumps. The set $\mathbb{M}$ represents all the event types (for example market, limit and cancel orders for every stock). The state $X$ and the reference prices $P$ can jump together with $N$, but they can also jump alone. For example, if $N$ represents only the market orders, and $X$ is the imbalance of every LOB, the process $X$ with limit orders and cancellations at the best price.

Furthermore, we suppose that there exists a $d \times d$ invertible matrix $\Sigma$, a family of measurable intensity functions $\Lambda^k:\mathcal{X} \times \mathbb{R}^d \to [0,\infty)$, $k \in \mathbb{M}$, $\Lambda:\mathcal{X} \to [0, \infty)$ and finally the transition probability kernels $\mu^k:\mathcal{X} \times \mathcal{P} \times \mathcal{B}(\mathcal{X} \times \mathcal{P}) \to [0,1]$, $k \in \mathbb{M}$ and $\mu:\mathcal{X} \times \mathcal{P} \times \mathcal{B}(\mathcal{X} \times \mathcal{P}) \to [0,1]$ such that $(S, P, X, N)$ is a Markov process (potentially killed the explosion time of $(S, P, X)$, we allow $N$ to take infinite values) with infinitesimal generator $\hat{\mathcal{A}}$ defined by
\begin{equation*}
    \begin{split}
        \hat{\mathcal{A}}\phi(y, p, x, n)
        &=\frac{1}{2} \tr \big(\Sigma \Sigma^T \nabla^2_y \phi(y, p, x, n)\big)\\
        &\phantom{=}+ \sum_{k \in \mathbb{M}}\bigg(\int_{\mathcal{X} \times \mathcal{P}} \phi(y, p', x', n+e_k)\mu^k(x, p,\diff x', \diff p') - \phi(y, p, x, n) \bigg) \Lambda^k(x, y-p)\\
        &\phantom{=}+
        \bigg(\int_{\mathcal{X} \times \mathcal{P}} \phi(y, p', x', n)\mu(x,p,\diff x', \diff p') - \phi(y, p, x, n) \bigg) \Lambda(x)
    \end{split}
\end{equation*}
for $(y, p, x, n) \in \mathbb{R}^d \times \mathcal{P} \times \mathcal{X} \times \mathbb{N}^{\mathbb{M}}$ and $\phi: \mathbb{R}^d \times \mathcal{P} \times \mathcal{X} \times \mathbb{N}^{\mathbb{M}} \to \mathbb{R}$ measurable, two times differentiable with respect to its first variable, $\nabla_y^2$ representing its Hessian.

The efficient price is a Brownian motion with volatility $\Sigma$. The events of interest $(N^k)$ occur at a rate dependent on the distance between the reference price and the efficient price. Tracking $N$ is only for the purpose of estimation, $N$ does not actually influence the dynamics of $(S,P,X)$ which are Markovian on their own right with infinitesimal generator $\mathcal{A}$ defined by
\begin{equation*}
    \begin{split}
        \mathcal{A}\phi(y, p, x)
        &=\frac{1}{2} \tr \big(\Sigma \Sigma^T \nabla^2_y \phi(y, p, x)\big)\\
        &\phantom{=}+ \sum_{k \in \mathbb{M}}\bigg(\int_{\mathcal{X} \times \mathcal{P}} \phi(y, p', x')\mu^k(x,p,\diff x', \diff p') - \phi(y, p, x) \bigg) \Lambda^k(x, y-p)\\
        &\phantom{=}+
        \bigg(\int_{\mathcal{X} \times \mathcal{P}} \phi(y, p', x')\mu(x,p,\diff x', \diff p') - \phi(y, p, x) \bigg) \Lambda(x)
    \end{split}
\end{equation*}
for $(y, p, x) \in \mathbb{R}^d \times \mathcal{P} \times \mathcal{X}$ and $\phi: \mathbb{R}^d \times \mathcal{P} \times \mathcal{X}\to \mathbb{R}$ measurable, two times differentiable with respect to its first variable.

In the following subsections, we present particular cases assuring that the reference price does not stay \enquote{too far} from the efficient price. Specifically, we define the rescaled processes $\tilde{P}^{(n)}$ and $\tilde{S}^{(n)}$ by
\begin{equation*}
    \tilde{P}^{(n)}_t = \frac{1}{\sqrt{n}}(P_{nt} - P_0)\text{ and }
    \tilde{S}^{(n)}_t = \frac{1}{\sqrt{n}}(S_{nt} - S_0),\quad \forall n \in \mathbb{N}^*,\ t \geqslant 0.
\end{equation*}
We show that in the presented models, $\tilde{P}^{(n)}$ and $\tilde{S}^{(n)}$ become arbitrarily close in probability as $n$ tends to infinity. This way, the order book fills its role of price revelation.

\subsection{Signal-driven model}
\label{subsec:signal_modulated}

In this subsection, we suppose that the reference price can increase or decrease by one tick at a time. Its jumps have an intensity depending on a Markovian signal. The signal can be a quantity directly observable on the LOB such as the volume imbalance or the spread.

We now state the assumptions needed for the stability of the model.

Assumption \ref{assumption:signal_modulated_main_assumption} below is the main assumption characterizing the particular model of this subsection. The transition kernel $\tilde{\mu}^{i,+}$ (resp. $\tilde{\mu}^{i,+}$) describes the transition behavior of the signal when the reference price of the $i$th asset jumps up (resp. down). Note that the signal can also jump alone.
\begin{assumption}
    \label{assumption:signal_modulated_main_assumption}
    The set $\mathbb{M}$ is equal to $\stocks \times \pmset$. For $i \in \stocks$ there exists two transition probability kernels $\tilde{\mu}^{i, +}:\mathcal{X} \times \mathcal{B}(\mathcal{X}) \to [0,1]$ and $\tilde{\mu}^{i, -}:\mathcal{X} \times \mathcal{B}(\mathcal{X}) \to [0,1]$ such that for $x \in \mathcal{X}$, $A \in \mathcal{B}(\mathcal{X})$ and $p \in \mathcal{P}$,
    \begin{equation*}
        \mu^{i, +}(x,p,A \times \{p+\delta^i e_i\}) = \tilde{\mu}^{i, +}(x, A) \text{ and }
        \mu^{i, -}(x,p,A \times \{p-\delta^i e_i\}) = \tilde{\mu}^{i, -}(x, A).
    \end{equation*}
    There exists two transition probability kernels $\tilde{\mu}^{+}:\mathcal{X} \times \mathcal{B}(\mathcal{X}) \to [0,1]$ and $\tilde{\mu}^{-}:\mathcal{X} \times \mathcal{B}(\mathcal{X}) \to [0,1]$ such that for $x \in \mathcal{X}$, $A \in \mathcal{B}(\mathcal{X})$ and $p \in \mathcal{P}$,
    \begin{equation*}
        \mu^{+}(x,p,A \times \{p\}) = \tilde{\mu}^{+}(x, A) \text{ and }
        \mu^{-}(x,p,A \times \{p\}) = \tilde{\mu}^{-}(x, A).
    \end{equation*}
    Furthermore, $\mathcal{X}$ is compact. For $i \in \stocks$, there exist two measurable intensity functions $\tilde{\Lambda}^{i, +}:\mathcal{X} \times \mathbb{R} \to [0, \infty)$ and $\tilde{\Lambda}^{i, -}:\mathcal{X} \times \mathbb{R} \to [0, \infty)$ such that for $x \in \mathcal{X}$ and $y = (y^i)_{i \in \stocks} \in \mathbb{R}^d$,
    \begin{equation*}
        \Lambda^{i, +}(x,y)= \tilde{\Lambda}^{i, +}(x, y^i) \text{ and }
        \Lambda^{i, -}(x,y)= \tilde{\Lambda}^{i, -}(x, y^i).
    \end{equation*}
\end{assumption}

Under Assumption \ref{assumption:signal_modulated_main_assumption}, $(S-P, X)$ is Markovian with infinitesimal generator $\mathcal{L}$ defined below. For $\phi$ regular enough and $(y,x) \in \mathbb{R}^d \times \mathcal{X}$,
\begin{equation*}
    \begin{split}
        \mathcal{L} \phi (y,x)
        &=\frac{1}{2} \tr \big(\Sigma \Sigma^T \nabla^2_y \phi(y, x)\big)\\
        &\phantom{=}+ \sum_{i \in \stocks}\bigg(\int_{\mathcal{X}} \phi(y-\delta^i e_i,  x')\tilde{\mu}^{i,+}(x,\diff x') - \phi(y, x) \bigg) \tilde{\Lambda}^{i, +}(x, y^i)\\
        &\phantom{=}+ \sum_{i \in \stocks}\bigg(\int_{\mathcal{X}} \phi(y+\delta^i e_i,  x')\tilde{\mu}^{i,-}(x,\diff x') - \phi(y, x) \bigg) \tilde{\Lambda}^{i, -}(x, y^i)\\
        &\phantom{=}+
        \bigg(\int_{\mathcal{X}} \phi(y, x')\tilde{\mu}(\diff x') - \phi(y, x) \bigg) \Lambda(x).
    \end{split}
\end{equation*}

Let $U:\mathbb{R}\to (0, \infty)$ be a $C^{\infty}$ function such that $U(y)=|y|$ for $y \in (-\infty,-1] \cup [1, \infty)$ and $U(y) \leqslant 1$ for $|y| \leqslant 1$. For $y = (y^i)_{i \in \stocks} \in \mathbb{R}^d$ and $x \in \mathcal{X}$, define $V(y, x) \defeq \prod_{i \in \stocks} e^{U(y^i)}$. The function $V$ is norm-like.

Assumption \ref{assumption:limits} below is crucial for the stability of the process $S-P$. It states that is $P$ is too far from $S$, there is a high probability of a jump that will bring back $P$ closer to $S$ in the near future.
\begin{assumption}
    \label{assumption:limits}
    For $i \in \stocks$, $\tilde{\Lambda}^{i, +}$ and $\tilde{\Lambda}^{i, -}$ are locally bounded and
    \begin{align*}
        \lim_{y \to \infty}\inf_{x \in C} \tilde{\Lambda}^{i, +}(x, y) = \infty,\ 
        \limsup_{y \to -\infty}\sup_{x \in C} \tilde{\Lambda}^{i, +}(x, y) < \infty,\\
        \limsup_{y \to \infty}\sup_{x \in C} \tilde{\Lambda}^{i, -}(x, y) < \infty,\ 
        \lim_{y \to \infty}\inf_{x \in C} \tilde{\Lambda}^{i, -}(x, y) = \infty.
    \end{align*}
\end{assumption}

Proposition \ref{prop:exp_inequality} below, proved in Section \ref{subsec:proofs_signal_modulated}, shows a Lyapunov inequality for $V$ that will ensure the stability of $(S-P, X)$.
\begin{proposition}
    \label{prop:exp_inequality}
    Under Assumptions \ref{assumption:signal_modulated_main_assumption} and \ref{assumption:limits}, there exists $K \geqslant 1$ such that for all $y \in \mathbb{R}^d$ and $x \in \mathcal{X}$,
    \begin{equation*}
        \mathcal{L}V(x, y) \leqslant -V(y) + K.
    \end{equation*}
\end{proposition}

By Proposition \ref{prop:exp_inequality} and \parencite[Theorem 2.1]{meyntweedieIII}, $(S-P, X)$ is non-explosive. Proposition \ref{prop:non_explosiveness_lyapunov} then gives us Corollary \ref{corol:signal_modulated_non_explosive}.

\begin{corollary}
    \label{corol:signal_modulated_non_explosive}
    Under Assumptions \ref{assumption:signal_modulated_main_assumption} and \ref{assumption:limits}, the Markov process $(S,P,X,N)$ is non-explosive.
\end{corollary}

As a direct consequence of Proposition \ref{prop:exp_inequality}, we have that $\tilde{P}^{(n)}$ and $\tilde{S}^{(n)}$ become arbitrarily close in probability, uniformly on every compact, as $n$ tends to infinity. This is a particular case of Proposition \ref{prop:lyapunov_convergence_proba_general}.

\begin{proposition}
    \label{prop:signal_modulated_convergence_proba}
    Suppose that $S_0 - P_0$ is bounded.  Under Assumptions \ref{assumption:signal_modulated_main_assumption} and \ref{assumption:limits}, for all $T, \epsilon > 0$,
    \begin{equation*}
        \lim_{n\to \infty} \Proba[][\big| \tilde{S}^{(n)}-\tilde{P}^{(n)}\big|_{\infty, [0,T]} \geqslant \epsilon] = 0.
    \end{equation*}
\end{proposition}

The corollary below states that, in the macroscopic scale, $P$ behaves like $S$. It is a direct consequence of Corollary \ref{corol:skorokhod_convergence}.

\begin{corollary}
    Under Assumptions \ref{assumption:signal_modulated_main_assumption} and \ref{assumption:limits}, $(\tilde{P}^{(n)})_{n \in \mathbb{N}^*}$ converges in distribution in $D([0, \infty), \mathbb{R}^d)$ to $\Sigma B$ where $B$ is a $d$-dimensional Brownian motion.
\end{corollary}

\subsection{The queue-reactive model}
\label{subsec:queue_reactive}

\subsubsection{Model and notations}

We now concentrate on the queue-reactive model of \textcite{huang2015simulating}. Instead of modelling just a signal driving the price as in Subsection \ref{subsec:signal_modulated}, we model the whole state of the order book at a bounded distance from the reference price. For each stock $i$, we consider its limit order book (LOB) up to $K \in \mathbb{N}^*$ piles, on the bid and ask sides. We define the sets
\begin{align*}
    \mathbb{J} &\defeq \{a, b \} \times \{1,\dots, K\},
    \\
    \mathcal{Q} &\defeq \left\{q = (q^j)_{j \in \mathbb{J}}: \forall j \in \mathbb{J}, q^j \in \mathbb{N}, \exists (l,l') \text{ s.t. }q^{b,l}, q^{a,l'}>0, \left|\inf\{l:q^{a,l} > 0\} - \inf\{l:q^{b,l} > 0\}\right| \leqslant 1\right\}.
\end{align*}

$\mathcal{Q}$ represents all the possible states of the LOB of a stock, $q^j$ being the volume pending at pile $j$. $(a, l) \in \mathbb{J}$ represents the pile on the ask side, $l$ ticks above the reference price. $(b, l) \in \mathbb{J}$ represents the pile on the bid side, $l$ ticks below the reference price. As it is built, the reference price cannot more than half a tick far away from the mid-price. We impose there is at least one pending limit order on the bid and on the ask side of the observable LOB. For an order book state $q = (q^j)_{j \in \mathbb{J}} \in \mathcal{Q}$, we define the best bid and ask piles of the LOB as the ones with the best proposed price with non-zero volume pending, as well as the second-best:
\begin{align*}
    j^{a,best}(q) &\defeq \big(a, \inf \{l \in \{1,\dots,K\}: q^{a,l} > 0\}\big)\\
    j^{b,best}(q) &\defeq \big(b, \inf \{l \in \{1,\dots,K\}: q^{b,l} > 0\}\big) \\
    j^{a,best,2}(q) &\defeq \big(a, \inf \{l \in \{j^{a,best}(q)+1,\dots,K\}: q^{a,l} > 0\} \wedge(K+2)\big)\\
    j^{b,best,2}(q) &\defeq \big(b, \inf \{l \in \{j^{b,best}(q)+1,\dots,K\}: q^{b,l} > 0\} \wedge(K+2)\big).
\end{align*}
Imposing the second-best pile to be at $K + 2$ ticks is only one pile is filled on one side is for notational convenience, and the reason will appear below when we treat the reference price jumps.

We also define the subtraction operator on $\mathbb{J}$ which corresponds to the (signed) distance in ticks of two piles:
\begin{equation*}
    (s_1, l_1) - (s_2, l_2) = 
    \begin{cases}
        l_1 - l_2 &\text{if }s_1=s_2=a\\
        l_2 - l_1 &\text{if }s_1=s_2=b\\
        l_1 + l_2 -1 &\text{if }s_1=a \text{ and } s_2=b\\
        -l_1 - l_2 +1 &\text{if }s_1=b \text{ and } s_2=a
    \end{cases},\quad (s_1, l_1), (s_2, l_2) \in \mathbb{J}.
\end{equation*}
Using this operator, the spread in ticks for a LOB state $q \in \mathcal{Q}$ is $spread(q)\defeq j^{a,best}(q) - j^{b,best}(q)$. When the spread is odd, a bigger part of it is either above the reference price (the ask side) or below it (the bid side). For $q \in \mathcal{Q}$ and $j^{a,best}(q) = (a,l_1)$, $j^{b,best}(q) = (b,l_2)$, define $d(q) \defeq l_1 - l_2 \in \{-1, 0, 1\}$ the direction which the spread covers more.

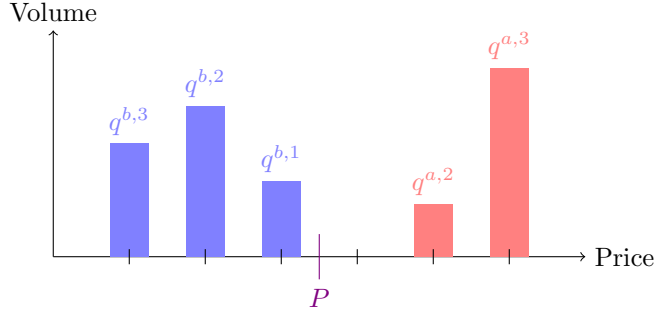
\begin{figure}
    \centering
    \begin{tikzpicture}
        \draw[->] (0,0) -- (7,0) node[right] {Price};
        \draw[->] (0,0) -- (0,3) node[above] {Volume};
    
        \fill[blue!50] (0.75,0) rectangle (1.25,1.5) node[midway, above] {1.5};
        \node[blue!50] at (1.0, 1.8) {$q^{b,3}$};
        \fill[blue!50] (1.75,0) rectangle (2.25,2) node[midway, above] {2};
        \node[blue!50] at (2.0, 2.3) {$q^{b,2}$};
        \fill[blue!50] (2.75,0) rectangle (3.25,1) node[midway, above] {1};
        \node[blue!50] at (3.0, 1.3) {$q^{b,1}$};
    
        \fill[red!50] (4.75,0) rectangle (5.25,0.7);
        \node[red!50] at (5.0, 1.0) {$q^{a,2}$};
        \fill[red!50] (5.75,0) rectangle (6.25,2.5);
        \node[red!50] at (6.0, 2.8) {$q^{a,3}$};

        \foreach \x in {1,2,3,4,5,6} {
        \draw (\x, 0.1) -- (\x, -0.1);
    }
    \draw[violet] (3.5, 0.3) -- (3.5, -0.3) node[below] {$P$};
    
    \end{tikzpicture}
    \caption{Illustration of a limit order book. Here, $j^{a, 1}(q)$ is empty, $j^{b,best}(q) = 1$, $j^{a,best}(q) = 2$ and $d(q) = 1$.}
    \label{fig:LOB}
\end{figure}

Let $\mathbb{T}\defeq \mathbb{T}^{-} \cup \mathbb{T}^{+} \cup \mathbb{T}^{modif}$ be the set of all event types for a stock where
\begin{align*}
    \mathbb{T}^{-} &\defeq \{-\} \times \mathbb{J} \times \mathbb{N}^*,\\
    \mathbb{T}^{+} &\defeq \{+\} \times \mathbb{J} \times \{b, a\} \times \mathbb{N}^*,\\
    \mathbb{T}^{modif} &\defeq \{\text{modif}\} \times \mathbb{J} \times  \mathbb{J}\times\mathbb{N}^*.
\end{align*}
An event $(-, j, n) \in \mathbb{T}^{-}$ corresponds to an order consuming $n$ units of the asset at pile $j$. An event $(-, j, b, n) \in \mathbb{T}^{+}$ corresponds to a buy limit order at pile $j$ with volume $n$, $(-, j, a, n) \in \mathbb{T}^{+}$ corresponds to a sell limit order at pile $j$ with volume $n$. An event $(\text{modif}, j_1, j_2, n)$ is a modification: a quantity $n$ that was pending at pile $j_1$ is now pending at pile $j_2$. A buy limit order is always modified to a buy limit order, and a sell limit order is always modified to a sell limit order.

\begin{remark}
    In practice, there are two types of orders that consume liquidity, represented by a \enquote{-} here: market orders and cancellations. Assuming market orders do not consume more than one pile, both have exactly the same role in LOB dynamics, so we do not separate them in the stability analysis. It is still possible to consider them separately for estimation purposes, see Section \ref{sec:estimation}.
\end{remark}

\begin{remark}
    It is necessary to specify if a limit order is buy or sell, mentioning the pile is not enough for orders in the spread. For example, if piles $(b, 2)$, $(b, 1)$ and $(a,2)$ are filled but $(a, 1)$ is empty, a limit order at pile $(a,1)$ could be either buy or sell.
\end{remark}

For each event type $e \in \mathbb{T}$, we define the set $\mathcal{Q}^{e, licit} \subset \mathcal{Q}$ of order book states for which such an event can occur. For $e=(-,j,n) \in \mathbb{T}^{-}$,
\begin{equation*}
    \mathcal{Q}^{e, licit} \defeq \{(q^{j'})_{j' \in \mathbb{J}} \in \mathcal{Q}: q^j \geqslant n\},
\end{equation*}
meaning that an order consuming a volume $n$ of the asset at pile $j$ is only possible when the said pile has at least $n$ units pending. For $e = (+,(s_1,l),s2,n) \in \mathbb{T}^{+}$,
\begin{equation*}
    \mathcal{Q}^{e,licit} \defeq
    \begin{cases}
        \mathcal{Q} & \text{if }s_1=s_2\\
        \{q \in \mathcal{Q}: l < j^{a,best}(q)\} &\text{if }s_1=a\text{ and }s_2=b\\
        \{q \in \mathcal{Q}: l < j^{b,best}(q)\} &\text{if }s_1=b\text{ and }s_2=a,
    \end{cases}
\end{equation*}
meaning a bid limit order cannot be placed at a price higher or equal than the best ask price, and an ask bid limit order cannot be placed at a price lower or equal than the best ask price. Finally, a modification $e=(\text{modif}, (s_1, l_1), j_2, n)\in \mathbb{T}^{modif}$ is licit if and only if n units of the asset can be taken from $(s_1, l_1)$ and can be added to $j_2$:
\begin{equation*}
    \mathcal{Q}^{e, licit} \defeq \begin{cases}Q^{(-, (s_1, l_1), n),licit} \cap Q^{(+, j_2, s_1, n),licit}& \text{if } j_2 \neq (s_1,l_1)\\
       \emptyset & \text{if } j_2 = (s_1,l_1)
    \end{cases},
\end{equation*}
excluding modifications in the same queue.

As in \parencite[Section 2.2.2]{huang2015simulating}, we consider a reference price that is either the mid-price if the spread is odd, or at half-tick from the mid-price if the spread is even. At a mid-price change, for example when the best queue is wiped or a limit order is placed in the spread, the new reference price is either:
\begin{itemize}
    \item the new mid-price if the new spread is even,
    \item the new mid-price plus or minus half a tick if the spread is odd, whatever is closer to the old reference price.
\end{itemize}
Precisely, we define the price variation function $\delta p^{e}$ for each event type $e \in \mathbb{T}$. For a liquidity consuming event $e=(-,j=(s,l), n) \in \mathbb{T}^{-}$,
\begin{equation*}
    \delta p^e(q) = \begin{cases}
        0 &\text{if } j \neq j^{s,best}(q)\text{ or }n \neq q^{j}\\
        \floor*{\frac{1}{2}(d(q)+j^{a,best,2}(q)-j^{a,best}(q))} &\text{if } s=a \text{ and } l = j^{a,best}(q)\text{ and }n = q^{j}\\
        \floor*{\frac{1}{2} + \frac{1}{2}(d(q)+j^{b,best,2}(q)-j^{b,best}(q))} &\text{if } s=b \text{ and } j = j^{b,best}(q)\text{ and }n = q^{j}
    \end{cases},\quad q \in \mathcal{Q}^{e,licit}.
\end{equation*}
For a limit order $e=(+,j,s,n) \in \mathbb{T}^{+}$,
\begin{equation*}
    \delta p^e(q) = \begin{cases}
        0 &\text{if } s=b\text{ and }j-j^{b,best}(q) \leqslant 0\\
        0& \text{if } s=a\text{ and }j-j^{a,best}(q) \geqslant 0\\
        \floor*{\frac{1}{2} + \frac{1}{2}(d(q)+j-j^{a,best}(q))} & \text{if }s=a \text{ and } j - j^{a,best}(q) < 0 \\
        \floor*{\frac{1}{2}(d(q)+j-j^{b,best}(q))} & \text{if }s=b \text{ and } j - j^{b,best}(q) > 0 
    \end{cases},\quad q \in \mathcal{Q}^{e,licit}.
\end{equation*}
For a modification $e=(\text{modif},j_1=(s,l),j_2,n) \in \mathbb{T}^{modif}$,
\begin{equation*}
    \delta p^e(q) = \begin{cases}
        \delta p^{(+,j_2,s, n)}(q) &\text{if } s=b\text{ and }j-j^{b,best}(q) > 0\\
        \delta p^{(+,j_2,s, n)}(q) &\text{if } s=a\text{ and }j-j^{a,best}(q) < 0\\
        \min\bigg(\begin{aligned}[t] &\floor*{\frac{1}{2}(d(q)+j^{a,best,2}(q)-j^{a,best}(q))},\\ &\floor*{\frac{1}{2}(d(q)+j_2-j^{a,best}(q))} \bigg)\end{aligned} & 
        \begin{aligned}[t]&\text{if }s=a \text{ and } j_1=j^{a,best}(q)\\ &\text{and }j_2 - j^{a,best}(q) > 0\\ & \text{and }n=q^j \end{aligned}\\
        \max\bigg(\begin{aligned}[t] &\floor*{\frac{1}{2} + \frac{1}{2}(d(q)+j^{b,best,2}(q)-j^{b,best}(q))},\\&\floor*{\frac{1}{2} + \frac{1}{2}(d(q)+j_2-j^{b,best}(q))} \bigg)\end{aligned} & 
        \begin{aligned}[t]&\text{if }s=b \text{ and } j_1=j^{b,best}(q)\\ &\text{and }j_2 - j^{b,best}(q) < 0\\ & \text{and }n=q^j \end{aligned} 
    \end{cases}, q \in \mathcal{Q}^{e,licit}.
\end{equation*}

\begin{remark}
    Since $j^{s,best,2}(q)$, $s \in \{a,b\}$ is fixed at $K+2$ if only one pile of side $s$ contains pending limit orders, a depletion of the full side $s$ will always trigger a reference price jump. Otherwise, the reference price would never jump in the case $K = 1$.
\end{remark}

We also define the variation in volume $\delta q$ in the LOB after every type of order. This quantity does not depend on the current state of the LOB. For $e \in \mathbb{T}$, we have
\begin{equation*}
    \delta q^e = \begin{cases}
        -ne_j & \text{if } e=(-,j,n) \in \mathbb{T}^{-}\\
        ne_j & \text{if } e=(+,j,s,n) \in \mathbb{T}^{+}\\
        -ne_{j_1} + ne_{j_2} & \text{if } e=(+,j_1,j_2,n) \in \mathbb{T}^{modif}.
    \end{cases}
\end{equation*}
We now introduce a notation to describe the state of the LOB after new piles are discovered due to a shift of the observation window. Let $q  = (q^j)_{j \in \mathbb{J}}\in \mathbb{N}^{\mathbb{J}}$ and $v=(v^m)_{1 \leqslant m \leqslant K'} \in \mathbb{N}^{K'}$ for some $1 \leqslant K' \leqslant K$. Then $[v, q] = (q^j_1)_{j \in \mathbb{J}}$ and $[q,v] = (q^j_2)_{j \in \mathbb{J}}$ where
\begin{align*}
    &(q_1^{b,K},\dots,q_1^{b,1},q_1^{a,1},\dots,q_1^{a,K})
    = (v^{K'},\dots,v^1,q^{b,K},\dots,q^{b,1}, q^{a,1},\dots,q^{a,K-K'})\\
    &(q_2^{b,K},\dots,q_2^{b,1},q_2^{a,1},\dots,q_2^{a,K})
    = (q^{b,K-K'},\dots,q^{b,1}, q^{a,1},\dots,q^{a,K-K'},v^1,\dots,v^K).
\end{align*}
For $q = (q_i^j)_{i \in \mathbb{I},j\in \mathbb{J}} \in (\mathbb{N}^{\mathbb{J}})^{\mathbb{I}}$ and $i \in \stocks$ we define the shifts on the component $i$ $[v,q]^{(i)}=(q^*_{i'})_{i' \in \stocks} \in (\mathbb{N}^{\mathbb{J}})^{\mathbb{I}}$ and $[q,v]^{(i)}=(q^{**}_{i'})_{i' \in \stocks} \in (\mathbb{N}^{\mathbb{J}})^{\mathbb{I}}$ by
$q^*_{i'} = q^{**}_{i'} = q_{i'}$ if $i' \neq  i$, $q^*_i = [v,q]_i$ and $q^{**}_i = [q,v]_i$. This bracket notation is illustrated in Figure \ref{fig:illustration_crochet}.

\begin{figure}
    \centering
    \begin{subfigure}[b]{0.45\textwidth}
        \centering
        \resizebox{\textwidth}{!}{
        \begin{tikzpicture}
            \draw[->] (0,0) -- (7,0) node[right] {Price};
            \draw[->] (0,0) -- (0,3) node[above] {Volume};
        
            \fill[blue!50] (0.75,0) rectangle (1.25,1.8);
            \node[blue!50] at (1.0, 2.1) {$v^2$};
            \fill[blue!50] (1.75,0) rectangle (2.25,1.1);
            \node[blue!50] at (2.0, 1.4) {$v^1$};
            \fill[blue!50] (2.75,0) rectangle (3.25,1.5);
            \node[blue!50] at (3.0, 1.8) {$q^{b,3}$};
        
            \fill[red!50] (3.75,0) rectangle (4.25,2);
            \node[red!50] at (4.0, 2.3) {$q^{b,2}$};
            \fill[red!50] (4.75,0) rectangle (5.25,1);
            \node[red!50] at (5.0, 1.3) {$q^{b,1}$};
    
            \foreach \x in {1,2,3,4,5,6} {
            \draw (\x, 0.1) -- (\x, -0.1);
        }
        \draw[violet] (3.5, 0.3) -- (3.5, -0.3) node[below] {$P$};
        
        \end{tikzpicture}
        }
        \caption{$[v,q]$}
    \end{subfigure}
    \hfill
    \begin{subfigure}[b]{0.45\textwidth}
        \centering
        \resizebox{\textwidth}{!}{
        \begin{tikzpicture}
            \draw[->] (0,0) -- (7,0) node[right] {Price};
            \draw[->] (0,0) -- (0,3) node[above] {Volume};
        
            \fill[blue!50] (0.75,0) rectangle (1.25,1);
            \node[blue!50] at (1.0, 1.3) {$q^{b,1}$};
            \fill[blue!50] (2.75,0) rectangle (3.25,0.7);
            \node[blue!50] at (3.0, 1.0) {$q^{a,2}$};
        
            \fill[red!50] (3.75,0) rectangle (4.25,2.5);
            \node[red!50] at (4.0, 2.8) {$q^{a,3}$};
            \fill[red!50] (4.75,0) rectangle (5.25,1.1);
            \node[red!50] at (5.0, 1.4) {$v^1$};
            \fill[red!50] (5.75,0) rectangle (6.25,1.8);
            \node[red!50] at (6.0, 2.1) {$v^2$};
    
            \foreach \x in {1,2,3,4,5,6} {
            \draw (\x, 0.1) -- (\x, -0.1);
        }
        \draw[violet] (3.5, 0.3) -- (3.5, -0.3) node[below] {$P$};
        
        \end{tikzpicture}}
        \caption{$[q,v]$}
    \end{subfigure}

    \caption{Illustration of the notations $[v,q]$ and $[q,v]$ where $q$ is the one drawn in Figure \ref{fig:LOB}.}
    \label{fig:illustration_crochet}
\end{figure}
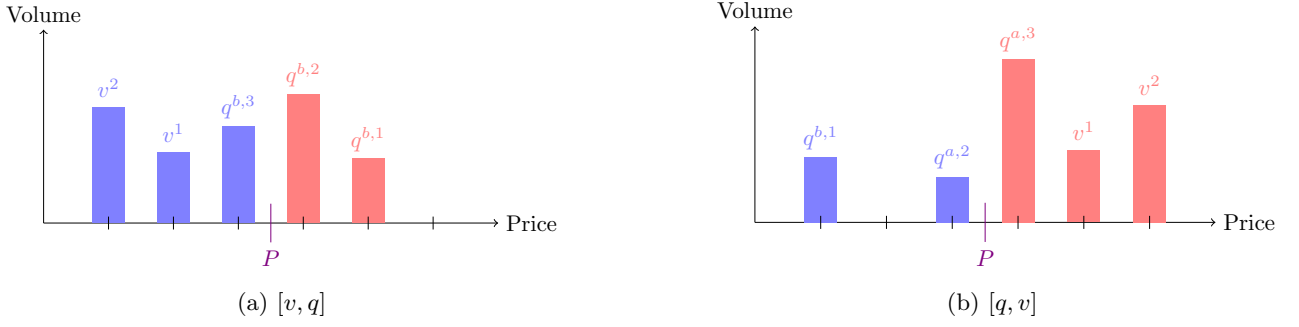

Assumption \ref{assumption:QR_main_assumption} below is the main assumption characterizing the multi-stock queue-reactive models we study here. It will be assumed to hold in this whole Subsection. It relies heavily on the notations we introduced in this section. It can be summarized as follows:
\begin{itemize}
    \item An event only affects one asset at a time.
    \item If an event does not trigger a reference price modification, the volumes are updated in the order book according to the order type (add volume to the specified pile for a limit order, remove volume in the case of a market order or a cancellation, remove volume from a pile and add it to another in the case of a modification).
    \item If an event triggers a reference price jump, for example an increase of $j$ ticks, the volumes are first updated in the LOB. Then, the observation window slides upwards by $j$ ticks. We discover $j$ new ask piles, whose volumes are drawn randomly. In parallel, $j$ of the piles we were observing on the bid side disappear. If the bid side is left empty, some volume is added, at random, on the last observed bid pile. It can be interpreted as quick modifications from orders that were too far from the opposite side still expecting to be executed. This way, we ensure there is liquidity on the bid and ask sides at all times. One can observe that the mid-price stays at half-tick from the reference price in any case. The same approach was used in the models with $K=1$ of \parencite{cont2012order,cont2013price,lehalle2017limit,lehalle2021optimal}.
\end{itemize}
\begin{assumption}
    \label{assumption:QR_main_assumption}
    The set $\mathbb{M}$ is equal to $\stocks \times \mathbb{T}$. The state space $\mathcal{X}$ is equal to $\mathcal{Q}^{\stocks}$. The intensity $\Lambda \equiv 0$. For $i \in \stocks$, there exist probability measures $\mathcal{V}^{i,b,1}$, $\mathcal{V}^{i,a,1}$, \dots, $\mathcal{V}^{i,b,K}$, $\mathcal{V}^{i,a,K}$ on the sets $\mathbb{N^*}$, $\mathbb{N}^*$, \dots, $(\mathbb{N^*})^K$, $(\mathbb{N}^*)^K$ respectively such that everything below holds. We suppose they all have finite first-order moments. Let $i \in \stocks$, $e \in \mathbb{T}$, $q=(q_{i'})_{i' \in \stocks} \in \mathcal{Q}^{\stocks}$ and $p=(p_{i'})_{i' \in \stocks}$.
    \begin{itemize}
        \item If $\delta p^e(q_i) = 0$, then $\mu^{i,e}(q,p, \{q + e_i \otimes \delta q^e\} \times \{p\} ) = 1$.
        \item If $p^e(q_i) > 0$, let $(q'^j)_{j \in \mathbb{J}} \defeq q + \delta q^e$. If $(q'^{(b,1)},\dots,q'^{(b,\delta p^e(q_i))}) = (0,\dots,0)$, then
        \begin{equation*}
            \mu^{i,e}(q,p, \{[q + e_i \otimes \delta q^e, v]^{(i)} + e_i \otimes (w e_{b,K})\} \times \{p+\delta^i\delta p^e(q_i)e_i\} ) = \mathcal{V}^{i,a,\delta p^e(q_i)}(\{v\})\mathcal{V}^{i,b,1}(\{w\})
        \end{equation*}
        for $(v,w) \in \mathbb{N}^* \times (\mathbb{N}^*)^{\delta p^e(q_i)}$.
        If $(q'^{(b,1)},\dots,q'^{(b,\delta p^e(q_i))}) \neq (0,\dots,0)$, then
        \begin{equation*}
            \mu^{i,e}(q,p, \{[q + e_i \otimes \delta q^e, v]^{(i)}\} \times \{p+\delta^i\delta p^e(q_i)e_i\} ) = \mathcal{V}^{i,a,\delta p^e(q_i)}(\{v\})
        \end{equation*}
        for $v \in (\mathbb{N}^*)^{\delta p^e(q_i)}$.
        \item If $p^e(q_i) < 0$, let $(q'^j)_{j \in \mathbb{J}} \defeq q + \delta q^e$. If $(q'^{(a,1)},\dots,q'^{(a,-\delta p^e(q_i))}) = (0,\dots,0)$, then
        \begin{equation*}
            \mu^{i,e}(q,p, \{[v,q + e_i \otimes \delta q^e]^{(i)} + e_i \otimes (w e_{a,K})\} \times \{p+\delta^i\delta p^e(q_i)e_i\} ) = \mathcal{V}^{i,b,-\delta p^e(q_i)}(\{v\})\mathcal{V}^{i,a,1}(\{w\})
        \end{equation*}
        for $(v,w) \in \mathbb{N}^* \times (\mathbb{N}^*)^{-\delta p^e(q_i)}$.
        If $(q'^{(a,1)},\dots,q'^{(a,-\delta p^e(q_i))}) \neq (0,\dots,0)$, then
        \begin{equation*}
            \mu^{i,e}(q,p, \{[v,q + e_i \otimes \delta q^e]^{(i)}\} \times \{p+\delta^i\delta p^e(q_i)e_i\} ) = \mathcal{V}^{i,b,-\delta p^e(q_i)}(\{v\})
        \end{equation*}
        for $v \in (\mathbb{N}^*)^{-\delta p^e(q_i)}$.
    \end{itemize}

    For $(i,e) \in \mathbb{M}$, there exists a measurable intensity function $\tilde{\Lambda}^{i,e}:\mathcal{Q} \times \mathbb{R} \to [0, \infty)$ such that for $q=(q_i)_{i \in \stocks} \in \mathcal{Q}^{\stocks}$ and $y=(y^i)_{i \in \stocks}$,
    \begin{equation*}
        \Lambda^{i,e}(q,y) = \tilde{\Lambda}^{i,e}(q_i,y^i).
    \end{equation*}
    Furthermore, suppose that $\tilde{\Lambda}^{i,e}\mathds{1}_{(\mathcal{Q} \setminus \mathcal{Q}^{e,licit}) \times \mathbb{R}} \equiv 0$.
\end{assumption}

Under Assumption \ref{assumption:QR_main_assumption}, $(S-P, X)$ is Markovian with infinitesimal generator $\mathcal{L}$ defined below. For $\phi$ regular enough and $(y, q=(q_i^j)_{i \in \stocks, j \in \mathbb{J}}) \in \mathbb{R}^d \times \mathcal{Q}^{\stocks}$
\begin{equation*}
    \begin{split}
        \mathcal{L} \phi (y,q)
        &=\frac{1}{2} \tr \big(\Sigma \Sigma^T \nabla^2_y \phi(y, q)\big)\\
        &\phantom{=}+ \sum_{i \in \stocks} \sum_{\substack{e \in \mathbb{T}\\\delta p^e(q_i) = 0}}\big( \phi(y,  q + e_i \otimes \delta q^e) - \phi(y, q) \big) \tilde{\Lambda}^{i, e}(q_i, y^i)\\
        &\phantom{=}+ \sum_{i \in \stocks} \sum_{\substack{e \in \mathbb{T}\\\delta p^e(q_i) > 0}}\bigg(\int_{(\mathbb{N}^*)^{\delta p^e(q_i)}} \phi(y-\delta^i\delta p^e(q_i)e_i, [q + e_i \otimes \delta q^e, v]^{(i)})\mathcal{V}^{i,a,\delta p^e(q_i)}(\diff v) - \phi(y, q) \bigg) \tilde{\Lambda}^{i, e}(q_i, y^i)\\
        &\phantom{=}+ \sum_{i \in \stocks} \sum_{\substack{e \in \mathbb{T}\\\delta p^e(q_i) < 0}}\bigg(\int_{(\mathbb{N}^*)^{-\delta p^e(q_i)}} \phi(y-\delta^i\delta p^e(q_i)e_i, [v,q + e_i \otimes \delta q^e]^{(i)})\mathcal{V}^{i,b,-\delta p^e(q_i)}(\diff v) - \phi(y, q) \bigg) \tilde{\Lambda}^{i, e}(q_i, y^i).
    \end{split}
\end{equation*}

\subsubsection{Results}
\label{subsubsec:qr_results}
We start stating some assumptions that will ensure that finite volumes stay in the LOB, and that its dynamics follow the efficient price. They are similar to the assumptions needed for ergodicity of the queue-reactive model given by \textcite{huang2017ergodicity}.

\begin{assumption}
    \label{assumption:non_exploding_queue}
    There exist $Q \in \mathbb{N}^*$ and $\lambda \in (0,\infty)$ such that for $i \in \stocks$, $q=(q^j)_{j \in \mathbb{J}} \in \mathcal{Q}$, $j \in \mathbb{J}$ and $y \in \mathbb{R}$, if $q^j > Q$, then
    \begin{equation*}
         \sum_{\substack{e= (+,j',s,n) \in \mathbb{T}^+ \\ j' = j}}n\tilde{\Lambda}^{i,e}(q,y)
        +\sum_{\substack{e= (\text{modif},j',j'',n) \in \mathbb{T}^{modif} \\ j'' = j}}n\tilde{\Lambda}^{i,e}(q,y) \leqslant \lambda.
    \end{equation*}
\end{assumption}

\begin{assumption}
    \label{assumption:finite_expectation_size_orders}
    There exists $Q' \in \mathbb{N}^*$ and $\lambda \in (0,\infty]$ such that for all $i \in \stocks$, $j \in \mathbb{J}$, $q \in \mathcal{Q}$ and $y \in \mathbb{R}$,
    \begin{equation*}
        \sum_{\substack{e=(+,j',n) \in \mathbb{T}^+ \\ j'=j \\n > Q'}}n \tilde{\Lambda}^{i,e}(q,y)
        + \sum_{\substack{e=(\text{modif},j',j'',n)\in \mathbb{T}^{modif}\\j'=j \\n > Q'}}n \tilde{\Lambda}^{i,e}(q,y)
        \leqslant \lambda.
    \end{equation*}
\end{assumption}

\begin{assumption}
    \label{assumption:qr_local_boundedness}
    For $i \in \stocks$ and $e \in \mathbb{T}^+ \cup \mathbb{T}^{modif}$, $\sup_{q \in \mathcal{Q}} \tilde{\Lambda}^{i,e}(q, \cdot)$ is locally bounded. For $i \in \stocks$ and $e \in \mathbb{T}$, $\sup_{q \in \mathcal{Q}} \tilde{\Lambda}^{i,e}(q, \cdot) \mathds{1}_{\{\delta p^e(q) \neq 0\}}$ is locally bounded.
\end{assumption}

\begin{assumption}
    \label{assumption:qr_limits}
    We have the following bounds at $\pm \infty$
    \begin{equation*}
        \begin{split}
            \limsup_{y \to \infty} \inf_{q \in \mathcal{Q}}\Bigg[&\sum_{\substack{e=(+,j,s,n) \in \mathbb{T}^{+}\\s= a}} n \tilde{\Lambda}^{i, e}(q, y)
        +\sum_{\substack{e=(\text{modif},(s,l),j,n) \in \mathbb{T}^{modif}\\s=a\\ j - j^{a,best(q)} < 0}} n\tilde{\Lambda}^{i, e}(q, y)\\
        &+\sum_{\substack{e=(-,(s,l),n) \in \mathbb{T}^-\\s=b\\ (s,l)=j^{b,best}(q)}}n\tilde{\Lambda}^{i, e}(q, y)
        + \sum_{\substack{e=(\text{modif},(s,l),j,n) \in \mathbb{T}^{modif}\\s=b\\ (s,l)=j^{b,best}(q), j - j^{b,best}(q) < 0}}n\tilde{\Lambda}^{i, e}(q, y)\Bigg] < \infty
        \end{split}
    \end{equation*}
    and
    \begin{equation*}
        \begin{split}
            \limsup_{y \to -\infty} \inf_{q \in \mathcal{Q}}\Bigg[&\sum_{\substack{e=(+,j,s,n) \in \mathbb{T}^{+}\\s= b}} n \tilde{\Lambda}^{i, e}(q, y)
        +\sum_{\substack{e=(\text{modif},(s,l),j,n) \in \mathbb{T}^{modif}\\s=b\\ j - j^{b,best(q)}> 0}} n\tilde{\Lambda}^{i, e}(q, y)\\
        &+\sum_{\substack{e=(-,(s,l),n) \in \mathbb{T}^-\\s=a\\ (s,l)=j^{a,best}(q)}}n\tilde{\Lambda}^{i, e}(q, y)
        + \sum_{\substack{e=(\text{modif},(s,l),j,n) \in \mathbb{T}^{modif}\\s=a\\ (s,l)=j^{a,best}(q), j - j^{a,best}(q) > 0}}n\tilde{\Lambda}^{i, e}(q, y)\Bigg] < \infty.
        \end{split}
    \end{equation*}
    There exists a function $\xi:\mathbb{R} \to [0,\infty)$ such that $\lim_{y \to \infty} \xi(y) = \infty$ and, for $i \in \stocks$,
    \begin{align*}
        \sum_{\substack{e=(-,(s,l),n) \in \mathbb{T}^-\\s=a}}n\tilde{\Lambda}^{i, e}(q, y) & \geqslant y \xi(y),\\
        \sum_{\substack{e=(-,(s,l),n) \in \mathbb{T}^-\\s=b}}n\tilde{\Lambda}^{i, e}(q, y) & \geqslant -y \xi(-y).
    \end{align*}
\end{assumption}

Assumption \ref{assumption:non_exploding_queue} is there to ensure that queue sizes do not grow too fast: once their size reached a threshold, the intensity of arrival of limit orders remains bounded. Assumption \ref{assumption:finite_expectation_size_orders} also moderates the growth of queue sizes: it states that the arrival intensity of big orders remains bounded. Assumption \ref{assumption:qr_local_boundedness} is there to guarantee that, as the efficient price remains at a bounded distance from the reference price, the intensities of arrival of limit orders, modifications and reference price changes remain bounded. Assumption \ref{assumption:qr_limits} is the one making the connection between the efficient price and the behavior of the order flow. If the efficient price of an asset is very higher than the efficient price, there will be many cancellations on the ask side (which can also be interpreted as modification to a position beyond the observable LOB), since the pending limit orders are at a bad price for those who placed them. The fact that Assumption \ref{assumption:QR_main_assumption} guarantees that there is always liquidity on the ask side, imposing an arbitrarily large intensity of arrival of cancellations is non-contradictory. On the contrary, they tend to place fewer limit sell orders or modify them to a lower price. Market orders arrivals also increase since people can buy the asset at a price lower than its actual value. On the bid side, those who placed limit orders tend either to keep them where they are, getting a better price, or modify them to higher prices to get executed faster, still getting good value for their price.

Proposition \ref{prop:qr_lyapunov} below, proved in Section \ref{subsec:proofs_queue_reactive}, builds a suitable Lyapunov function for the queue-reactive model, that will ensure its stability.
\begin{proposition}
    \label{prop:qr_lyapunov}
    Under Assumptions \ref{assumption:QR_main_assumption}, \ref{assumption:non_exploding_queue}, \ref{assumption:finite_expectation_size_orders}, \ref{assumption:qr_local_boundedness} and \ref{assumption:qr_limits},
    there exist a $C^{2}$ even function $\psi:\mathbb{R} \to [0,\infty)$ having $+\infty$ as a limit at $\pm \infty$, non-decreasing on $[0,\infty)$, a constant $C \in (0, \infty)$, and a norm-like function $V:\mathbb{R}^d \times \mathcal{Q}^{\stocks} \to [0,\infty)$, $C^2$ in its first variable $y$, such that for $q =(q_i^j)_{i\in \stocks,j\in \mathbb{J}} \in \mathcal{Q}^{\stocks}$ and $y = (y^i)_{i \in \stocks}$,
    \begin{equation*}
        V(y,q) \geqslant \sum_{i \in \stocks} (y^i)^2 \psi(y^i) \text{ and } \mathcal{L}V(y,q) \leqslant C.
    \end{equation*}
\end{proposition}

By Proposition \ref{prop:qr_lyapunov} and \parencite[Theorem 2.1]{meyntweedieIII}, $(S-P, X)$ is non-explosive. The purpose of Assumption \ref{assumption:qr_local_boundedness_all} is to ensure the non-explosiveness of $N$.
 Proposition \ref{prop:non_explosiveness_lyapunov} then gives us Corollary \ref{corol:qr_non_explosive}.

\begin{assumption}
    \label{assumption:qr_local_boundedness_all}
    For $i \in \stocks$, $\sum_{e \in \mathbb{T}}\tilde{\Lambda}^{i,e}$ is locally bounded.
\end{assumption}

\begin{corollary}
    \label{corol:qr_non_explosive}
    Under Assumptions \ref{assumption:QR_main_assumption}, \ref{assumption:non_exploding_queue}, \ref{assumption:finite_expectation_size_orders}, \ref{assumption:qr_local_boundedness}, \ref{assumption:qr_limits} and \ref{assumption:qr_local_boundedness_all}, the Markov process $(S,P,X,N)$ is non-explosive.
\end{corollary}

As in the model of Section \ref{subsec:signal_modulated}, we have the consequences Proposition \ref{prop:lyapunov_convergence_proba_general} and its Corollary \ref{corol:skorokhod_convergence}, assuring that $P$ behaves like $S$ at the macroscopic scale.

\begin{proposition}
    \label{prop:qr_modulated_convergence_proba}
    Suppose that $S_0 - P_0$ is bounded. Under Assumptions \ref{assumption:QR_main_assumption}, \ref{assumption:non_exploding_queue}, \ref{assumption:finite_expectation_size_orders}, \ref{assumption:qr_local_boundedness} and \ref{assumption:qr_limits},, for all $T, \epsilon > 0$,
    \begin{equation*}
        \lim_{n\to \infty} \Proba[][\big| \tilde{S}^{(n)}-\tilde{P}^{(n)}\big|_{\infty, [0,T]} \geqslant \epsilon] = 0.
    \end{equation*}
\end{proposition}

\begin{corollary}
    Under Assumptions \ref{assumption:QR_main_assumption}, \ref{assumption:non_exploding_queue}, \ref{assumption:finite_expectation_size_orders}, \ref{assumption:qr_local_boundedness} and \ref{assumption:qr_limits}, $(\tilde{P}^{(n)})_{n \in \mathbb{N}^*}$ converges in distribution in $D([0, \infty), \mathbb{R}^d)$ to $\Sigma B$ where $B$ is a $d$-dimensional Brownian motion.
\end{corollary}

\section{Estimation procedure}
\label{sec:estimation}
In this section, we describe a procedure to estimate the intensities $\Lambda^k$ and the volatility $\Sigma$ in a set of parameters $\Theta$. We use a maximum likelihood approach. Specifically we build processes $(S^{\theta}, P, X, N)$ and probabilities $\Proba[\theta]$, all equivalent to some probability $\Proba$ in finite time such that under $\Proba[\theta]$, $(S^{\theta}, P, X, N)$ is a Markov process with infinitesimal generator $\hat{\mathcal{A}}_{\theta}$ defined by
\begin{equation}
    \label{eq:generator_theta}
    \begin{split}
        \hat{\mathcal{A}}_{\theta}\phi(y, p, x, n)
        &=\frac{1}{2} \tr \big(\Sigma_{\theta} \Sigma_{\theta}^T \nabla^2_y \phi(y, p, x, n)\big)\\
        &\phantom{=}+ \sum_{k \in \mathbb{M}}\bigg(\int_{\mathcal{X} \times \mathcal{P}} \phi(y, p', x', n+e_k)\mu^k(\diff x', \diff p') - \phi(y, p, x, n) \bigg) \Lambda_{\theta}^k(x, y-p)\\
        &\phantom{=}+
        \bigg(\int_{\mathcal{X} \times \mathcal{P}} \phi(y, p', x', n)\mu(\diff x', \diff p') - \phi(y, p, x, n) \bigg) \Lambda(x)
    \end{split}
\end{equation}
for $(y, p, x, n) \in \mathbb{R}^d \times \mathcal{P} \times \mathcal{X} \times \mathbb{N}^{\mathbb{M}}$ and $\phi: \mathbb{R}^d \times \mathcal{P} \times \mathcal{X} \times \mathbb{N}^{\mathbb{M}} \to \mathbb{R}$ bounded and measurable, two times differentiable with respect to its first variable with bounded second derivatives.

Proving the convergence of the maximum likelihood estimator of models with a hidden diffusion process has been show to be a hard task, and is often done on a case-by-case basis \parencite{kutoyants2019parameter,nadtochiy2024consistencymlepartiallyobserved,ekren2025consistencymlepartiallyobserved}. Here, we rather give a practical algorithm to compute the likelihood, and hint the convergence with numerical tests.

In Subsection \ref{subsec:likelihood}, we derive a likelihood function for our model and express with the help of PDEs. In Subsection \ref{subsec:estimator_implementation}, we describe an approximation procedure to compute the likelihood. In Subsection \ref{subsec:estimator_validation}, we discuss the convergence of our estimator on numerical examples.

\subsection{The likelihood}
\label{subsec:likelihood}

We start by an assumption that will hold for this whole section. It states that a jump of $N^k$ can occur (under an observed state $X_{t-}$) at $t$ for a model $\theta$ if and only if it can hold for all the other models. This assumption is crucial for the equivalence of the laws of each model under our observation.
\begin{assumption}
    \label{assumption:equivalence_intensities}
    The set $\mathbb{M}$ is finite.
    For $k \in \mathbb{M}$, there exists a measurable subset $\mathcal{X}^{k}_+$ of $\mathcal{X}$ such that for all $(\theta, x, y) \in \Theta \times \mathcal{X}^{k}_+ \times \mathbb{R}^d$, $\Lambda^k_{\theta}(x,y) > 0$ and for all $(\theta, x, y) \in \Theta \times (\mathcal{X} \setminus\mathcal{X}^{k}_+) \times \mathbb{R}^d$, $\Lambda^k_{\theta}(x,y) = 0$ 
\end{assumption}

The finiteness of $\mathbb{M}$ does not correspond exactly to the queue-reactive model of Subsection \ref{subsec:queue_reactive}. However, we can still estimate queue-reactive models with a finite number of order types, if for example we do not consider separately orders with a different volume.

Let $(\Omega, \mathcal{F}, \Proba)$ be a probability space supporting a d-dimensional Brownian motion $W$, a square integrable random variable $S_0$, and a Markov process $(X, P, N)$ with generator $\mathcal{A}_0$ defined by
\begin{equation*}
    \begin{split}
        \mathcal{A}_0 \phi(x, p, n)
        &= \sum_{k \in \mathbb{M}}\bigg(\int_{\mathcal{X} \times \mathcal{P}} \phi( p', x', n+e_k)\mu^k(x,p,\diff x', \diff p') - \phi(p, x, n) \bigg) \mathds{1}_{\mathcal{X}^k_+}(x) \\
        &\phantom{=}+
        \bigg(\int_{\mathcal{X} \times \mathcal{P}} \phi(p', x', n)\mu(x,p,\diff x', \diff p') - \phi(p, x, n) \bigg) \Lambda(x),
    \end{split}
\end{equation*}
for $(y,p,x,n) \in \mathbb{R}^d \times \mathcal{P} \times \mathcal{X} \times \mathbb{N}^{\mathbb{M}}$ and any measurable bounded function $\phi:\mathbb{R}^d \times \mathcal{P} \times \mathcal{X} \times \mathbb{N}^{\mathbb{M}}\to \mathbb{R}$. The random variables $W$, $S_0$ and  $(X, P, N)$ are supposed independent. For $\theta \in \Theta$, and $t > 0$ we define $S^{\theta}_t \defeq S_0 + \Sigma W_t$. We define the filtrations $(\mathcal{F}^W_t)_t$ and $(\mathcal{F}^{obs}_t)_t$ to be the completed natural filtrations of $W$ and $(X, P, N)$ respectively. In practice, we only observe events in $(\mathcal{F}^{obs}_t)_t$ since $W$ is hidden to us. Unless stated otherwise, terms referring to a filtration such as \enquote{martingale} or \enquote{adapted} will always refer to the filtration $(\mathcal{F}_t)_t$ defined by $\mathcal{F}_t \defeq \mathcal{F}_t^{W} \vee \mathcal{F}_t^{obs}$ fo all $t \geqslant 0$.

Assumption \ref{assumption:non-explosiveness} below assures that all the processes in consideration are well-defined at all times.
\begin{assumption}
    \label{assumption:non-explosiveness}
    The process $(X, P, N)$ is non-explosive under $P$. For all $\theta \in \Theta$, a Markov process with generator $\hat{\mathcal{A}}_{\theta}$ defined by Equation \eqref{eq:generator_theta} is non-explosive under any initial condition.
\end{assumption}

\begin{assumption}
    \label{assumption:locally_bounded_intensity}
    For all $\theta \in \Theta$, the function $\sum_{k \in \mathbb{M}} \Lambda^k_{\theta}$ is locally bounded.
\end{assumption}

For $\theta \in \Theta$, under Assumptions \ref{assumption:equivalence_intensities} and \ref{assumption:non-explosiveness}, we define the process
\begin{equation*}
    Z_{t}^{\theta} \defeq
    \exp\bigg(\sum_{k \in \mathbb{M}} \int_{(0,t]} \ln\big(\Lambda^k_{\theta}(X_{u-}, S^{\theta}_u - P_{u-})\big) \diff N^k_{u} + \sum_{k \in \mathbb{M}} \int_0^{t} \big(1-\Lambda_{\theta}^k(X_{u-}, S^{\theta}_u - P_{u-})\big)\mathds{1}_{\mathcal{X}^k_+}(X_{u-}) \diff u \bigg)
\end{equation*}
for $t \geqslant 0$. Being the exponential process of a local martingale, $(Z_t^{\theta})_{t\geqslant 0}$ is a positive local martingale (and therefore also a supermartingale). Proposition \ref{prop:martingale} below, proved in Section \ref{subsec:proof_martingale}, shows that it is actually a true martingale, thanks to the non-explosiveness property of the considered processes.

\begin{proposition}
    \label{prop:martingale}
    Under Assumptions \ref{assumption:equivalence_intensities}, \ref{assumption:non-explosiveness} and \ref{assumption:locally_bounded_intensity}, $(Z_t^{\theta})_{t\geqslant 0}$ is martingale for all $\theta \in \Theta$.
\end{proposition}

The proof of Proposition \ref{prop:martingale} leads immediately to Corollary \ref{corol:existence_ptheta} which states the existence of the equivalent measure $\Proba[\theta]$ under which $(S, P, X, N)$ has the desired distribution.

\begin{corollary}
    \label{corol:existence_ptheta}
    Under Assumptions \ref{assumption:equivalence_intensities}, \ref{assumption:non-explosiveness} and \ref{assumption:locally_bounded_intensity}, for all $\theta \in \Theta$ and $T > 0$, there exists a probability $\Proba[\theta]$ on $\mathcal{F}_T$ equivalent to $\Proba$ such that $\frac{\diff \Proba[\theta]}{\diff \Proba} = Z^{\theta}_T$, such that $(S^{\theta}, P, X, N)_{\cdot \wedge T}$ is a Markov process of infinitesimal generator $\hat{\mathcal{A}}_{\theta}$ under $\Proba[\theta]$ and
    \begin{equation*}
        \Proba \circ (S_0, P_0, X_0, N_0)^{-1} = \Proba[\theta] \circ (S_0, P_0, X_0, N_0)^{-1}.
    \end{equation*}
\end{corollary}

Under our observation, we have that for all $\theta \in \Theta$, the likelihood is given by
\begin{equation*}
    \frac{\Proba[\theta]|_{\mathcal{F}_T^{obs}}}{\Proba[]|_{\mathcal{F}_T^{obs}}}
     = \E[][Z_T^{\theta}| \mathcal{F}_T^{obs}].
\end{equation*}

Assumption \ref{assumption:intensities_regularity} below states mild regularity and growth conditions on the intensities for the likelihood $\E[][Z_T^{\theta}| \mathcal{F}_T^{obs}]$ to be expressed as the solution of a linear PDE jumping at the jump times of $(P, X)$. We define the class $\mathcal{C}_e$ of functions by 
\begin{equation*}
    \mathcal{C}_e \defeq \big\{f:\mathbb{R}^d \to \mathbb{R}: \exists (A, B) \in(0, \infty) \times [1, 2), \forall y \in \mathbb{R}^d |f(y)| \leqslant A e^{|y|^B} \big\}.
\end{equation*}
\begin{assumption}
    \label{assumption:intensities_regularity}
    (Regularity) For every compact $K$ of $\mathbb{R}^d$ and $x \in \mathcal{X}$ and $\theta \in \Theta$, there exists $\alpha \in (0, 1)$ such that $\sum_{k \in \mathbb{M}} \Lambda^k_{\theta}(x, \cdot)$ is $\alpha$-Hölder continuous on $K$.\\
    (Growth) For all $x \in \mathcal{X}$ and $\theta \in \Theta$, 
        $\sum_{k \in \mathbb{M}} \Lambda^k_{\theta}(x, \cdot)  \in \mathcal{C}_e$.
\end{assumption}

Lemma \ref{lem:parabolic_pde_existence} below, proved in Section \ref{subsec:proof_parabolic_pde_existence}, states the existence of regular solutions of a linear parabolic PDE of interest.
\begin{lemma}
    \label{lem:parabolic_pde_existence}
    Let $\Sigma$ be an invertible $d \times d$ matrix. Let $c:\mathbb{R}^d\to \mathbb{R}$ be a function such that for every compact $K$, there exists $\alpha \in (0, 1)$ such that $c$ is $\alpha$-Hölder continuous on $K$. Let $f \in \mathcal{C}_e$. Then, for all $T > 0$, there exists a unique function $u \in C([0, T] \times \mathbb{R}^d) \cap C^{1, 2}((0, T) \times \mathbb{R}^d) \cap \{g:[0,T]\times \mathbb{R}^d \to \mathbb{R}: \forall t \in [0, T], g(t,\cdot) \in \mathcal{C}_e\}$ such that
    \begin{equation}
        \label{eq:linear_pde}
        \partial_t u(t,y) + \frac{1}{2}\tr\big(\Sigma \Sigma^T \nabla^2_y u(t, y)\big) - c(y)u(t,y) = 0,\quad (t, y) \in (0,T) \times \mathbb{R}^d
    \end{equation}
    and
    \begin{equation}
        \label{eq:linear_pde_final}
        u(T, y) = f(y),\quad y \in \mathbb{R}^d.
    \end{equation}
\end{lemma}
\begin{remark}
    We can replace $c$ by $(c-\alpha)$ for a constant $\alpha \in \mathbb{R}$ by doing the transformation $\tilde{u}(t,\cdot) = e^{(T-t)\alpha}u(t,\cdot)$. 
\end{remark}

Theorem \ref{thm:PDE_characterization_lik}, proved in Section \ref{subsec:proof_PDE_characterization_lik}, gives a PDE characterization of the likelihood, following a standard verification approach. This characterization will allow us to compute it numerically.
\begin{theorem}
    \label{thm:PDE_characterization_lik}
    Suppose $N_0 = 0$ and $S_0$ follows the law $m_{S_0} \defeq \mathbb{P} \circ S_0^{-1}$. Let $T > 0$ and $\theta \in \Theta$. Let $0=t_0<t_1<\dots< t_M = T$ be the jump times of $N$ and $(X,P)$ on $[0, T]$. Let $z_0,\dots,z_M$ be the corresponding event types: $z_m = k \in \mathbb{M}$ if and only if there is a jump of $N^k$ at $t_m$ (which can be simultaneously a jump of $(X,P)$), and $z_m = \emptyset$ if $t_m$ is a jump time of $(X,P)$ only. Suppose Assumptions \ref{assumption:equivalence_intensities}, \ref{assumption:non-explosiveness}, \ref{assumption:locally_bounded_intensity} and \ref{assumption:intensities_regularity} are satisfied. Let $u:[0,T] \times \mathbb{R}^d \to \mathbb{R}$ be the unique function such that $u(T,\cdot) \equiv 1$ for all $m \in \{0,\dots,M-1\}$:
    \begin{itemize}
        \item $u|_{[t_m, t_{m+1}) \times \mathbb{R}^d} \in C([t_m, t_{m+1}) \times \mathbb{R}^d) \cap C^{1, 2}((t_m, t_{m+1}) \times \mathbb{R}^d) \cap \{g:[t_m, t_{m+1}]\times \mathbb{R}^d \to \mathbb{R}: \forall t \in [t_m, t_{m+1}], g(t,\cdot) \in \mathcal{C}_e\}$,
        \item $u|_{[t_m, t_{m+1}) \times \mathbb{R}^d}$ solves \eqref{eq:linear_pde} with $c = \sum_{k \in \mathbb{M}} (1-\Lambda_{\theta}^k(X_{t_m}, \cdot - P_{t_m}))$ on $(t_m, t_{m+1}) \times \mathbb{R}^d$,
        \item $u|_{[t_m, t_{m+1}) \times \mathbb{R}^d}$ can be extended to a continuous function $v$ on $[t_m, t_{m+1}] \times \mathbb{R}^d$ such that $v(t_{m+1}, \cdot) = \Lambda_{\theta}^{z_{m+1}}(\cdot - P_{t_m}, X_{t_m}) u(t_{m+1}, \cdot)$,
    \end{itemize}
    where we used the convention $\Lambda^{\emptyset}_{\theta} \equiv 1$. Then,
    \begin{equation*}
        \E[][Z_T^{\theta}| \mathcal{F}_T^{obs}]
        = \int_{\mathbb{R}^d}u(0, y) m_{S_0}(\diff y)\quad \mathbb{P}-a.s.
    \end{equation*}
\end{theorem}

\begin{remark}
    In Theorem \ref{thm:PDE_characterization_lik}, $M$, $(t_m)$ and $u$ are all $\mathcal{F}_T^{obs}$-measurable random variables.
\end{remark}

\subsection{Approximation of the likelihood}
\label{subsec:estimator_implementation}
Fix $\theta \in \Theta$. Thanks to Theorem \ref{thm:PDE_characterization_lik}, computing the likelihood boils down to solving a linear second-order PDE with jumps. We now describe our approximation procedure. We assume the following exponential form of the tracked intensities:
\begin{equation*}
    \Lambda_{\theta}^k(x,y) = \exp\bigg(\sum_{\alpha \in \mathbb{N}^d, |\alpha| \leqslant D} b^k_{\alpha}(x) y^{\alpha}\bigg),\quad (k,x,y) \in \mathbb{M} \times \mathcal{X} \times \mathbb{R}^d
\end{equation*}
for some $D > 0$
and the power series expansion of their sum:
\begin{equation*}
    \sum_{k \in \mathbb{M}}\big(-1+\Lambda_{\theta}^k(x,y)\big) = \sum_{\alpha \in \mathbb{N}^d} b_{\alpha}(x) y^{\alpha},\quad (x,y) \in \mathcal{X} \times \mathbb{R}^d,
\end{equation*}
with $\sum_{\alpha \in \mathbb{N}^d}^{\infty}\sum_{i=1}^d |b_{\alpha}(x) y^{\alpha}| < \infty$ for all $(x,y) \in \mathcal{X} \times \mathbb{R}^d$. We fix $b_{\alpha}^k \equiv 0$ for $|\alpha| > D$. Suppose that $u$ is a smooth solution of \eqref{eq:linear_pde} on a set $(t_1,t_2)\times \mathbb{R}^d$ with $c =  \sum_{k \in \mathbb{M}}\big(-1+\Lambda_{\theta}^k(x,\cdot)\big)$ for a fixed $x \in \mathcal{X}$, and that $u$ has the form
\begin{equation*}
    u(t,y) = \exp\bigg(-\sum_{a \in \mathbb{N}} a_{\alpha}(t) y^{\alpha} \bigg),\quad (t,y) \in (t_1,t_2)\times \mathbb{R}^d,
\end{equation*}
$(a_{\alpha})_{\alpha \in \mathbb{N}^d}$ being a family of differentiable functions on $(t_1, t_2)$ with $\sup_{t \in (t_1,t_2)}\sum_{a \in \mathbb{N}} |a_{\alpha}(t) y^{\alpha}| < \infty$ for all $(t,y) \in (t_1,t_2) \times \mathbb{R}^d$. Then, for all $\alpha \in \mathbb{N}^d$ and $t \in (t_1,t_2)$,
\begin{equation}
    \label{eq:system_power_series}
    \begin{split}
        0 &= -a'_{\alpha}(t) - \sum_{1 \leqslant i < j \leqslant d}(\Sigma \Sigma^T)_{ij}(\alpha_i+1)(\alpha_j + 1) a_{\alpha + e_i+e_j}(t) - \frac{1}{2}\sum_{i=1}^d (\Sigma \Sigma^T)_{ii} (\alpha_i+1)(\alpha_i + 2) a_{\alpha + 2e_i}(t)\\
        &\phantom{=}+ \frac{1}{2}\sum_{1 \leqslant i , j \leqslant d}(\Sigma \Sigma^T)_{ij}\sum_{\beta \leqslant \alpha}(\alpha_j-\beta_j+1)(\beta_i+1) a_{\beta+e_i}(t)a_{\alpha-\beta+e_j}(t) - b_{\alpha}(x).
        \end{split}
\end{equation}
In practice, we fix a degree $n_{deg} \in 2\mathbb{N}^*$, and we solve the system, \eqref{eq:system_power_series} for $\alpha \in \mathbb{N}^*$ such that $|a| \leqslant n_{deg}$, fixing $a_{\alpha} \equiv 0$ if $|\alpha| > n_{deg}$. \textcite{derchu2020bayesian} used a similar approach with $n_{deg} = 2$ to approximate the forward equation, characterizing the conditional law of $S_t$ with respect to $\mathcal{F}_t^{obs}$. This allowed for explicit computations and a highly tractable Gaussian distribution for $S_t$. However, $n_{deg} = 2$ happened to be too small for our estimation purpose.

Let $u$ be the function given by Theorem \ref{thm:PDE_characterization_lik}. For $t \in [0,T]$, define $\hat{u}(t,y) = u(t, \cdot + P_t)$. We approximate $\hat{u}$ by
\begin{equation*}
    \hat{u}(t,y) \approx \exp\bigg(-\sum_{|\alpha| \leqslant n_{deg}}a_{\alpha}(t)y^{\alpha}\bigg),\quad (t,y) \in [0,T] \times \mathbb{R}^d.
\end{equation*}
We approximate our likelihood by $\hat{u}(0, 0)$, assuming $S_0 = P_0$ almost surely ($m_{S_0} = \delta_{P_0}$ with the notations of Theorem \ref{thm:PDE_characterization_lik}). The details of our iterative procedure to compute $\hat{u}$ are given in Algorithm \ref{algo:hat_u}. We use an explicit scheme to solve the system of ODEs \eqref{eq:system_power_series}.

\begin{algorithm}
    \caption{Computation of $\hat{u}$. Notations from Theorem \ref{thm:PDE_characterization_lik}.}
    \label{algo:hat_u}
    \begin{algorithmic}[1]
        \STATE Initialize $ a_{\alpha}(T) \gets 0 $, $a_{\alpha} \in \mathbb{N}^d \setminus \{0\}$.
        \STATE Initialize $ a_{0}(T) \gets 1 $.
        \FOR{\( m \) in \( \{M-1, \dots, 0\} \)}
            \STATE $\hat{a}_{\alpha} \gets a_{\alpha}(t_{m+1})$, $\alpha \in \mathbb{N}^d$
            \FOR{$i \in \stocks$}
                \STATE $\Delta P \gets P_{t_{m+1}}^i - P_{t_{m+1}-}^i$.
                \STATE $(\hat{a}_{\alpha})_{\alpha \in \mathbb{N}^d} \gets (\sum_{n=0}^{\infty}\hat{a}_{\alpha+n e_i}\binom{\alpha_i+n}{\alpha_i}(\Delta P)^n)_{\alpha \in \mathbb{N}^d}$.
            \ENDFOR
            \STATE Solve the system \eqref{eq:system_power_series} on $[t_m, t_{m+1}) $ for $|\alpha| \leqslant n_{deg}$, assuming $a_{\alpha} \equiv 0$ if $|\alpha| > n_{deg}$,  with terminal condition:
            \begin{equation*}
                \lim_{t \to t_{m+1}-} a_{\alpha}(t) =
                \hat{a}_{\alpha} - b^{z_{t_{m+1}}}_{\alpha}(X_{t_m}),\quad \alpha \in \mathbb{N}^d, |\alpha| \leqslant n_{deg}
            \end{equation*}
            with the convention $b^{\emptyset}_{\alpha} \equiv 0.$
        \ENDFOR
    \end{algorithmic}
    \end{algorithm}

\subsection{Numerical validation}
\label{subsec:estimator_validation}

\subsubsection{Models considered}

We now apply our estimation method to simulated data to check its convergence. We test the two following models.

\textbf{Model 1:} $d = 1$ stock. Queue-reactive model of Section \ref{subsec:queue_reactive} with $K = 1$ tracked queue, à la \textcite{cont2012order}. We track the limit, cancel and market orders and for $e \in \{\text{limit}, \text{cancel}, \text{market}\} \times \{\text{bid}, \text{ask}\}$, the intensity of arrival of an order of type $e$ is given by
\begin{equation*}
    \tilde{\Lambda}^{e}(q^b, q^a, y) = \exp\big(
        \alpha^e_0 + \alpha^e_1 y + \alpha^e_2 q^b + \alpha^e_3 q^a
    \big),\quad y \in \mathbb{R}, q^b \in \mathbb{R}, q^a \in \mathbb{R},
\end{equation*}
$q^b$ and $q^a$ being the volumes pending on the bid and ask piles respectively. We impose the following bid-ask symmetry conditions, to lower the dimension of the state space:
\begin{equation*}
    \begin{cases}
      \alpha_0^{f, \text{bid}} = \alpha_0^{f, \text{ask}}\\
      \alpha_1^{f, \text{bid}} = -\alpha_1^{f, \text{ask}}\\
      \alpha_2^{f, \text{bid}} = \alpha_3^{f, \text{ask}}\\
      \alpha_3^{f, \text{bid}} = \alpha_2^{f, \text{ask}}
    \end{cases},\quad f \in \{\text{limit}, \text{cancel}, \text{market}\}.
\end{equation*}
The order sizes are unitary, expect market which have a probability $p=0.3$ to wipe the whole queue. When the ask (resp. bid) queue is depleted the reference price jumps up (resp. down) by one tick $\delta = 0.01$. The true values used for simulation are given in Table \ref{tab:model_1_true_values}. A sample simulated path of the efficient and the reference prices is drawn on Figure \ref{fig:1_stock_sample_path}. We observe that at the macroscopic scale, the efficient price and the reference price are not very different from each other.

\begin{table}
    \centering
    \begin{tabular}{|c|c|c|c|c|c|c|c|c|}
    \hline
    \textbf{Coefficient} & $\alpha_0^{\text{limit}}$ & $\alpha_1^{\text{limit}}$ & $\alpha_2^{\text{limit}}$ & $\alpha_3^{\text{limit}}$ & $\alpha_0^{\text{cancel}}$ & $\alpha_1^{\text{cancel}}$ & $\alpha_2^{\text{cancel}}$ & $\alpha_3^{\text{cancel}}$ \\
    \hline
    \textbf{True Value} & $\ln(2)$ & $2.5$ & $-1.0$ & $0.2$ & $\ln(1.9)$ & $-2.5$ & $1.0$ & $-0.2$\\
    \hline \multicolumn{6}{c}{} \\[-2.5ex] \cline{1-6}
    \textbf{Coefficient} & $\alpha_0^{\text{market}}$ & $\alpha_1^{\text{market}}$ & $\alpha_2^{\text{market}}$ & $\alpha_3^{\text{market}}$ & $\Sigma$ \\
    \cline{1-6}
    \textbf{True Value} & $\ln(0.1)$ & $-2.5$ & $-1.0$ & $0.2$ & $0.01$\\
    \cline{1-6}
    \end{tabular}
    \caption{True values used for simulation in Model 1, on the bid side.}
    \label{tab:model_1_true_values}
    \end{table}

\begin{figure}
    \centering
    \includegraphics[width=0.95\textwidth]{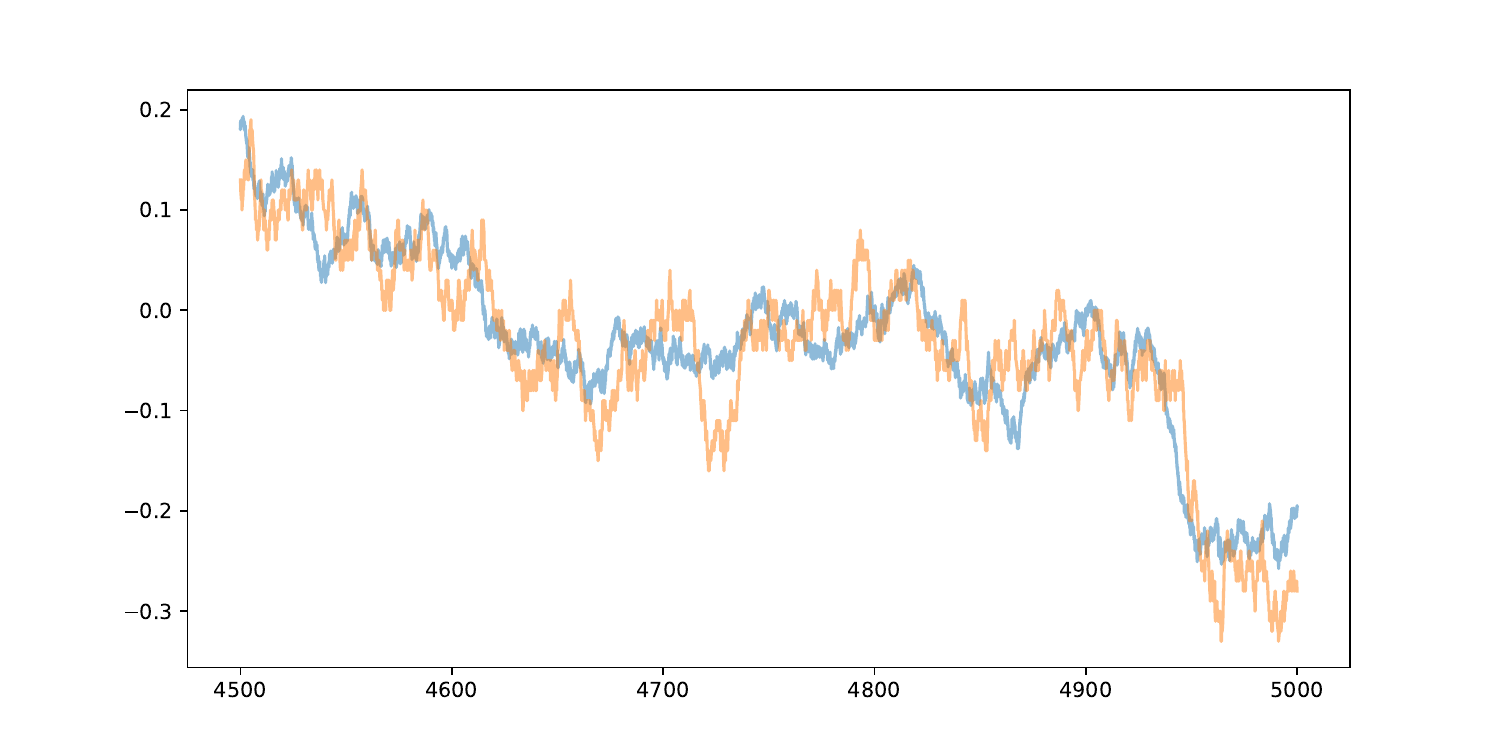}
    \includegraphics[width=0.95\textwidth]{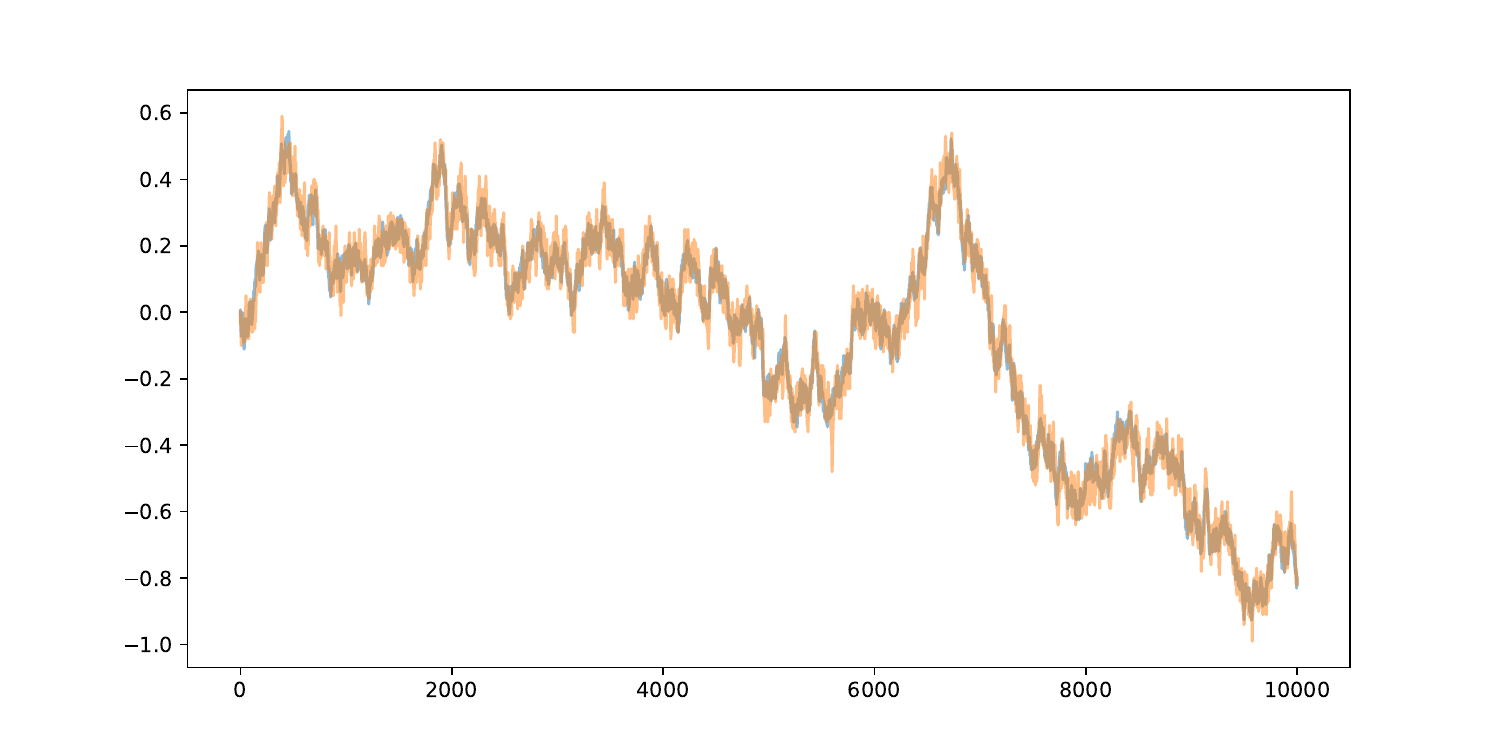}
    \caption{Simulated path of Model 1 with a horizon $T=10000$. In blue the efficient price $\Sigma W$, in orange the reference price $P$. Above: path on the time window [4500, 5000]. Below: whole path.}
    \label{fig:1_stock_sample_path}
\end{figure}

\textbf{Model 2:} $d = 2$ stocks. Signal-driven price model of Section \ref{subsec:signal_modulated}, with a state space $\mathcal{X} \defeq \mathbb{R}^2 \times \mathbb{R}^2$, the first two components being associated to Stock 1, the next ones to Stock 2. Denoting the state process by $X_t = (X^1_t, X^2_t)$, we have that $X^1$ and $X^2$ jump, independently at a Poisson rate of 1 and their value after the jump follows a standard normal distribution in two dimensions. Moreover, when there is a price jump of stock $i$, $X^i$ will immediately jump again following a standard normal distribution in two dimensions. A price jump of the stock $i$ and direction $e \in {-,+}$ will have an intensity function of arrival of
\begin{equation*}
    \tilde{\Lambda}^{i,e}(X^{i}_t, y) = \exp\big(
        \beta^{i,e}_0 + \beta^e_1 y + \beta^{i,e}_2 X^{i,1}_t + \beta^{i,e}_3 X^{i,2}_t
    \big),\quad y \in \mathbb{R}.
\end{equation*}
Again, we impose symmetry conditions
\begin{equation*}
    \begin{cases}
        \beta_0^{i, -} = \beta_0^{i, +}\\
        \beta_1^{i, -} = -\beta_1^{i, +}\\
        \beta_2^{i, -} = -\beta_2^{i, +}\\
        \beta_3^{i, -} = -\beta_3^{i, +}
      \end{cases},\quad i \in \{1,2\}.
\end{equation*}
Stock 1 has a tick size of $\delta^1 = 0.01$ and Stock 2 of $\delta^2=0.005$. The volatility of Stock $i$ is denoted by $\sigma_i$, and their correlation by $\rho$. The true values used for simulation are given in Table \ref{tab:model_2_true_values}.

\begin{table}
    \centering
    \begin{tabular}{|c|c|c|c|c|c|c|c|c|}
    \hline
    \textbf{Coefficient} & $\beta_0^{1,-}$ & $\beta_1^{1,-}$ & $\beta_2^{1,-}$ & $\beta_3^{1,-}$ & $\beta_0^{2,-}$ & $\beta_1^{2,-}$ & $\beta_2^{2,-}$ & $\beta_3^{2,-}$ \\
    \hline
    \textbf{True Value} & $\ln(2)$ & $-1$ & $-0.5$ & $1.0$ & $0.0$ & $-1.6$ & $2.0$ & $1.0$\\
    \hline \multicolumn{4}{c}{} \\[-2.5ex] \cline{1-4}
    \textbf{Coefficient} & $\sigma_1$ & $\sigma_2$ & $\rho$\\
    \cline{1-4}
    \textbf{True Value} & $0.01$ & 0.02 & 0.6\\
    \cline{1-4}
    \end{tabular}
    \caption{True values used for simulation in Model 2.}
    \label{tab:model_2_true_values}
    \end{table}

The exponential form of the intensities is chosen because it is computationally convenient and easily interpretable. It was also used in \parencite{muni2017modelling,sfendourakis_lob_2023}. The exponential part $e^{\alpha y}$ depending on the efficient price matches the development of Section \ref{sec:estimation} as it can be represented as a power series. The same choice was made by \textcite{derchu2020bayesian}.

\begin{remark}
    Both Models 1 and 2 do not satisfy the Assumptions for stability stated in Sections \ref{subsec:signal_modulated} and \ref{subsec:queue_reactive}. However, they do graphically seem stable, hinting that these assumptions may not be optimal.
\end{remark}

\subsubsection{Estimation results}

To estimate Models 1 and 2, we approximate the likelihood using Algorithm \ref{algo:hat_u}, with $n_{deg} = 10$ for Model 1 and $n_{deg} = 6$ for Model 2 (there is an ODE in $\frac{n_{deg}(n_{deg} + 1)}{2}$ dimensions to solve in the two-dimensional case, so we have to keep the degrees of the polynomial low for reasonable computation time). Between jump times, the system of ODEs \eqref{eq:system_power_series} is solved using an explicit Euler scheme.

The likelihood presents many local maxima. We use the CMAES optimization algorithm \parencite{hansen2016cma,nomura2024cmaes} which has been proven very robust in such cases. To increase robustness, since the algorithm is stochastic, we run it three times for Model 1 and six times for Model 2, and we keep the try with the highest likelihood.

For each horizon $T$, the estimation has been carried on 20 simulations. In Figures \ref{fig:1_stock_estimation_means} and \ref{fig:2_stocks_estimation_means}, we display the empirical average of each estimated coefficient. The empirical mean squared error is plotted on Figures \ref{fig:1_stock_estimation_mse} and \ref{fig:2_stocks_estimation_mse}. We see a clear decrease in the mean squared error as the horizon grows, with a slope of roughly -1 in log-log scale, hinting a standard deviation that decays in $\frac{1}{\sqrt{T}}$, as in most maximum likelihood frameworks.

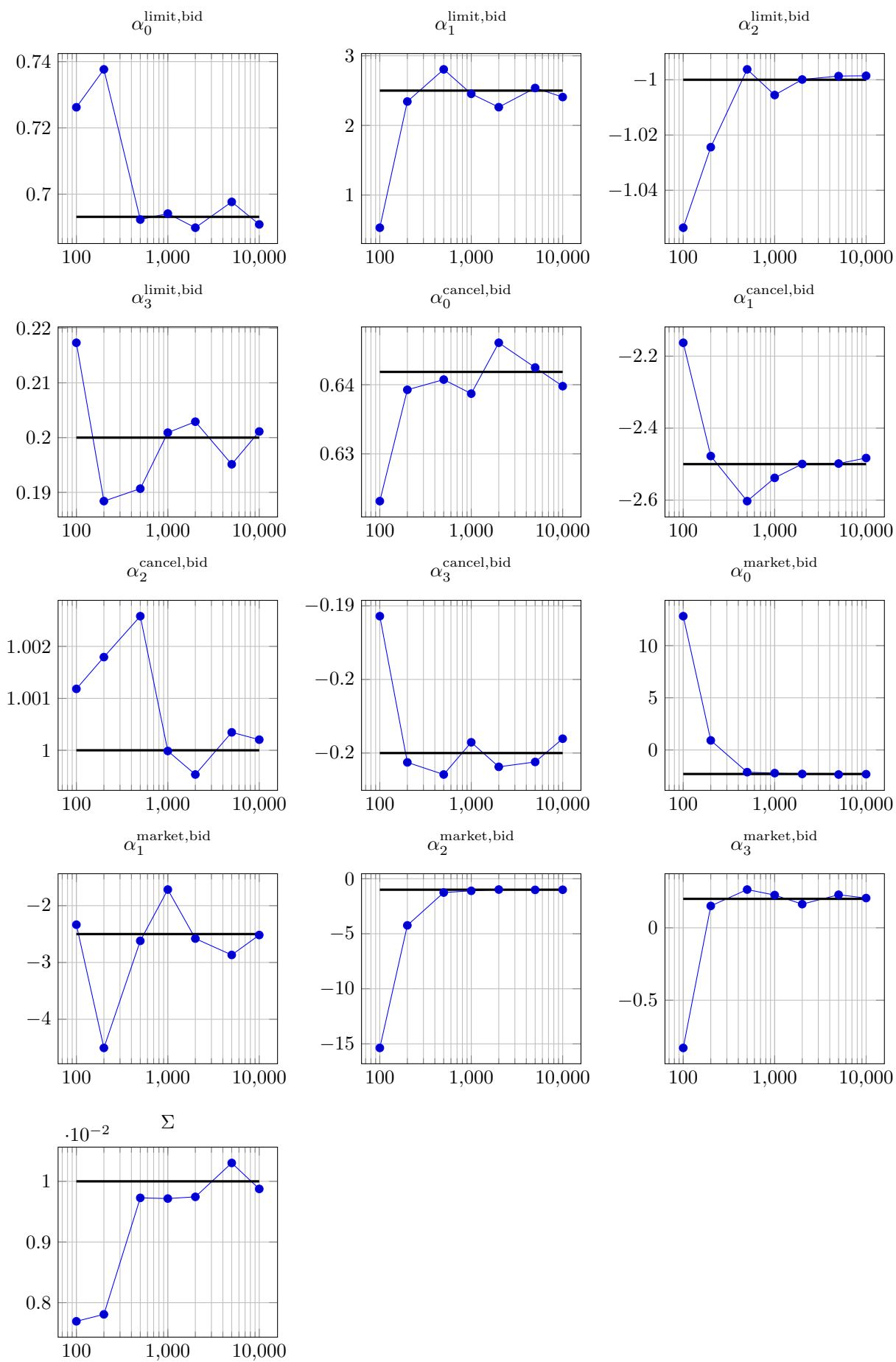
\begin{figure}
    \centering
    \begin{tikzpicture}
    \begin{groupplot}[group style={group size=3 by 5,horizontal sep=1.5cm,
                vertical sep=1.5cm},
                       height=5cm, width=5.5cm,
                       grid=both,xmode=log,
                       log ticks with fixed point,]
    \nextgroupplot[title={$\alpha_0^{\text{limit},\text{bid}}$}]
    \addplot coordinates {
    (100.0, 0.7262237980054647) 
    (200.0, 0.7376702221678665) 
    (500.0, 0.6923238817286568) 
    (1000.0, 0.6941319089222667) 
    (2000.0, 0.6898626984885579) 
    (5000.0, 0.6976662525748334) 
    (10000.0, 0.6908593473476236) 
    };
    \addplot[very thick, black] coordinates {(100.0, 0.6931471805599453) (10000.0, 0.6931471805599453)};
    \nextgroupplot[title={$\alpha_1^{\text{limit},\text{bid}}$}]
    \addplot coordinates {
    (100.0, 0.5307012209672185) 
    (200.0, 2.3432168637928283) 
    (500.0, 2.8048795084460947) 
    (1000.0, 2.453820779803185) 
    (2000.0, 2.262028765903975) 
    (5000.0, 2.537527330894066) 
    (10000.0, 2.407238686193273) 
    };
    \addplot[very thick, black] coordinates {(100.0, 2.5) (10000.0, 2.5)};
    \nextgroupplot[title={$\alpha_2^{\text{limit},\text{bid}}$}]
    \addplot coordinates {
    (100.0, -1.0535454383259937) 
    (200.0, -1.0243949378708472) 
    (500.0, -0.9962416497786274) 
    (1000.0, -1.005529037130483) 
    (2000.0, -0.9999105334541524) 
    (5000.0, -0.99866664527252) 
    (10000.0, -0.9985592321199424) 
    };
    \addplot[very thick, black] coordinates {(100.0, -1) (10000.0, -1)};
    \nextgroupplot[title={$\alpha_3^{\text{limit},\text{bid}}$}]
    \addplot coordinates {
    (100.0, 0.2173205193739429) 
    (200.0, 0.188403190667393) 
    (500.0, 0.1906789211409617) 
    (1000.0, 0.2009066172477907) 
    (2000.0, 0.2029072960959222) 
    (5000.0, 0.195128554125783) 
    (10000.0, 0.2011168976255322) 
    };
    \addplot[very thick, black] coordinates {(100.0, 0.2) (10000.0, 0.2)};
    \nextgroupplot[title={$\alpha_0^{\text{cancel},\text{bid}}$}]
    \addplot coordinates {
    (100.0, 0.623166850980584) 
    (200.0, 0.6392653358863988) 
    (500.0, 0.6407396156160206) 
    (1000.0, 0.6387099467131295) 
    (2000.0, 0.6460673124873909) 
    (5000.0, 0.6424847724452751) 
    (10000.0, 0.6397921226716926) 
    };
    \addplot[very thick, black] coordinates {(100.0, 0.6418538861723947) (10000.0, 0.6418538861723947)};
    \nextgroupplot[title={$\alpha_1^{\text{cancel},\text{bid}}$}]
    \addplot coordinates {
    (100.0, -2.162860751533624) 
    (200.0, -2.4775330118605683) 
    (500.0, -2.6028173887161903) 
    (1000.0, -2.5382014826175427) 
    (2000.0, -2.499630657258341) 
    (5000.0, -2.4984945335952693) 
    (10000.0, -2.4830734685528) 
    };
    \addplot[very thick, black] coordinates {(100.0, -2.5) (10000.0, -2.5)};
    \nextgroupplot[title={$\alpha_2^{\text{cancel},\text{bid}}$}, yticklabel style={/pgf/number format/fixed, /pgf/number format/precision=4}]
    \addplot coordinates {
    (100.0, 1.0011822434886504) 
    (200.0, 1.0017941775152015) 
    (500.0, 1.0025818465371756) 
    (1000.0, 0.999986301299048) 
    (2000.0, 0.9995337173814868) 
    (5000.0, 1.0003454975790678) 
    (10000.0, 1.0002044532103498) 
    };
    \addplot[very thick, black] coordinates {(100.0, 1.0) (10000.0, 1.0)};
    \nextgroupplot[title={$\alpha_3^{\text{cancel},\text{bid}}$}]
    \addplot coordinates {
    (100.0, -0.1907009072996558) 
    (200.0, -0.2006337621213236) 
    (500.0, -0.2014584861188011) 
    (1000.0, -0.199270541710358) 
    (2000.0, -0.2009365826934617) 
    (5000.0, -0.2006054262227182) 
    (10000.0, -0.199025378770621) 
    };
    \addplot[very thick, black] coordinates {(100.0, -0.2) (10000.0, -0.2)};
    \nextgroupplot[title={$\alpha_0^{\text{market},\text{bid}}$}]
    \addplot coordinates {
    (100.0, 12.82218520837174) 
    (200.0, 0.9155515732187935) 
    (500.0, -2.1280407697559065) 
    (1000.0, -2.2183590871711347) 
    (2000.0, -2.300011839730033) 
    (5000.0, -2.353579921972025) 
    (10000.0, -2.321522490938867) 
    };
    \addplot[very thick, black] coordinates {(100.0, -2.3025850929940455) (10000.0, -2.3025850929940455)};
    \nextgroupplot[title={$\alpha_1^{\text{market},\text{bid}}$}]
    \addplot coordinates {
    (100.0, -2.334817368568466) 
    (200.0, -4.504285601369446) 
    (500.0, -2.619373161222433) 
    (1000.0, -1.716382458417376) 
    (2000.0, -2.5793664684202104) 
    (5000.0, -2.8674126578743686) 
    (10000.0, -2.516827731369304) 
    };
    \addplot[very thick, black] coordinates {(100.0, -2.5) (10000.0, -2.5)};
    \nextgroupplot[title={$\alpha_2^{\text{market},\text{bid}}$}]
    \addplot coordinates {
    (100.0, -15.369263938566684) 
    (200.0, -4.238454298338952) 
    (500.0, -1.2581161459047103) 
    (1000.0, -1.101126090966838) 
    (2000.0, -0.9859838309261608) 
    (5000.0, -1.0147284271357897) 
    (10000.0, -1.0024209418446188) 
    };
    \addplot[very thick, black] coordinates {(100.0, -1) (10000.0, -1)};
    \nextgroupplot[title={$\alpha_3^{\text{market},\text{bid}}$}]
    \addplot coordinates {
    (100.0, -0.8291728898495672) 
    (200.0, 0.1510457990761948) 
    (500.0, 0.2640233782789871) 
    (1000.0, 0.2260635387718753) 
    (2000.0, 0.1634430533517138) 
    (5000.0, 0.2282155708937723) 
    (10000.0, 0.2053895613093069) 
    };
    \addplot[very thick, black] coordinates {(100.0, 0.2) (10000.0, 0.2)};
    \nextgroupplot[title={$\Sigma$}]
    \addplot coordinates {
    (100.0, 0.0076934000740094) 
    (200.0, 0.0078078424188686) 
    (500.0, 0.0097275891514378) 
    (1000.0, 0.0097162023929248) 
    (2000.0, 0.009744246664904) 
    (5000.0, 0.0103028897071818) 
    (10000.0, 0.0098750142659512) 
    };
    \addplot[very thick, black] coordinates {(100.0, 0.01) (10000.0, 0.01)};
    \end{groupplot} \end{tikzpicture}
    \caption{Empirical mean of the estimated coefficients of Model 1 with respect to the simulation horizon. The true value is the black horizontal line. The $x$-axis is in log-scale.}
    \label{fig:1_stock_estimation_means}
\end{figure}

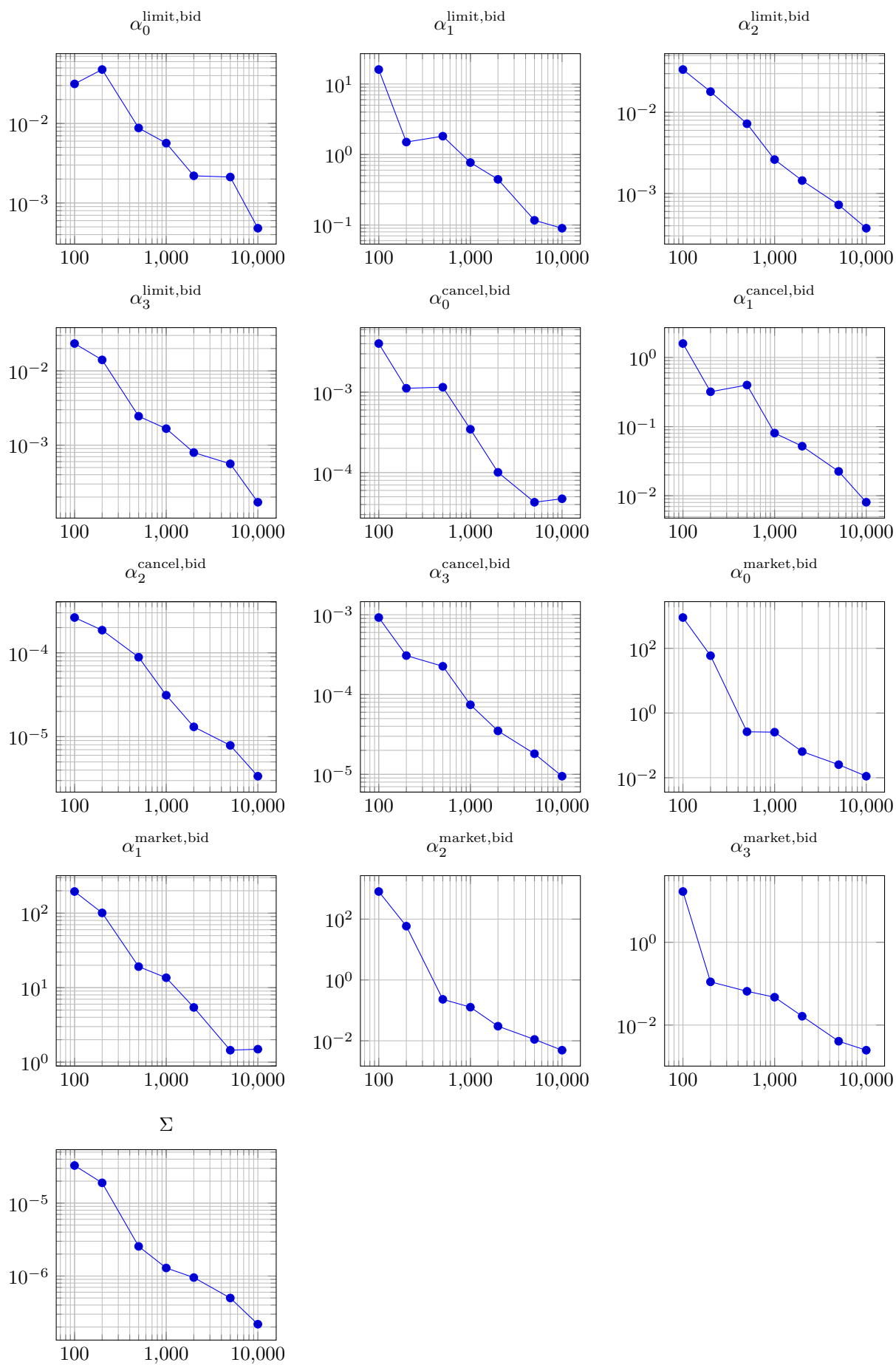
\begin{figure}
    \centering
    \begin{tikzpicture}
    \begin{groupplot}[group style={group size=3 by 5,horizontal sep=1.5cm,
                vertical sep=1.5cm},
                       height=5cm, width=5.5cm,
                       grid=both,xmode=log, ymode=log,
                       log ticks with fixed point,
        yticklabel={
            \pgfmathparse{int(round(\tick / ln(10)))}
            $10^{\pgfmathresult}$}
        ]
    \nextgroupplot[title={$\alpha_0^{\text{limit},\text{bid}}$},]
    \addplot coordinates {
    (100.0, 0.03146986434898492) 
    (200.0, 0.04773554608540021) 
    (500.0, 0.00877561037487799) 
    (1000.0, 0.005659749277006536) 
    (2000.0, 0.0021954142226528993) 
    (5000.0, 0.0021205890508302467) 
    (10000.0, 0.000481013570965259) 
    };
    \nextgroupplot[title={$\alpha_1^{\text{limit},\text{bid}}$}]
    \addplot coordinates {
    (100.0, 16.050766645566203) 
    (200.0, 1.5014973949186352) 
    (500.0, 1.8188317284572777) 
    (1000.0, 0.7688760523689351) 
    (2000.0, 0.44349428452362505) 
    (5000.0, 0.11658751644962151) 
    (10000.0, 0.08999798436092654) 
    };
    \nextgroupplot[title={$\alpha_2^{\text{limit},\text{bid}}$}]
    \addplot coordinates {
    (100.0, 0.03370134346590898) 
    (200.0, 0.017993518535729666) 
    (500.0, 0.007236967112311518) 
    (1000.0, 0.002615642409610791) 
    (2000.0, 0.0014436243507225024) 
    (5000.0, 0.0007242471076125116) 
    (10000.0, 0.00037304743563622414) 
    };
    \nextgroupplot[title={$\alpha_3^{\text{limit},\text{bid}}$}]
    \addplot coordinates {
    (100.0, 0.02335971348157384) 
    (200.0, 0.014060310015246553) 
    (500.0, 0.002450817216065685) 
    (1000.0, 0.0016685690767684383) 
    (2000.0, 0.0007947926246907667) 
    (5000.0, 0.0005605851229289066) 
    (10000.0, 0.0001704636670076447) 
    };
    \nextgroupplot[title={$\alpha_0^{\text{cancel},\text{bid}}$}]
    \addplot coordinates {
    (100.0, 0.004023551991898234) 
    (200.0, 0.0011170464911084694) 
    (500.0, 0.0011517130517324163) 
    (1000.0, 0.0003454817738748967) 
    (2000.0, 0.00010064832857413267) 
    (5000.0, 4.274487408766998e-05) 
    (10000.0, 4.721983610618981e-05) 
    };
    \nextgroupplot[title={$\alpha_1^{\text{cancel},\text{bid}}$}]
    \addplot coordinates {
    (100.0, 1.5904335762049961) 
    (200.0, 0.31931414991510193) 
    (500.0, 0.39895302375549213) 
    (1000.0, 0.08061038452791801) 
    (2000.0, 0.052039000058261724) 
    (5000.0, 0.022502363939549384) 
    (10000.0, 0.008064739608219693) 
    };
    \nextgroupplot[title={$\alpha_2^{\text{cancel},\text{bid}}$}]
    \addplot coordinates {
    (100.0, 0.0002633038271841697) 
    (200.0, 0.00018613816664217694) 
    (500.0, 8.864891298407199e-05) 
    (1000.0, 3.122299497969919e-05) 
    (2000.0, 1.309533450939061e-05) 
    (5000.0, 7.872291155897423e-06) 
    (10000.0, 3.3772437423887073e-06) 
    };
    \nextgroupplot[title={$\alpha_3^{\text{cancel},\text{bid}}$}]
    \addplot coordinates {
    (100.0, 0.0009235406849150732) 
    (200.0, 0.0003077538751345874) 
    (500.0, 0.00022638732406762515) 
    (1000.0, 7.444322102348983e-05) 
    (2000.0, 3.5011846076271375e-05) 
    (5000.0, 1.8116098704863104e-05) 
    (10000.0, 9.47142331383032e-06) 
    };
    \nextgroupplot[title={$\alpha_0^{\text{market},\text{bid}}$}]
    \addplot coordinates {
    (100.0, 896.2397536621813) 
    (200.0, 59.26883027827327) 
    (500.0, 0.26308742998017587) 
    (1000.0, 0.25624109770753833) 
    (2000.0, 0.06432064770560712) 
    (5000.0, 0.025332141658727943) 
    (10000.0, 0.011133114688857242) 
    };
    \nextgroupplot[title={$\alpha_1^{\text{market},\text{bid}}$}]
    \addplot coordinates {
    (100.0, 195.24840105905145) 
    (200.0, 100.83425032252505) 
    (500.0, 19.15152697886192) 
    (1000.0, 13.58959407607494) 
    (2000.0, 5.429108234290634) 
    (5000.0, 1.446061171575592) 
    (10000.0, 1.4932974922499709) 
    };
    \nextgroupplot[title={$\alpha_2^{\text{market},\text{bid}}$}]
    \addplot coordinates {
    (100.0, 806.7130190244856) 
    (200.0, 58.597503539221755) 
    (500.0, 0.23055357098251672) 
    (1000.0, 0.1274525104960458) 
    (2000.0, 0.030205060748339112) 
    (5000.0, 0.011092885155572411) 
    (10000.0, 0.004902751328181655) 
    };
    \nextgroupplot[title={$\alpha_3^{\text{market},\text{bid}}$}]
    \addplot coordinates {
    (100.0, 16.867610254638755) 
    (200.0, 0.11085426311599789) 
    (500.0, 0.06552978876828373) 
    (1000.0, 0.047137457417079386) 
    (2000.0, 0.016357490979499414) 
    (5000.0, 0.004067572105607004) 
    (10000.0, 0.0024679645528616224) 
    };
    \nextgroupplot[title={$\Sigma$}]
    \addplot coordinates {
    (100.0, 3.274647905870392e-05) 
    (200.0, 1.895478929150195e-05) 
    (500.0, 2.551332143924983e-06) 
    (1000.0, 1.2896207667926327e-06) 
    (2000.0, 9.536567824187423e-07) 
    (5000.0, 4.996121903220649e-07) 
    (10000.0, 2.1841639622072307e-07)
    };
    \end{groupplot} \end{tikzpicture}
    \caption{Empirical mean square error of the estimated coefficients of Model 1 with respect to the simulation horizon. In log-log scale.}
    \label{fig:1_stock_estimation_mse}
\end{figure}

\begin{figure}
    \centering
    \begin{tikzpicture}
    \begin{groupplot}[group style={group size=3 by 4,horizontal sep=1.5cm,
                vertical sep=2cm},
                       height=5cm, width=5.5cm,
                       grid=both,xmode=log,
                       log ticks with fixed point,]
    \nextgroupplot[title={$\beta_0^{1,-}$},yticklabel style={/pgf/number format/fixed, /pgf/number format/precision=3}]
    \addplot coordinates {
    (200.0, 0.6882552369499029) 
    (500.0, 0.6961478657732773) 
    (1000.0, 0.6917241750923242) 
    (2000.0, 0.6929684019273131) 
    };
    \addplot[very thick, black] coordinates {(100.0, 0.6931471805599453) (2000.0, 0.6931471805599453)};
    \nextgroupplot[title={$\beta_1^{1,-}$}]
    \addplot coordinates {
    (200.0, -1.0829855331857865) 
    (500.0, -1.0326140291015415) 
    (1000.0, -0.9822187479946076) 
    (2000.0, -1.0102494760416296) 
    };
    \addplot[very thick, black] coordinates {(100.0, -1) (2000.0, -1)};
    \nextgroupplot[title={$\beta_2^{1,-}$}, yticklabel style={/pgf/number format/fixed, /pgf/number format/precision=4}]
    \addplot coordinates {
    (200.0, -0.4935317631254633) 
    (500.0, -0.4970570440820004) 
    (1000.0, -0.503193128558894) 
    (2000.0, -0.4966955561099401) 
    };
    \addplot[very thick, black] coordinates {(100.0, -0.5) (2000.0, -0.5)};
    \nextgroupplot[title={$\beta_3^{1,-}$}, yticklabel style={/pgf/number format/fixed, /pgf/number format/precision=4}]
    \addplot coordinates {
    (200.0, 0.9911470057222092) 
    (500.0, 0.9931438305003576) 
    (1000.0, 0.9974057213682818) 
    (2000.0, 1.0022871107493392) 
    };
    \addplot[very thick, black] coordinates {(100.0, 1) (2000.0, 1)};
    \nextgroupplot[title={$\beta_0^{2,-}$}]
    \addplot coordinates {
    (200.0, 0.0006664783020552) 
    (500.0, 0.0083923244801105) 
    (1000.0, 0.0044412608141231) 
    (2000.0, -0.0043957055932243) 
    };
    \addplot[very thick, black] coordinates {(100.0, 0.0) (2000.0, 0.0)};
    \nextgroupplot[title={$\beta_1^{2,-}$}]
    \addplot coordinates {
    (200.0, -1.8481777258601184) 
    (500.0, -1.7523170033526874) 
    (1000.0, -1.6700701207774895) 
    (2000.0, -1.6523967645021655) 
    };
    \addplot[very thick, black] coordinates {(100.0, -1.6) (2000.0, -1.6)};
    \nextgroupplot[title={$\beta_2^{2,-}$}, yticklabel style={/pgf/number format/fixed, /pgf/number format/precision=3}]
    \addplot coordinates {
    (200.0, 1.987055338884296) 
    (500.0, 1.996526635591304) 
    (1000.0, 1.992843849149692) 
    (2000.0, 2.0015985016680133) 
    };
    \addplot[very thick, black] coordinates {(100.0, 2.0) (2000.0, 2.0)};
    \nextgroupplot[title={$\beta_3^{2,-}$}, yticklabel style={/pgf/number format/fixed, /pgf/number format/precision=3}]
    \addplot coordinates {
    (200.0, 1.0132856738235154) 
    (500.0, 0.9959493032886004) 
    (1000.0, 0.9984924982439196) 
    (2000.0, 1.0013703516413257) 
    };
    \addplot[very thick, black] coordinates {(100.0, 1.0) (2000.0, 1.0)};
    \nextgroupplot[title={$\sigma_1$}]
    \addplot coordinates {
    (200.0, 0.007295580309384) 
    (500.0, 0.0099231937538615) 
    (1000.0, 0.007833954597252) 
    (2000.0, 0.0091351797970221) 
    };
    \addplot[very thick, black] coordinates {(100.0, 0.01) (2000.0, 0.01)};
    \nextgroupplot[title={$\sigma_2$}]
    \addplot coordinates {
    (200.0, 0.0244844301307672) 
    (500.0, 0.0223715261800693) 
    (1000.0, 0.0190566142392045) 
    (2000.0, 0.0198516339921286) 
    };
    \addplot[very thick, black] coordinates {(100.0, 0.02) (2000.0, 0.02)};
    \nextgroupplot[title={$\rho$}]
    \addplot coordinates {
    (200.0, 0.686573207199462) 
    (500.0, 0.5805092034215853) 
    (1000.0, 0.7268345580997562) 
    (2000.0, 0.629425009617877) 
    };
    \addplot[very thick, black] coordinates {(100.0, 0.6) (2000.0, 0.6)};
    \end{groupplot} \end{tikzpicture}
    \caption{Empirical mean of the estimated coefficients of Model 2 with respect to the simulation horizon. The true value is the black horizontal line. The $x$-axis is in log-scale.}
    \label{fig:2_stocks_estimation_means}
\end{figure}
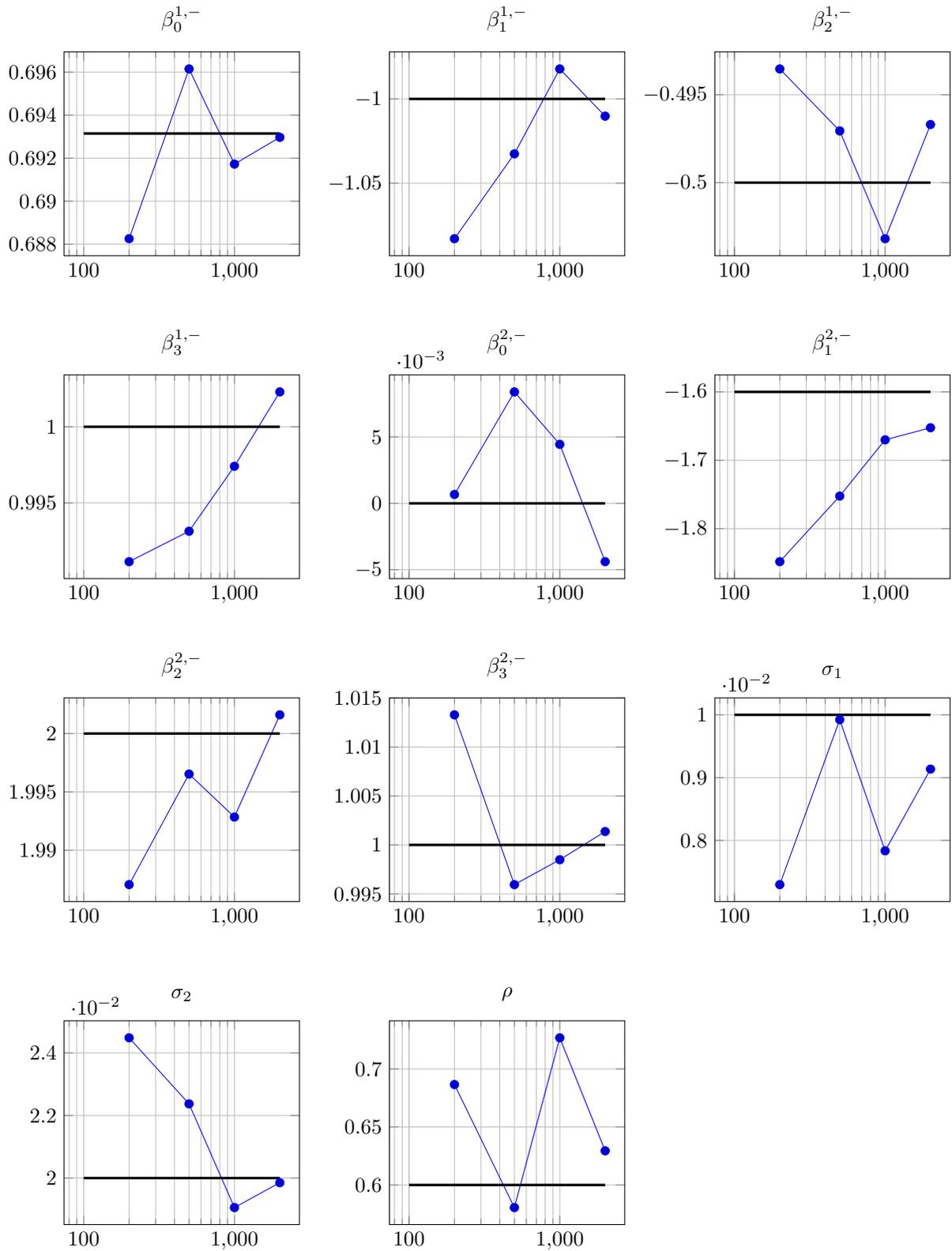

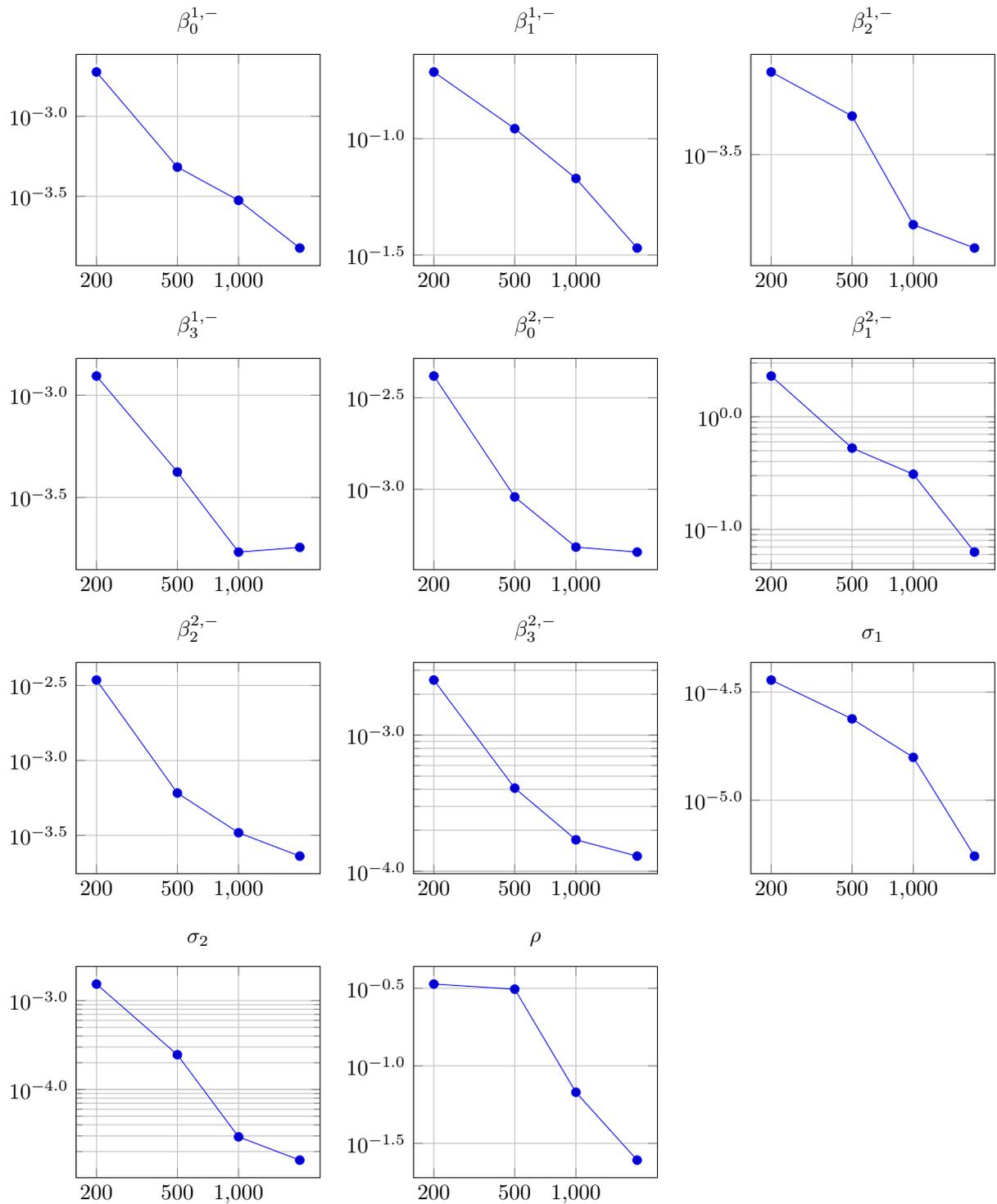
\begin{figure}
    \centering
    \begin{tikzpicture}
    \begin{groupplot}[group style={group size=3 by 4,horizontal sep=1.5cm,
                vertical sep=1.5cm},
                       height=5cm, width=5.5cm,
                       grid=both,xmode=log, ymode=log,
                       log ticks with fixed point,
                       yticklabel={
                           \pgfmathparse{round(\tick / ln(10) /0.1)/10}
                           $10^{\pgfmathresult}$},
                        xtick={200,500, 1000}]
    \nextgroupplot[title={$\beta_0^{1,-}$}]
    \addplot coordinates {
    (200.0, 0.0018924772218656444) 
    (500.0, 0.0004812641438969154) 
    (1000.0, 0.00029804514427985974) 
    (2000.0, 0.00015012763639198184) 
    };
    \nextgroupplot[title={$\beta_1^{1,-}$}]
    \addplot coordinates {
    (200.0, 0.19340304372769418) 
    (500.0, 0.11040137166745222) 
    (1000.0, 0.06743652142975301) 
    (2000.0, 0.033880365985427284) 
    };
    \nextgroupplot[title={$\beta_2^{1,-}$}]
    \addplot coordinates {
    (200.0, 0.0007294404876735305) 
    (500.0, 0.0004669375145732836) 
    (1000.0, 0.00015558120520271857) 
    (2000.0, 0.00012283958979535055) 
    };
    \nextgroupplot[title={$\beta_3^{1,-}$}]
    \addplot coordinates {
    (200.0, 0.0012411225378593023) 
    (500.0, 0.00042099142292727296) 
    (1000.0, 0.00017090710901080386) 
    (2000.0, 0.0001801894985757638) 
    };
    \nextgroupplot[title={$\beta_0^{2,-}$}]
    \addplot coordinates {
    (200.0, 0.004157911460731191) 
    (500.0, 0.0009089613468907792) 
    (1000.0, 0.00048346479429632236) 
    (2000.0, 0.00045424482595130336) 
    };
    \nextgroupplot[title={$\beta_1^{2,-}$}]
    \addplot coordinates {
    (200.0, 2.2995566483864436) 
    (500.0, 0.5275603748135516) 
    (1000.0, 0.30926994116530176) 
    (2000.0, 0.06299306863863217) 
    };
    \nextgroupplot[title={$\beta_2^{2,-}$}]
    \addplot coordinates {
    (200.0, 0.0034390650126892614) 
    (500.0, 0.0006058416158862767) 
    (1000.0, 0.000329159750329105) 
    (2000.0, 0.0002299314331834753) 
    };
    \nextgroupplot[title={$\beta_3^{2,-}$}]
    \addplot coordinates {
    (200.0, 0.0025422743635404948) 
    (500.0, 0.0004082300517970305) 
    (1000.0, 0.00017024835874145108) 
    (2000.0, 0.0001291403565998335) 
    };
    \nextgroupplot[title={$\sigma_1$}]
    \addplot coordinates {
    (200.0, 3.5945124326663035e-05) 
    (500.0, 2.377042106910968e-05) 
    (1000.0, 1.5787331035420633e-05) 
    (2000.0, 5.515339865544046e-06) 
    };
    \nextgroupplot[title={$\sigma_2$}]
    \addplot coordinates {
    (200.0, 0.001534367531038702) 
    (500.0, 0.0002458901941463744) 
    (1000.0, 2.9290571990153964e-05) 
    (2000.0, 1.602678138842647e-05) 
    };
    \nextgroupplot[title={$\rho$}]
    \addplot coordinates {
    (200.0, 0.33672652393399205) 
    (500.0, 0.31216512854912987) 
    (1000.0, 0.06755105795213334) 
    (2000.0, 0.024668122692650975) 
    };
    \end{groupplot} \end{tikzpicture}
    \caption{Empirical mean square error of the estimated coefficients of Model 2 with respect to the simulation horizon. In log-log scale.}
    \label{fig:2_stocks_estimation_mse}
\end{figure}

\section{Illustration on real data}
\label{sec:real_data}

Here, we apply our estimation method in a price model driven by the efficient price and the volume imbalance on the LOB of two stocks: BNP Paribas and Société Générale, traded at Euronext Paris in February 2022. We do one estimation per trading day, from 10 a.m. to 3 p.m.

Precisely, we record upwards (+) and downwards (-) jumps of the reference price (defined as in the queue-reactive model), and use as a signal the volume imbalance $X_t^i$ at best:
\begin{equation*}
    X_t^i = \frac{q^{i,b}_t - q^{i,a}_t}{q^{i,b}_t + q^{i,a}_t},\quad i \in \{BNP, SOGN\},
\end{equation*}
where $q^{i,b}$ is the volume pending at the first non-empty pile on the bid side of the LOB of Stock $i$, and $q^{i,a}$ is the equivalent on the ask side. It is known to be a reliable predictor of price moves \parencite{muni2017modelling,lehalle2021optimal,sfendourakis_lob_2023,pulido2023understanding}. We model the intensities of price jumps as
\begin{equation*}
    \tilde{\Lambda}^{i, e}(X_t^i, y)
    = \exp\big(\alpha_{intercept}^{i,e} + \alpha_{efficient}^{i,e} y + \alpha_{imbalance}^{i,e} X_t^{i} \big),\quad i \in \{BNP, SOGN\}, e \in \{-,+\}.
\end{equation*}
As is Section \ref{subsec:estimator_validation}, we impose the following bid-ask symmetry conditions
\begin{equation*}
    \begin{cases}
        \alpha_{intercept}^{i,-} = \alpha_{intercept}^{i,+}\\
        \alpha_{efficient}^{i,-} = -\alpha_{intercept}^{i,+}\\
        \alpha_{imbalance}^{i,-} = \alpha_{imbalance}^{i,+}
    \end{cases},\quad i \in \{BNP, SOGN\}.
\end{equation*}
The coefficient $\alpha_{efficient}^{i,-}$ is constrained to be negative to ensure stability. Again, the volatilities of the stocks are denoted by $\sigma_{BNP}$, $\sigma_{SOGN}$ and their correlation $\rho$. The time unit is the second, and the price unit is the hundredth of a tick. The estimation results are reported in Figure \ref{fig:bnp_sogn}.

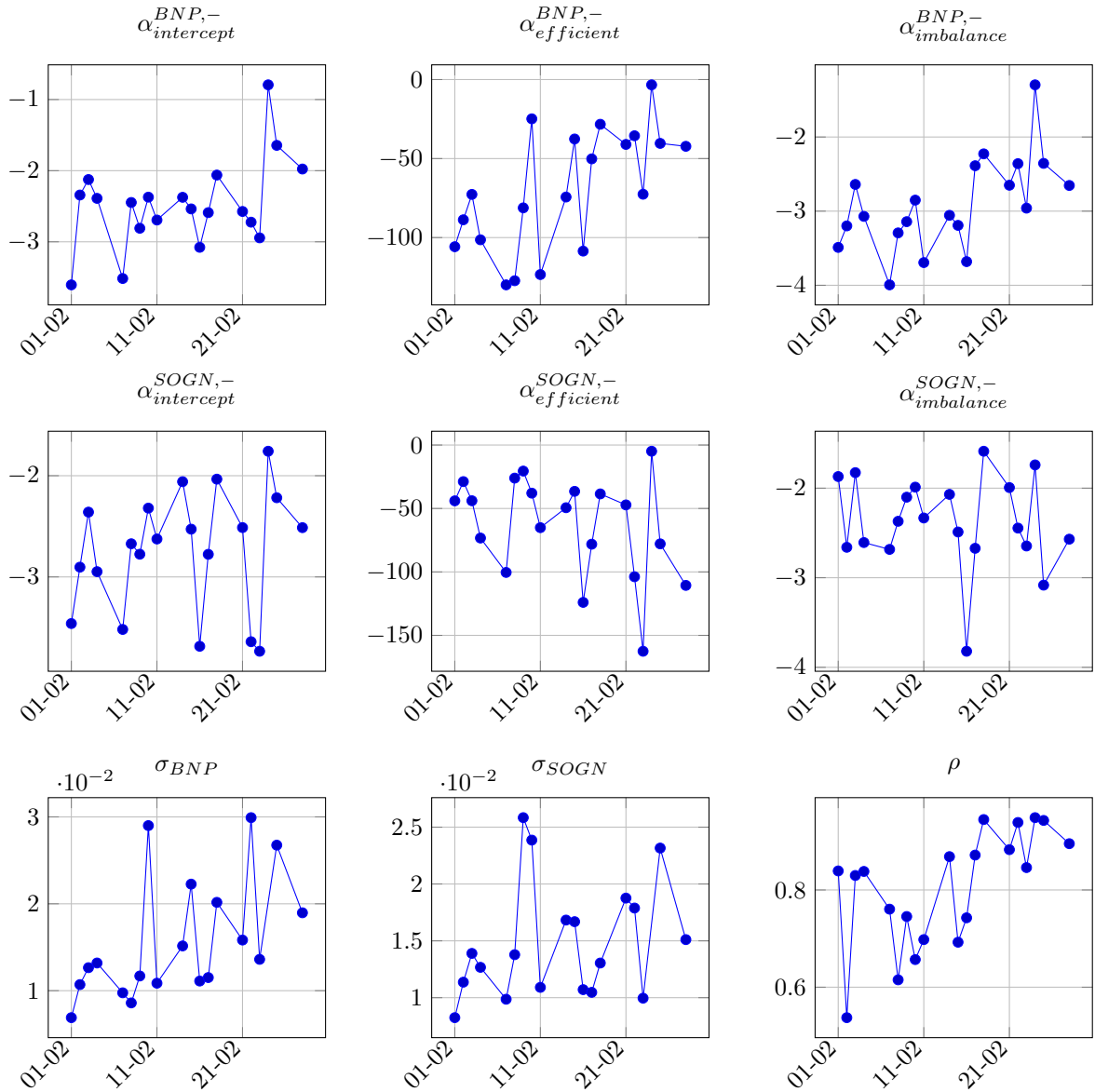
\begin{figure}
    \centering
    \begin{tikzpicture}
    \begin{groupplot}[
        group style={group size=3 by 3,horizontal sep=1.5cm,
        vertical sep=1.8cm},
        height=5cm, width=5.5cm,
        grid=major,
        xticklabel style={rotate=45, anchor=east},
        date coordinates in=x,
        xticklabel={\day-\month},
        x tick label style={/pgf/number format/1000 sep={}}
    ]
    \nextgroupplot[title={$\alpha_{intercept}^{BNP,-}$}] 
    \addplot coordinates {
    (2022-02-01, -3.605044156418763) 
    (2022-02-02, -2.342143467281692) 
    (2022-02-03, -2.124595459118932) 
    (2022-02-04, -2.3889826618329177) 
    (2022-02-07, -3.515570223022696) 
    (2022-02-08, -2.446950906066736) 
    (2022-02-09, -2.8094002907504105) 
    (2022-02-10, -2.372093353948999) 
    (2022-02-11, -2.693606519157913) 
    (2022-02-14, -2.374405294954577) 
    (2022-02-15, -2.537349481636921) 
    (2022-02-16, -3.0773012583599284) 
    (2022-02-17, -2.589482793128974) 
    (2022-02-18, -2.06187968168504) 
    (2022-02-21, -2.573418355407239) 
    (2022-02-22, -2.723944710925025) 
    (2022-02-23, -2.9442409299431125) 
    (2022-02-24, -0.7934948714823848) 
    (2022-02-25, -1.6441180410199885) 
    (2022-02-28, -1.9767788256973808) 
    };
    \nextgroupplot[title={$\alpha_{efficient}^{BNP,-}$}] 
    \addplot coordinates {
    (2022-02-01, -105.797594534538) 
    (2022-02-02, -88.74627456247396) 
    (2022-02-03, -72.67614675695741) 
    (2022-02-04, -101.40918743755448) 
    (2022-02-07, -129.92770966425448) 
    (2022-02-08, -127.3122665365391) 
    (2022-02-09, -81.1990305032183) 
    (2022-02-10, -24.842473110548497) 
    (2022-02-11, -123.400849441462) 
    (2022-02-14, -74.34048004575175) 
    (2022-02-15, -37.610165133771325) 
    (2022-02-16, -108.5611148229816) 
    (2022-02-17, -50.18629478963212) 
    (2022-02-18, -28.29777890339352) 
    (2022-02-21, -41.040493426679895) 
    (2022-02-22, -35.55273661854573) 
    (2022-02-23, -72.56398082142177) 
    (2022-02-24, -3.342080123842215) 
    (2022-02-25, -40.38719751067003) 
    (2022-02-28, -42.241907668674735) 
    };
    \nextgroupplot[title=$\alpha_{imbalance}^{BNP,-}$]
    \addplot coordinates {
    (2022-02-01, -3.4891548520462914) 
    (2022-02-02, -3.1990238521360497) 
    (2022-02-03, -2.6398756373314574) 
    (2022-02-04, -3.0706417007877818) 
    (2022-02-07, -3.9944120667151153) 
    (2022-02-08, -3.293398202079349) 
    (2022-02-09, -3.1419635945159703) 
    (2022-02-10, -2.851463346540972) 
    (2022-02-11, -3.6937468206778505) 
    (2022-02-14, -3.0561382389886083) 
    (2022-02-15, -3.191845441957652) 
    (2022-02-16, -3.6808831551664496) 
    (2022-02-17, -2.3874712718484066) 
    (2022-02-18, -2.227164170888276) 
    (2022-02-21, -2.6502881537639715) 
    (2022-02-22, -2.359370951591103) 
    (2022-02-23, -2.9593273062333414) 
    (2022-02-24, -1.296406764546621) 
    (2022-02-25, -2.3555670731779887) 
    (2022-02-28, -2.653565411646205) 
    };
    \nextgroupplot[title=$\alpha_{intercept}^{SOGN,-}$]
    \addplot coordinates {
    (2022-02-01, -3.4604680040848628) 
    (2022-02-02, -2.903503519522067) 
    (2022-02-03, -2.3596577811869635) 
    (2022-02-04, -2.9478406956025163) 
    (2022-02-07, -3.5187539439161366) 
    (2022-02-08, -2.672402387802042) 
    (2022-02-09, -2.7755730814065056) 
    (2022-02-10, -2.320038537238001) 
    (2022-02-11, -2.62514963423052) 
    (2022-02-14, -2.058985165369339) 
    (2022-02-15, -2.5280713912728348) 
    (2022-02-16, -3.685333750520164) 
    (2022-02-17, -2.7766656633156828) 
    (2022-02-18, -2.033986733622041) 
    (2022-02-21, -2.5124783595297315) 
    (2022-02-22, -3.640146471534006) 
    (2022-02-23, -3.734454858080214) 
    (2022-02-24, -1.7582446090482704) 
    (2022-02-25, -2.217590245318134) 
    (2022-02-28, -2.5130739747520527) 
    };
    \nextgroupplot[title=$\alpha_{efficient}^{SOGN,-}$]
    \addplot coordinates {
    (2022-02-01, -43.992459366422615) 
    (2022-02-02, -28.820765696801573) 
    (2022-02-03, -43.90601582579448) 
    (2022-02-04, -73.29968030750352) 
    (2022-02-07, -100.39810114922396) 
    (2022-02-08, -26.063095390139427) 
    (2022-02-09, -20.52825735471278) 
    (2022-02-10, -37.92823579764363) 
    (2022-02-11, -65.1128039277222) 
    (2022-02-14, -49.40241663386655) 
    (2022-02-15, -36.50604641772914) 
    (2022-02-16, -123.9956987667828) 
    (2022-02-17, -78.15001329647032) 
    (2022-02-18, -38.51380671189079) 
    (2022-02-21, -47.19346969828605) 
    (2022-02-22, -103.79524029625236) 
    (2022-02-23, -162.5316683158241) 
    (2022-02-24, -4.898334144289394) 
    (2022-02-25, -77.8868408726806) 
    (2022-02-28, -110.61736456125162) 
    };
    \nextgroupplot[title=$\alpha_{imbalance}^{SOGN,-}$]
    \addplot coordinates {
    (2022-02-01, -1.870567280251321) 
    (2022-02-02, -2.660073138848625) 
    (2022-02-03, -1.826664035960188) 
    (2022-02-04, -2.6071971472558166) 
    (2022-02-07, -2.6836849282470685) 
    (2022-02-08, -2.3699139664143245) 
    (2022-02-09, -2.100801846820775) 
    (2022-02-10, -1.988727739289712) 
    (2022-02-11, -2.3330231708172504) 
    (2022-02-14, -2.0694169191638943) 
    (2022-02-15, -2.4892293282500755) 
    (2022-02-16, -3.820944832279259) 
    (2022-02-17, -2.671442410009375) 
    (2022-02-18, -1.5880118333991502) 
    (2022-02-21, -1.993539989055911) 
    (2022-02-22, -2.4458780264390305) 
    (2022-02-23, -2.645608404563028) 
    (2022-02-24, -1.740590736505356) 
    (2022-02-25, -3.082691973970829) 
    (2022-02-28, -2.5695543603203714) 
    };
    \nextgroupplot[title=$\sigma_{BNP}$]
    \addplot coordinates {
    (2022-02-01, 0.0068940146608903) 
    (2022-02-02, 0.0107013090913912) 
    (2022-02-03, 0.0126465309524018) 
    (2022-02-04, 0.0131855212747643) 
    (2022-02-07, 0.0097548970978176) 
    (2022-02-08, 0.0085918246112148) 
    (2022-02-09, 0.0116915954804105) 
    (2022-02-10, 0.0290004852543987) 
    (2022-02-11, 0.0108615267395422) 
    (2022-02-14, 0.0151515917369264) 
    (2022-02-15, 0.0222631465058065) 
    (2022-02-16, 0.0111098642717372) 
    (2022-02-17, 0.0115108713432797) 
    (2022-02-18, 0.0201664336059374) 
    (2022-02-21, 0.0158284390278017) 
    (2022-02-22, 0.0299147794414503) 
    (2022-02-23, 0.0136053885335245) 
    (2022-02-25, 0.0267499876547561) 
    (2022-02-28, 0.0189649083641511) 
    };
    \nextgroupplot[title=$\sigma_{SOGN}$]
    \addplot coordinates {
    (2022-02-01, 0.0082441800343962) 
    (2022-02-02, 0.0113648066015285) 
    (2022-02-03, 0.0138966721748639) 
    (2022-02-04, 0.0126710249950731) 
    (2022-02-07, 0.0098648125453407) 
    (2022-02-08, 0.0137828458999092) 
    (2022-02-09, 0.025835823278771) 
    (2022-02-10, 0.0238598546914603) 
    (2022-02-11, 0.0109080610690931) 
    (2022-02-14, 0.0168267642061554) 
    (2022-02-15, 0.0166863233745525) 
    (2022-02-16, 0.0107095913040069) 
    (2022-02-17, 0.0104706254311458) 
    (2022-02-18, 0.0130491003083546) 
    (2022-02-21, 0.0187601137623567) 
    (2022-02-22, 0.0178856785614653) 
    (2022-02-23, 0.0099577395859236) 
    (2022-02-25, 0.0231592241738659) 
    (2022-02-28, 0.0151084724357174) 
    };
    \nextgroupplot[title=$\rho$]
    \addplot coordinates {
    (2022-02-01, 0.8391683633805607) 
    (2022-02-02, 0.5373091668362019) 
    (2022-02-03, 0.8298197803842217) 
    (2022-02-04, 0.8381350614964974) 
    (2022-02-07, 0.760530592130854) 
    (2022-02-08, 0.6151477406500342) 
    (2022-02-09, 0.7453934848348341) 
    (2022-02-10, 0.6568245652948629) 
    (2022-02-11, 0.697913521833949) 
    (2022-02-14, 0.8685829925397651) 
    (2022-02-15, 0.6925304201698222) 
    (2022-02-16, 0.7427449925548881) 
    (2022-02-17, 0.8717745591594935) 
    (2022-02-18, 0.9448576196012696) 
    (2022-02-21, 0.8829677917714823) 
    (2022-02-22, 0.9388425109816948) 
    (2022-02-23, 0.8460453469185575) 
    (2022-02-24, 0.948765074144774) 
    (2022-02-25, 0.9430560476239016) 
    (2022-02-28, 0.8952187147153097) 
    };
    \end{groupplot}\end{tikzpicture}
    \caption{Estimated coefficients for the dynamics of market orders for the stocks BNP Paribas and Société Générale. One estimation per trading day, over February 2022. We removed February 24 from the plots of $\sigma_{BNP}$ (value: 0.31085) and $\sigma_{SOGN}$ (value: 0.37024) for clarity.}
    \label{fig:bnp_sogn}
\end{figure}

We observe that the coefficients stay in the same orders of magnitude across February, except on the \nth{24} where the volatilities reported are a hundred times higher than the other days. It is probably due to the optimizer being stuck in a local maximum, behavior that was also observed on simulated data.

The imbalance coefficients $\alpha_{imbalance}^{BNP,-}$ and $\alpha_{imbalance}^{SOGN,-}$ are consistently negative, which corresponds to the know empirical findings: when there is more volume on the best ask pile than on the bid, the price is likely to go down in the near future. The correlation coefficient $\rho$ is always reported positive (and high), validating that the price move of the two stocks are, in a macroscopic scale, positively correlated.

\section{Application: market impact and liquidation}
\label{sec:liquidation}

We use, as in \parencite{huang2015simulating}, our model in the context of execution of a fixed volume of some assets. Typically, see \parencite{almgren2001optimal}, a trader wants to buy or sell a fixed quantity of an asset and looks to do so for the most advantageous price. Trading too fast is suboptimal because of transient market impact: the cost of trading is an increasing function of trading speed.

In general, the market impact of a trade is defined as the difference between the actual price of the asset and the price the asset would have if the said trade never occurred (the definition of the \enquote{price} can vary: mid-price, last traded price or reference price to cite a few). It consists of two components: a transient part that vanishes over time, and a permanent part.

The fact that a huge trade costs more is directly observable in the LOB framework. Say that an agent wants to buy 100 units of a stock. In the ask side of the LOB, there are 50 units pending at 100\euro{}, 20 units at 100.5\euro{} and 40 units at 101\euro{}. If the agent places a buy market order for 100 units, she pays 100.4\euro{} per unit. Immediately after the trade, the best ask price jumps to 101\euro{}. Instead, she could buy 50 units and wait for the pile at 100\euro{} to be replenished to buy again at 100\euro{}, hence reducing her costs, with the drawback that there is no guarantee that someone will post ask limit orders at this price again, introducing some variance in her strategy.

Having a reliable LOB model is thus of importance for the agent, who can then backtest her execution strategies. The model of \textcite{huang2015simulating} exhibits mainly permanent market impact (see Figure 12 in said paper). On the contrary, the market impact in the models presented here can only be transient: we saw that the reference price will always come back close to the efficient price. In Figure \ref{fig:impact}, we plot the market impact (difference between the reference price right before the market order and the reference price at time $t$) over 4 minutes of a buy market order of size 100. The model used is the efficient price-driven queue-reactive model of a single stock with $K=3$, tick size 0.01, volatility 0.2, and the limit, cancel and market orders of unit size. A limit order of at the $j$th ask pile has intensity $A^j e^{B(S_t-P_{t-}) + C q_t^{a,j}}$ where $q_t^{a,j}$ is the volume pending on the said pile, with $(A^1, A^2, A^3) = (1.5, 1.2, 1.0)$, $B = -1.0$ and $C=-0.1$. A liquidity consuming order (market or cancel) has intensity $A'^j e^{B'(S_t-P_{t-}) + C' q_t^{a,j}}\mathds{1}_{\{q_t^{a,j} > 0\}}$ where $(A'^1, A'^2, A'^3) = (1.2, 1.0, 0.8)$, $B' = 1.5$ and $C'=0.15$. The parameters are symmetric on the bid side. We observe a sharp increase of the reference price after the large market order: a big portion of the ask side of the LOB is depleted. Then, over time, the reference price goes back to normal, close to the efficient price.

\begin{figure}
    \centering
    \includegraphics[height=6cm]{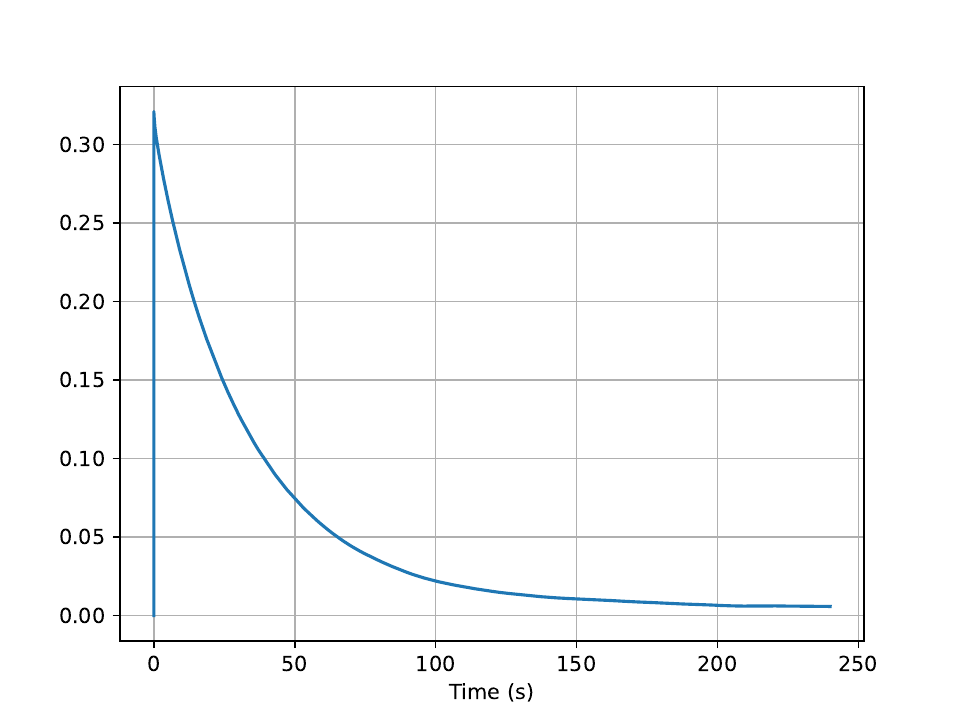}
    \caption{Market impact of a buy market order of size 100 over 4 minutes. Monte-Carlo simulation with 500000 samples.}
    \label{fig:impact}
\end{figure}

When designing execution strategies over multiple assets, taking into account their correlation is important, as the distribution of their profit and loss is dependent on it. To illustrate it, we plot in Figure \ref{fig:execution_pnl} the execution cost per share of buying 500 units of Stock 1 and 300 units of Stock 2 over 10 minutes, placing a market order every 30 seconds (of size 25 for Stock 1 and 15 for Stock 2). The volatilities are 0.02 and 0.01 and three different correlation parameters have been chosen: -0.8, 0.0 and 0.8. The other parameters (intensities and tick size) are the same as in the previous paragraph. We observe that with this strategy, the agent overpays on average 17 cents per share, however the variance decreases with the correlation: with a negative correlation, price rises of Stock 1 are often compensated by price drops of Stock 2.

\begin{figure}
    \centering
    \includegraphics[height=6cm]{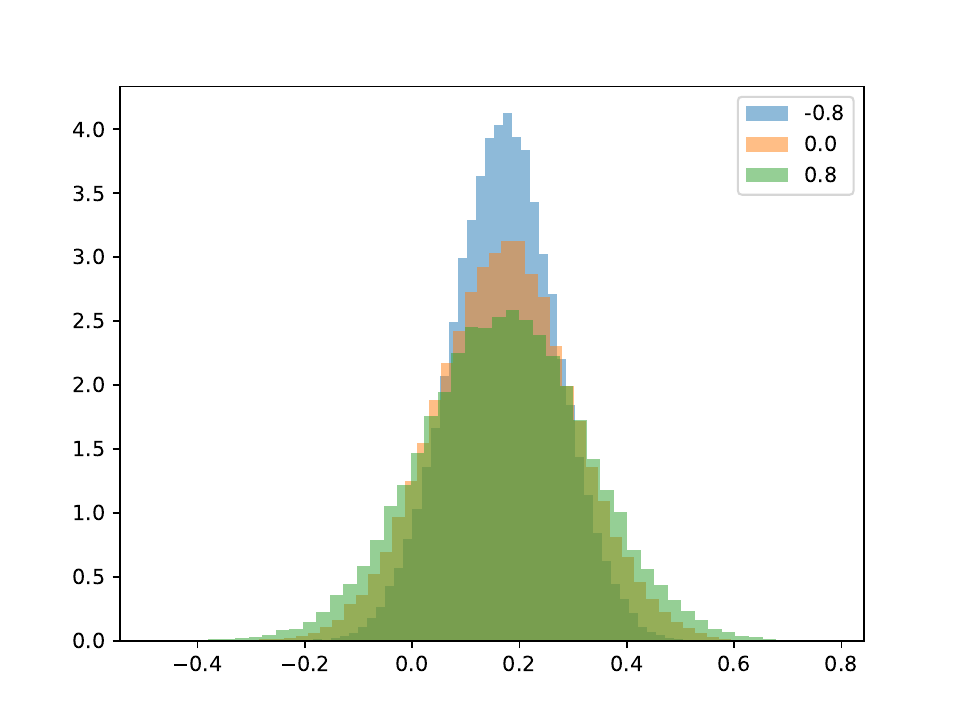}
    \caption{Histogram of the difference of the price paid per share by the buyer and their initial price using her execution strategy. Monte-Carlo simulation with 50000 samples. Plotted for three correlation parameters: -0.8 (in blue), 0.0 (in orange), 0.8 (in green).}
    \label{fig:execution_pnl}
\end{figure}

The agent can use our model to design her execution strategies over multiple assets. However, our model does not capture permanent market impact, nor cross impact (a market order on a stock does not influence the LOB of the others in any way), while these phenomena are observed in real markets. A possible improvement of the model in this regard is allowing market orders to act on the efficient price, for example through a constant jump.

\section{Proofs of stability}
\label{section:proofs_stability}
In this section, we present the proofs for the results stated in Section \ref{sec:models}. In Subsections \ref{subsec:proofs_signal_modulated} and \ref{subsec:proofs_queue_reactive}, we build Lyapunov functions and show upper bounds for them., for the signal-driven price model and the queue-reactive model respectively. In Subsection \ref{subsec:proofs_general_results}, we prove some consequences of the Lyapunov inequalities applicable to both models.

We use the notations and processes defined in Section \ref{sec:models}.

\subsection{General results}
\label{subsec:proofs_general_results}

We start by deducing the non-explosiveness of $(S,P,X,N)$ given $(S-P, X)$ is non-explosive.

\begin{proposition}
    \label{prop:non_explosiveness_lyapunov}
    Suppose that the Markov process $(S-P, X)$ is non-explosive and that $\sum_{k \in \mathbb{K}} \Lambda^{k}$ is locally bounded. Then, the Markov process $(S,P,X,N)$ is non-explosive.
\end{proposition}

\begin{proof}
    Since $S$ is non-explosive, being an affine transformation of a Brownian motion, it is sufficient to show that $(S-P, X, N)$ is non-explosive. Since $(S-P, X)$ is non-explosive by assumption, it remains to show that it is also the case for $N$. Suppose this does not hold, in particular
\begin{equation*}
    \mathbb{P}\bigg[\sum_{k \in \mathbb{M}}N_t^k = \infty\bigg] > 0
\end{equation*}
for some $t > 0$. Since $(S-P, X)$ is non-explosive there exists $M > 0$ such that $\mathbb{P}[\tilde{N}^k_t = \infty, A] > 0$ where $\tilde{N} \defeq \sum_{k \in \mathbb{M}}N^k$ and $A \defeq \{\sup_{u\in[0,t]} (|S_u-P_u| + |X_u|) \leqslant M\}$. For $u \in [0, \infty)$ and $m \in \mathbb{N}$, define the function $f_m(u) \defeq \mathbb{P}[\tilde{N}^k_u \geqslant m, A]$. The function $f_m$ is right-continuous. We show that $\lim_{m\to \infty}f_m(u) = 0$, which contradicts $\mathbb{P}[\tilde{N}^k_t = \infty, A] > 0$.

Specifically, we show by induction that \begin{equation*}
    f_m(u) \leqslant 1 - e^{-Cu}\sum_{n=0}^{m-1}\frac{(Cu)^n}{n!}
\end{equation*}
where $C\defeq \sup_{(x,y)\in\mathcal{X} \times [-M, M]^{\mathbb{M}}} \sum_{k \in \mathbb{M}}\Lambda^k(x, y)$. The property holds for $m = 0$. Suppose now it holds for some $m \in \mathbb{N}$. For $u \geqslant 0$,
\begin{equation*}
    f_{m+1}(u) = \sum_{k \in \mathbb{M}}\int_0^u \E[][\mathds{1}_A \big( \mathds{1}_{\{\tilde{N}_r \geqslant m+1\}} - \mathds{1}_{\{\tilde{N}_r \geqslant m\}}\big) \Lambda^k(X_r, S_r - P_r)] \diff r.
\end{equation*}
Thus, we have, for $u \geqslant 0$,
\begin{equation*}
    e^{Cu}f_{m+1}(u) = \sum_{k \in \mathbb{M}}\int_0^u C e^{Cr} f_{m+1}(r) \diff r + \sum_{k \in \mathbb{M}}\int_0^u e^{Cr} \E[][\mathds{1}_A \big( \mathds{1}_{\{\tilde{N}_r \geqslant m\}} - \mathds{1}_{\{\tilde{N}_r \geqslant m+1\}}\big) \Lambda^k(X_r, S_r - P_r)] \diff r.
\end{equation*}
By the definition of $A$ and $C$, we find an upper bound for the second terms which \textit{in fine} gives
\begin{equation*}
    e^{Cu}f_{m+1}(u) \leqslant C \int_0^u e^{Cr} f_{m-1}(r) \diff r
    \leqslant C \int_0^u e^{Cr} \diff r - \sum_{n=0}^{m-1}\frac{C^{n+1}}{n!}\int_0^u r^n \diff r.
\end{equation*}
The result follows after computing the integrals.
\end{proof}

We give a condition of Lyapunov function for the convergence in probability of $\tilde{S}^{(n)}-\tilde{P}^{(n)}$ to 0.

\begin{proposition}
    \label{prop:lyapunov_convergence_proba_general}
    Suppose there exists a constant $y_0 \in \mathbb{R}$, a non-decreasing $C^{2}$ function $\psi:[y_0, \infty) \to \mathbb{R}$, a constant $C \in (0, \infty)$, and a locally bounded function $V:\mathbb{R}^d \times \mathcal{X} \to [0,\infty)$, $C^2$ in its first variable $y$, such that for $x \in \mathcal{X}$ and $y = (y^i)_{i \in \stocks}$,
    \begin{equation*}
        V(y,x) \geqslant \sum_{i \in \stocks} (y^i)^2 \psi(y^i) \mathds{1}_{\{y^i \geqslant y_0\}}
        \text{ and } \mathcal{L}V(y,x) \leqslant C.
    \end{equation*}
    Suppose moreover that $S_0 - P_0$ and $X_0$ are bounded. Then, for all $T, \epsilon > 0$.
    \begin{equation*}
        \lim_{n\to \infty} \Proba[][\big| \tilde{S}^{(n)}-\tilde{P}^{(n)}\big|_{\infty, [0,T]} \geqslant \epsilon] = 0.
    \end{equation*}
\end{proposition}

\begin{remark}
    The assumptions of Proposition \ref{prop:lyapunov_convergence_proba_general} are satisfied for the $V$ defined in Section \ref{subsec:signal_modulated}, taking $\psi(y) = \frac{1}{d} y^4$ and $y_0$ being such that for all $y \geqslant y_0$, $\frac{1}{d}y^4 < e^y$.
\end{remark}

\begin{proof}
    Let $T > 0$. Let $A > y_0$. Then, we have
    \begin{equation*}
        \Proba[][|S-P|_{\infty, [0,T]}  \geqslant A] \leqslant
        \Proba[][\sup_{u \in [0,T]} V(S_u - P_u, X_u)  \geqslant A^2 \psi(A)]. 
    \end{equation*}
    By \parencite[Chapter III, Corollary 1.1]{kushner1967stochastic},
    \begin{equation*}
        \Proba[][\sup_{u \in [0,T]} V(S_u - P_u, X_u)  \geqslant A^2 \psi(A)]
        \leqslant \frac{\E[][V(S_0 - P_0, X_0)]}{A^2 \psi(A)} + \frac{CT}{A^2 \psi(A)}.
    \end{equation*}
    Note that $\E[][V(S_0 - P_0, X_0)]$ is finite thanks to our assumptions. Let $\epsilon > 0$ and $n \in \mathbb{N}^*$. We deduce from the above that
    \begin{equation*}
        \Proba[][|\tilde{S}^{(n)}-\tilde{P}^{(n)}|_{\infty, [0,T]}  \geqslant \epsilon] = \Proba[][|S - P|_{\infty, [0,nT]}  \geqslant \sqrt{n}\epsilon]
        \leqslant
        \frac{\E[][V(S_0 - P_0, X_0)]}{n \epsilon^2 \psi(\sqrt{n}\epsilon)} + \frac{CT}{\epsilon^2 \psi(\sqrt{n}\epsilon)}.
    \end{equation*}
    We conclude observing that the right-hand side tends to zero as $n$ tends to infinity.
\end{proof}

We now deduce the following convergence in law from the convergence in probability.
\begin{corollary}
    \label{corol:skorokhod_convergence}
    Under the assumptions of Proposition \ref{prop:lyapunov_convergence_proba_general}, $\tilde{P}^{(n)}$ converges in distribution in $D([0, \infty), \mathbb{R}^d)$ to $\Sigma B$ where $B$ is a $d$-dimensional Brownian motion.
\end{corollary}

\begin{proof}
    Since $\tilde{S}^{(n)}$ has the law of $\Sigma B$ for all $n$, by \parencite[Theorem 3.1]{BillingsleyPatrick1925-20111999Copm}, $(\tilde{P}^{(n)})_{n \in \mathbb{N}^*}$ converges in distribution in $D([0, T], \mathbb{R}^d)$ to $\Sigma B$.

    This being true for all $T > 0$ allows us to conclude with \parencite[Lemma 16.3]{BillingsleyPatrick1925-20111999Copm}.
\end{proof}

\subsection{Signal-driven price}
\label{subsec:proofs_signal_modulated}
Now, we prove Proposition \ref{prop:exp_inequality}. Let $y=(y^i)_{i\in \stocks} \in \mathbb{R}^d$ and $x \in \mathcal{X}$. For $i, j \in \stocks$, we have
\begin{equation*}
    \partial^2_{y^i y^j} V(y, x) = \begin{cases}
        U'(y^i)U'(y^j)V(y, x), & \text{if } i \neq j\\
        (U''(y^i) + U'(y^i)^2 )V(y, x),& \text{if } i = j.
    \end{cases}
\end{equation*}
Since when $|y^i| \geqslant 1$, $U''(y^i) = 0$ and $|U'(y^i)|=1$, $\nabla^2_y V = V \cdot f$ where $f$ is a bounded matrix-valued function. Thus, there exists a constant $M > 0$ such that $\frac{1}{2}\tr(\Sigma \Sigma^T \nabla^2_y V) \leqslant M V$. Thus,
\begin{equation*}
\mathcal{L}V(y,x) \leqslant V(y, x)\bigg(
    M + \sum_{i \in \stocks}\tilde{\Lambda}^{i, +}(x, y^i)\bigg( \frac{V(y-\delta^i e_i,x)}{V(y, x)}  - 1\bigg)
    + \sum_{i \in \stocks}\tilde{\Lambda}^{i, -}(x, y^i)\bigg( \frac{V(y+\delta^i e_i,x)}{V(y, x)}  - 1\bigg)
\bigg).
\end{equation*}
Define $\Delta \defeq \max(1, \max_{i \in \stocks}\delta^i)$. By Assumption \ref{assumption:limits}, there exists $M' > 0$ such that for all $i \in \stocks$, $\tilde{\Lambda}^{i, +}\mathds{1}_{\mathcal{X} \times (-\infty,  \Delta)} \leqslant M'$ and $\tilde{\Lambda}^{i, -}\mathds{1}_{\mathcal{X} \times (-\Delta, \infty)} \leqslant M'$.
Hence, the following inequality holds for some constant $M'' > 0$:
\begin{equation*}
\mathcal{L}V(y,x)
\leqslant
V(y, x)\bigg(
M'' + \sum_{i \in \stocks}\tilde{\Lambda}^{i, +}(x, y^i)\big(e^{-\delta^i} - 1\big) \mathds{1}_{\{y^i \geqslant \Delta\}}
+ \sum_{i \in \stocks}\tilde{\Lambda}^{i, -}(x, y^i)\big(e^{-\delta^i} - 1\big) \mathds{1}_{\{y^i \leqslant -\Delta\}}
\bigg).
\end{equation*}
Let $m > \Delta$ such that for all $i \in \stocks$, $(e^{-\delta^i} - 1)\tilde{\Lambda}^{i, +}\mathds{1}_{\mathcal{X} \times (m,  \infty)} \leqslant -(1+M'') \mathds{1}_{\mathcal{X} \times (m,  \infty)}$ and $(e^{-\delta^i} - 1)\tilde{\Lambda}^{i, +}\mathds{1}_{\mathcal{X} \times (-\infty, -m)} \leqslant - (1+M'')\mathds{1}_{\mathcal{X} \times (-\infty, -m)}$. If $y \in \mathbb{R}^d \setminus [-m,m]^d$, then $\mathcal{L}V(y,x) \leqslant - V(y,x)$. We conclude taking
\begin{equation*}
K \defeq  \sup_{x \in \mathcal{X}} \sup_{y \in [-m,m]^d} \big(|\mathcal{L}V(y,x)| + V(y,x)\big)
\end{equation*}
which is finite because the intensity functions are assumed locally bounded here.

\subsection{Queue-reactive model}
\label{subsec:proofs_queue_reactive}

\subsubsection{A preliminary lemma}

\begin{lemma}
    \label{lem:bound_from_below}
    Let $\xi:[0, \infty) \to [0,\infty)$ be a function such that $\lim_{y \to \infty}\xi(y)=\infty$. There exists a non-decreasing $C^{\infty}$ function $\psi:[0, \infty) \to [0,\infty)$ such that $\lim_{y \to \infty}\psi(y)=\infty$, $\lim_{y \to \infty} \frac{\xi(y)}{\psi(y)} = \infty$, $\psi'' \leqslant 0$ and for all $y \in [0, \infty)$, $\psi'(y) \leqslant \frac{1}{1+y}$.
\end{lemma}

\begin{proof}
    Replacing $\xi$ by $\hat{\xi}$ defined by $\hat{\xi}(y)\defeq \inf \{\xi(z): z \geqslant y\}$, we see that it is sufficient to prove the result if $\xi$ is non-decreasing, which we assume from now on.

    Suppose we have proved that there exists a $C^{\infty}$ function $g:[0, \infty) \to [0,\infty)$ such that  $\lim_{y \to \infty}g(y)=\infty$, $g'' \geqslant 0$ and for all $y \in [0, \infty)$, $0 \leqslant g'(y) \leqslant \frac{1}{1+y}$ and $\xi(y) - g(y) \geqslant -1$. Then, setting $\psi \defeq \ln(1+g)$, we get the desired result. The goal is now to find such a $g$.

    \textbf{Step 1:} Piecewise construction of $f$, an approximation of the target $g$.

    Let $(x_n)_{n \in \mathbb{N}}$ a strictly increasing sequence of non-negative reals such that $y_0$ = 0, $\lim_{n \to \infty} y_n = \infty$, and for all $n \in \mathbb{N}$, $\xi(y_{n+1}) - \xi(y_n) > 1$.

    If $\sup_{y \in [0, y_1]} \int_0^y \frac{\diff z}{1+z} - \xi(y) < 1$ and $\xi(y_1) \geqslant \int_0^{y_1} \frac{\diff z}{1+z}$, set $a_0 \defeq 1$. Otherwise, let $a_0 \in (0,1)$ be such that $\sup_{y \in [0, y_1]} a_0\int_0^y \frac{\diff z}{1+z} - \xi(y) < 1$, $\xi(y_1) \geqslant a_0\int_0^{y_1} \frac{\diff z}{1+z}$ and there exists $y \in [0,y_1]$ such that $\xi(y) \leqslant a_0\int_0^{y} \frac{\diff z}{1+z}$. Set $f(y) \defeq a_0\int_0^{y} \frac{\diff z}{1+z}$ for $y \in [0, y_1]$.

    Suppose we have built $a_0,\dots,a_{n-1}$ and $f$ on $[0,y_n]$, and suppose $f(y_n) \leqslant \xi(y_n)$.
    
    If $\sup_{y \in [y_n, y_{n+1}]} f(y_n)+\int_{y_n}^{y}  a_{n-1}\frac{\diff z}{1+z} - \xi(y) < 1$ and $\xi(y_{n+1}) \geqslant f(y_n)+\int_{y_n}^{y_{n-1}}  a_{n-1}\frac{\diff z}{1+z}$, set $a_n \defeq a_{n-1}$. Otherwise, let $a_n \in (0,a_{n-1})$ be such that $\sup_{y \in [y_n, y_{n+1}]} f(y_n)+\int_{y_n}^{y}  a_{n}\frac{\diff z}{1+z} - \xi(y) < 1$, $\xi(y_{n+1}) \geqslant f(y_n)+ \int_{y_n}^{y_{n-1}}  a_{n}\frac{\diff z}{1+z}$ and there exists $y \in [y_n,y_{n+1}]$ such that $\xi(y) \leqslant f(y_n)+\int_{y_n}^{y}  a_{n}\frac{\diff z}{1+z}$. Set $f(y) \defeq f(y_n) + \int_{y_n}^y a_n \frac{\diff z}{1+z}$ for $y \in [y_n, y_{n+1}]$.

    We have built a non-increasing sequence $(a_n)_{n \in \mathbb{N}} \in (0, 1)^{\mathbb{N}}$ and an increasing function $f:[0,\infty) \to [0, \infty)$ such that
    \begin{equation*}
        f(y) = \int_0^y \sum_{n=0}^{\infty}a_n\mathds{1}_{[y_n, y_{n+1})}(z)\frac{1}{1+z} \diff z,\quad y \in [0, \infty)
    \end{equation*}
    and for all $y \in [0, \infty)$, $\xi(y) - f(y) \geqslant -1$. If there exists $n_0 \in \mathbb{N}$ such that for all $n \geqslant n_0$, $a_n = a_{n_0}$, then for all $y \in [y_{n_0 - 1}, \infty)$,
    \begin{equation*}
        f(y) \geqslant a_{n_0} \int_{y_{n_0-1}}^y \frac{\diff z}{1 + z}
        = a_{n_0} \ln(y+1) - a_{n_0} \ln(y_{n_0-1}+1) \xrightarrow[y \to \infty]{} \infty.
    \end{equation*}
    Otherwise, by construction, there exists a sequence $(y'_n)_{n \in \mathbb{N}}$ such that $\lim_{n \to \infty}y'_n = \infty$ and, for all $n \in \mathbb{N}$, $f(y'_n) \geqslant \xi(y'_n)$. Thus, $\lim_{y \to \infty} f(y) = \infty$.

    \textbf{Step 2:} Construction of $g$ by smoothing $f$.

    Let $(\epsilon_n)_{n \in \mathbb{N}}$ be a sequence of strictly positive reals such that for all $n \in \mathbb{N}$, $\epsilon_n < \min(\frac{1}{2^{n+1}}, y_{n+1} - y_n)$. For $n \in \mathbb{N}$, let $\varphi_n:[y_n, y_{n+1}]\to \mathbb{R}$ be a non-increasing $C^{\infty}$ function such that $\varphi_n \equiv a_n$ on $[y_n, y_{n-1} - \epsilon_n]$, and $\varphi_n \equiv a_{n+1}$ in a neighborhood of $y_{n+1}$. We have that $\sum_{n=0}^{\infty}\varphi_n\mathds{1}_{[y_n, y_{n+1})}$ is non-increasing, $C^{\infty}$ and lower or equal than $\sum_{n=0}^{\infty}a_n\mathds{1}_{[y_n, y_{n+1})}$. Define
    \begin{equation*}
        g(y) \defeq \int_0^y \sum_{n=0}^{\infty}\varphi_n(z)\mathds{1}_{[y_n, y_{n+1})}(z)\frac{1}{1+z} \diff z,\quad y \in [0, \infty).
    \end{equation*}
    The function $g$ is $C^{\infty}$, $g'(y) = \sum_{n=0}^{\infty}\varphi_n(y)\mathds{1}_{[y_n, y_{n+1})}(y)\frac{1}{1+y} \leqslant \frac{1}{1+y}$ for $y \in [0, \infty)$, and is second derivative is non-positive. Since $f \geqslant g$, $\xi - g \geqslant - 1$ holds. For $y \in [0, \infty)$,
    \begin{equation*}
        f(y) - g(y) \leqslant \int_{0}^y \sum_{n= 0}^{\infty} \mathds{1}_{y_{n+1}-\epsilon_n, y_{n+1}} \frac{\diff z}{1 + z} \leqslant \sum_{n = 0}^{ \infty} \frac{1}{2^{n+1}} = 1.
    \end{equation*}
    Since $f$ tends to $\infty$ at infinity, so does $g$. The function $g$ satisfies all the desired properties.
\end{proof}

We immediately deduce Corollary \ref{corol:bound_from_below} from Lemma \ref{lem:bound_from_below}, extending the result to both sides of the real line.
\begin{corollary}
    \label{corol:bound_from_below}
    Let $\xi:\mathbb{R} \to [0,\infty)$ be a function such that $\lim_{y \to \pm\infty}\xi(y)=\infty$.
    There exists an even $C^{\infty}$ function $\psi: \mathbb{R} \to [0, \infty)$, non-decreasing on $[0, \infty)$ such that $\lim_{y \to \pm\infty}\psi(y)=\infty$, $\lim_{y \to \pm\infty} \frac{\xi(y)}{\psi(y)} = \infty$, and, for $y \in (-\infty, -1]\cup[1, \infty)$, $\psi''(y) \leqslant 0$ and $|\psi'(y)| \leqslant \frac{1}{1+y}$.
\end{corollary}

\subsubsection{Proof of Proposition \ref{prop:qr_lyapunov}}

Here, $(x)_+$ denotes $\max(0,x)$ for any real $x$. Let $\rho:\mathbb{R} \to [0,1]$ be a non-decreasing $C^{\infty}$ function such that for $x \in (-\infty, 0]$, $\rho(x) = 0$ and for $x \in [1, \infty)$, $\rho(x) = 1$. Let $Q, Q' > 0$ be the constants given by Assumptions \ref{assumption:non_exploding_queue} and \ref{assumption:finite_expectation_size_orders}. Let $\xi$ be the function given by Assumption \ref{assumption:qr_limits}, and $\psi$ the associated function given by Corollary \ref{corol:bound_from_below}.

For $q=(q_i^j)_{i \in \stocks,j \in \mathbb{J}} \in \mathcal{Q}^{\stocks}$ and $y = (y^i)_{i \in \stocks}$, define
\begin{equation*}
    V(y, q) \defeq
    \sum_{i \in \stocks} \bigg((y^i)^2 \psi(y^i) + \sum_{j \in \mathbb{J}}(q_i^j -Q -Q')_+ + \sum_{\substack{j=(s,l)\in \mathbb{J} \\ s=a}} \rho(y^i) \min(q_i^j, Q+Q') + \sum_{\substack{j=(s,l)\in \mathbb{J} \\ s=b}} \rho(-y^i) \min(q_i^j, Q+Q')\bigg).
\end{equation*}
We define, for $q=(q^j)_{j \in \mathbb{J}} \in \mathcal{Q}$ and $y \in \mathbb{R}$,
\begin{align*}
    \tilde{W}(y,q) &\defeq
    \sum_{j \in \mathbb{J}}(q^j -Q -Q')_+ + \sum_{\substack{j=(s,l)\in \mathbb{J} \\ s=a}} \rho(y^i) \min(q^j, Q+Q') + \sum_{\substack{j=(s,l)\in \mathbb{J} \\ s=b}} \rho(-y^i) \min(q^j, Q+Q'),\\
    W(q) &\defeq
    \sum_{\substack{j=(s,l)\in \mathbb{J} \\ s=b}}(q^j -Q -Q')_+ + \sum_{\substack{j=(s,l)\in \mathbb{J} \\ s=a}} q^j.
\end{align*}
Note that $W(q)=\tilde{W}(y,q)$ for any $y \geqslant 1$.
We have
\begin{equation*}
    \mathcal{L}V(y,q) = \sum_{i \in \stocks} \mathcal{L}_i V (y^i, q_i),\quad
    q=(q_i^j)_{i \in \stocks,j \in \mathbb{J}} \in \mathcal{Q}^{\stocks}, y = (y^i)_{i \in \stocks}
\end{equation*}
where, for $i \in \stocks$, $q=(q^j)_{j \in \mathbb{J}} \in \mathcal{Q}$ and $y \in \mathbb{R}$,
\begin{equation*}
    \begin{split}
    \mathcal{L}_i V(y, q)
    &= (\Sigma \Sigma^T)_{ii}\bigg(1 + 2y \psi'(y) + \frac{y^2}{2}\psi''(y)
    + \sum_{\substack{j=(s,l)\in \mathbb{J} \\ s=a}} \rho''(y) \min(q^j, Q+Q') + \sum_{\substack{j=(s,l)\in \mathbb{J} \\ s=b}} \rho''(-y) \min(q^j, Q+Q')\bigg) \\
    &\phantom{=}+ \sum_{\substack{e \in \mathbb{T}\\\delta p^e(q) = 0}}\big( \tilde{W}(y,  q +\delta q^e) - \tilde{W}(y,q) \big) \tilde{\Lambda}^{i, e}(q, y)\\
        &\phantom{=}+ \sum_{\substack{e \in \mathbb{T}\\\delta p^e(q) \neq 0}}\big((y-\delta^i \delta p^e(q))^2 \psi(y - \delta^i \delta p^e(q)) - y^2 \psi(y)\big) \tilde{\Lambda}^{i, e}(q, y)\\
        &\phantom{=}+ \sum_{\substack{e \in \mathbb{T}\\\delta p^e(q) > 0}}\bigg(\int_{(\mathbb{N}^*)^{\delta p^e(q)}} \tilde{W}(y-\delta^i\delta p^e(q), [q + \delta q^e, v])\mathcal{V}^{i,a,\delta p^e(q)}(\diff v) - \tilde{W}(y, q) \bigg) \tilde{\Lambda}^{i, e}(q, y)\\
        &\phantom{=}+ \sum_{\substack{e \in \mathbb{T}\\\delta p^e(q) < 0}}\bigg(\int_{(\mathbb{N}^*)^{-\delta p^e(q)}} \tilde{W}(y-\delta^i\delta p^e(q), [v,q + \delta q^e])\mathcal{V}^{i,b,-\delta p^e(q)}(\diff v) - \tilde{W}(y, q) \bigg) \tilde{\Lambda}^{i, e}(q, y).
    \end{split}
\end{equation*}
To show the desired result, it is sufficient to show that $\mathcal{L}_i V \leqslant C$ for some $C \in (0,\infty)$ for all $i \in \stocks$. Fix $i \in \stocks$. 

Define $M' > 0$ as
\begin{equation*}
    M' \defeq
    2\max_{l \in \{1,\dots,K\},s\in \{a,b\}} \bigg(
        \int_{(\mathbb{N}^*)^l} \bigg(\sum_{l' = 1}^l v^{l'} \bigg)\mathcal{V}^{i,s,l}(\diff v)
    \bigg).
\end{equation*}
Since $\lim_{y \to \pm\infty} \sup_{l \in \{-K,\dots,K\}}|\frac{y(\psi(y-l \delta^i)- \psi(y))}{\psi(y-l\delta^i)} + \frac{1}{y}\psi(y-l \delta^i)| = 0$ (by the condition on the derivative of $\psi$), there exists $y_0 \in (1 + 2K \delta^i, \infty)$ such that for $|y| \geqslant y_0$, $|\frac{y(\psi(y-l \delta^i)- \psi(y))}{\psi(y-l\delta^i)} + \frac{1}{y}\psi(y-l \delta^i)| \leqslant \delta^i$. Let $y_1 \in (y_0, \infty)$ such that for $y \geqslant y_1$, $\delta^i y \psi(y) > M'$.

Note that by the local boundedness of the intensities associated to price moves or potential increase of the value of $\tilde{W}$ (Assumption \ref{assumption:qr_local_boundedness}), $\sup_{q \in \mathcal{Q}} \mathcal{L}_i V(\cdot, q)$ is bounded from above on $[-y_1,y_1]$. Since the first term in the definition of $\mathcal{L}_i$ is uniformly bounded from above, it is sufficient to find an upper bound for $\mathcal{M}_i V$ defined below, for $y \geqslant y_1$ (the case $y \leqslant -y_1$ is similar). Let $q = (q^j)_{j \in \mathbb{J}} \in \mathcal{Q}$ and $y \in [y_1, \infty)$.

By Assumption \ref{assumption:qr_limits}, we have the existence of two constants $M > 0$ and $y_2 \in (y_1, \infty)$ such that for all $|y| \geqslant y_2$,
\begin{equation*}
    \begin{split}
        &\sum_{\substack{e=(+,j,s,n) \in \mathbb{T}^{+}\\s= a}} n \tilde{\Lambda}^{i, e}(q, y)
    +\sum_{\substack{e=(\text{modif},(s,l),j,n) \in \mathbb{T}^{modif}\\s=a\\ j - j^{a,best(q)} < 0}} n\tilde{\Lambda}^{i, e}(q, y)\\
    &+\sum_{\substack{e=(-,(s,l),n) \in \mathbb{T}^-\\s=b\\ (s,l)=j^{b,best}(q)}}n\tilde{\Lambda}^{i, e}(q, y)
    + \sum_{\substack{e=(\text{modif},(s,l),j,n) \in \mathbb{T}^{modif}\\s=b\\ (s,l)=j^{b,best}(q), j - j^{b,best}(q) < 0}}n\tilde{\Lambda}^{i, e}(q, y) \leqslant M
    \end{split}
\end{equation*}

We define
\begin{equation*}
    \begin{split}
        \mathcal{M}_i V(y,q)
        &\defeq \sum_{\substack{e \in \mathbb{T}\\\delta p^e(q) = 0}}\big( W(q +\delta q^e) - W(q) \big) \tilde{\Lambda}^{i, e}(q, y)\\
        &\phantom{=}+ \sum_{\substack{e \in \mathbb{T}\\\delta p^e(q) \neq 0}}\big((y-\delta^i \delta p^e(q))^2 \psi(y - \delta^i \delta p^e(q)) - y^2 \psi(y)\big) \tilde{\Lambda}^{i, e}(q, y)\\
        &\phantom{=}+ \sum_{\substack{e \in \mathbb{T}\\\delta p^e(q) > 0}}\bigg(\int_{(\mathbb{N}^*)^{\delta p^e(q)}} W([q + \delta q^e, v])\mathcal{V}^{i,a,\delta p^e(q)}(\diff v) - W(q) \bigg) \tilde{\Lambda}^{i, e}(q, y)\\
        &\phantom{=}+ \sum_{\substack{e \in \mathbb{T}\\\delta p^e(q) < 0}}\bigg(\int_{(\mathbb{N}^*)^{-\delta p^e(q)}} W([v,q + \delta q^e])\mathcal{V}^{i,b,-\delta p^e(q)}(\diff v) - W(q) \bigg) \tilde{\Lambda}^{i, e}(q, y).
    \end{split}
\end{equation*}
We denote these four terms by $A_1$, $A_2$, $A_3$ and $A_4$ respectively. We have
\begin{equation*}
    A_1 \leqslant \sum_{\substack{e=(+,j,s,n) \in \mathbb{T}^{+}\\s= b \\ j-(b,1) \leqslant 0\\\delta p^e(q)= 0}} \big((q^j +n-Q-Q')_+ - (q^j-Q-Q')_+\big)\tilde{\Lambda}^{i, e}(q, y)
    + \sum_{\substack{e=(+,j,s,n) \in \mathbb{T}^{+}\\s= a\\\delta p^e(q)= 0}} n \tilde{\Lambda}^{i, e}(q, y).
\end{equation*}
Using the definition of $y_0$, and expanding the square, we obtain
\begin{equation*}
    A_2 \leqslant
    \sum_{\substack{e \in \mathbb{T}\\\delta p^e(q) \neq 0}}
    -y \psi(y-\delta^i \delta p^e(q))\delta^i (2 \delta p^e(q) - 1)\tilde{\Lambda}^{i,e}(q,y).
\end{equation*}
Using the definition of $y_1$, and recalling
\begin{equation*}
    \sum_{\substack{e \in \mathbb{T}\\\delta p^e(q) < 0}} \tilde{\Lambda}^{i,e}(q,y) \leqslant M,
\end{equation*}
we get
\begin{equation*}
    A_2 \leqslant y\psi(y) M - M' \sum_{\substack{e \in \mathbb{T}\\\delta p^e(q) > 0}} \tilde{\Lambda}^{i,e}(q,y).
\end{equation*}
As for $A_3$,
\begin{equation*}
    \begin{split}
        A_3 & \leqslant \sum_{\substack{e=(-,(s,l),n) \in \mathbb{T}^-\\s=a\\ \delta p^e(q) > 0}}(-n + M')\tilde{\Lambda}^{i, e}(q, y)
        + \sum_{\substack{e=(\text{modif},(s,l),j,n) \in \mathbb{T}^{modif}\\s=a\\ \delta p^e(q) > 0}} M'\tilde{\Lambda}^{i, e}(q, y)\\
        & \phantom{\leqslant}
        \sum_{\substack{e=(\text{modif},(s,l),j,n) \in \mathbb{T}^{modif}\\s=b\\ \delta p^e(q) > 0}} M'\tilde{\Lambda}^{i, e}(q, y)
        + \sum_{\substack{e=(+,(s,l),n) \in \mathbb{T}^+\\s=b\\ \delta p^e(q) > 0}}((n-Q-Q')_+ + M')\tilde{\Lambda}^{i, e}(q, y),
    \end{split}
\end{equation*}
the last term corresponding to an insertion in the spread. Similarly,
\begin{equation*}
    \begin{split}
        A_4 & \leqslant \sum_{\substack{e=(-,(s,l),n) \in \mathbb{T}^-\\s=b\\ \delta p^e(q) < 0}}M'\tilde{\Lambda}^{i, e}(q, y)
        + \sum_{\substack{e=(\text{modif},(s,l),j,n) \in \mathbb{T}^{modif}\\s=b\\ \delta p^e(q) < 0}}(n+ M')\tilde{\Lambda}^{i, e}(q, y)\\
        & \phantom{\leqslant}
        \sum_{\substack{e=(\text{modif},(s,l),j,n) \in \mathbb{T}^{modif}\\s=a\\ \delta p^e(q) < 0}} M'\tilde{\Lambda}^{i, e}(q, y)
        + \sum_{\substack{e=(+,(s,l),n) \in \mathbb{T}^+\\s=a\\ \delta p^e(q) < 0}}(n + M')\tilde{\Lambda}^{i, e}(q, y).
    \end{split}
\end{equation*}
By Assumptions \ref{assumption:non_exploding_queue} and \ref{assumption:finite_expectation_size_orders},
\begin{equation*}
    \sum_{\substack{e=(+,j,s,n) \in \mathbb{T}^{+}\\s= b \\ j-(b,1) \leqslant 0\\\delta p^e(q)= 0}} \big((q^j +n-Q-Q')_+ - (q^j-Q-Q')_+\big)\tilde{\Lambda}^{i, e}(q, y)
    +\sum_{\substack{e=(+,(s,l),n) \in \mathbb{T}^+\\s=b\\ \delta p^e(q) > 0}}(n-Q-Q')_+\tilde{\Lambda}^{i, e}(q, y)
\end{equation*}
is uniformly (over $(q,y)$) bounded by some constant $C_0$. We also have
\begin{equation*}
    \begin{split}
        &\sum_{\substack{e=(+,j,s,n) \in \mathbb{T}^{+}\\s= a\\\delta p^e(q)= 0}} n \tilde{\Lambda}^{i, e}(q, y)
        +A_4
        \leqslant M(1+M')
    \end{split}
\end{equation*}
by the definition of $M$.
Recalling
\begin{equation*}
    \sum_{\substack{e=(-,(s,l),n) \in \mathbb{T}^-\\s=a\\ \delta p^e(q) > 0}}n\tilde{\Lambda}^{i, e}(q, y) \geqslant y \xi(y),
\end{equation*}
and combining all the former inequalities, we obtain
\begin{equation*}
    \mathcal{M}_i V(y,q)
    \leqslant C_0 + M(1+M') - y \xi(y) + y M\psi(y).
\end{equation*}
Since $\lim_{y \to \infty} (\xi(y) - M\psi(y)) = \infty$, we deduce the desired result.

\section{Proofs for the likelihood}
\label{sec:proof_likelihood}

In this section, we prove the propositions stated in Section \ref{sec:estimation}.

\subsection{Proof of Proposition \ref{prop:martingale}}
\label{subsec:proof_martingale}

The proof follows the same lines as the proof of \parencite[Lemma 7.3]{abi2019affine}. Let $T > 0$ and $\theta \in \Theta$. Let $(K_n)_{n \in \mathbb{N}}$ be an increasing sequence of compact subsets of $\mathcal{X} \times \mathbb{R}^d $ such that $\bigcup_{n \in \mathbb{N}}K_n = \mathcal{X} \times \mathbb{R}^d$. For $n \in \mathbb{N}$, define the stopping time $\tau_n$ by
    \begin{equation*}
        \tau_n \defeq \inf\{t \geqslant 0: (X_{t-}, S_t^{\theta} -P_{t}) \notin K_n\} \wedge T.
    \end{equation*}
    By \parencite[Theorem 2.4]{Sokol2012OptimalNC}, $(Z^{\theta}_{t \wedge \tau_n})_t$ is a true martingale because $(\Lambda^k_{\theta}(X_{t-}, S^{\theta}_t - P_{t-}))_{t \in [0, \tau^n]}$ is uniformly bounded. We can thus define, for each $n$, a probability $\mathbb{Q}^n$ on $(\Omega, \mathcal{F})$ such that $\frac{\diff \mathbb{Q}^n}{\diff \mathbb{P}} = Z_{\tau^n}^{\theta}$. Fix $n \in \mathbb{N}$. Let $\phi$ be a bounded measurable function. Then,
    \begin{equation*}
        \E[\mathbb{Q}^n][\phi(S_0, P_0, X_0, N_0)]
        = \mathbb{E}\big[\phi(S_0, P_0, X_0, N_0)\mathbb{E}[Z_{\tau^n}^{\theta}| \mathcal{F}_0]\big]
        = \E[][\phi(S_0, P_0, X_0, N_0)].
    \end{equation*}
    Hence, the law of $(S_0, P_0, X_0, N_0)$ under $\Proba$ and $\mathbb{Q}^n$ is the same. Let $\phi:\mathbb{R}^d \times \mathcal{P} \times \mathcal{X} \times \mathbb{N}^{M}$ be a bounded measurable function twice differentiable in its first set of variables, with bounded second derivatives. Let $t \geqslant 0$. By the product rule,
    \begin{align*}
        \E[\mathbb{Q}^n][\phi(S^{\theta}_{t \wedge \tau^n}, P_{t\wedge \tau^n}, X_{t\wedge \tau^n}, N_{t \wedge \tau^n})] &=
        \E[][\phi(S^{\theta}_{t \wedge \tau^n}, P_{t\wedge \tau^n}, X_{t\wedge \tau^n}, N_{t \wedge \tau^n})Z^{\theta}_{\tau^n}]\\
        &=\E[][\phi(S_0, P_0, X_0, N_0)] \\
        &\phantom{=} + 
        \frac{1}{2} \mathbb{E}\bigg[\int_0^{t \wedge \tau^n}Z_u^{\theta}\tr\big(\Sigma_{\theta} \Sigma_{\theta}^T \nabla^2_y \phi(S^{\theta}_u, P_u, X_u, N_u) \big) \diff u\bigg]\\
        & \phantom{=}+ \sum_{k \in \mathbb{M}}\mathbb{E}\bigg[
            \int_0^{t\wedge \tau^n} \int_{\mathcal{X} \times \mathcal{P}}Z_u^{\theta} \big( \Lambda^k_{\theta}(X_u, S^{\theta}_u - P_u)\phi(S^{\theta}_u, P_u + p', S^{\theta}_u + x', N_u + e_k) \\
            & \phantom{=+ \sum_{k \in \mathbb{M}}\mathbb{E}\bigg[
                \int_0^{t\wedge \tau^n}\int_{\mathcal{X} \times \mathcal{P}} Z_u^{\theta} \big(}- \phi(S^{\theta}_u, P_u, X_u, N_u) \mu(x,p,\diff x', \diff p') \big)\mathds{1}_{\mathcal{X}_+^k}(X_u) \diff u    
        \bigg]\\
        &\phantom{=}+\mathbb{E}\bigg[\int_0^{t \wedge \tau^n}\int_{\mathcal{X} \times \mathcal{P}} Z_u^{\theta}\big(\phi(S^{\theta}_u, P_u+p', X_u +x', N_u)\\
        &\phantom{=+\mathbb{E}\bigg[\int_0^{t \wedge \tau^n}\int_{\mathcal{X} \times \mathcal{P}} Z_u^{\theta}\big(}-\phi(S^{\theta}_u, P_u, X_u, N_u) \big)\mu(x,p,\diff x', \diff  p')\Lambda(X_u, S^{\theta}_u-P_u)\diff u \bigg] \\
        & \phantom{=}+
        \sum_{k \in \mathbb{M}}\mathbb{E}\bigg[\int_0^{t \wedge \tau^n} Z_u^{\theta}\big(1-\Lambda_{\theta}^k(X_u, S^{\theta}_u - P_u) \big)\phi(S^{\theta}_u, P_u, X_u, N_u) \mathds{1}_{\mathcal{X}_+^k}(X_u) \diff u \bigg]\\
        & = \E[][\phi(S_0, P_0, X_0, N_0)] + \mathbb{E}^{\mathbb{Q}^n}\bigg[\int_0^{t \wedge \tau_n}
        \hat{\mathcal{A}}_{\theta}\phi(S_u, P_u, X_u, N_u) \diff u
        \bigg].
    \end{align*}
    Hence, under $\mathbb{Q}^n$, $(S, P, X, N)_{\cdot \wedge \tau_n}$ has the law of a Markov process with infinitesimal generator $\hat{\mathcal{A}}_{\theta}$ stopped at the exit time of the process $(X, S-P)$ of $K_n$. We have $\lim_{n \to \infty} \mathbb{Q}^n[\tau_n < T]= 0$ because a Markov process with generator $\hat{\mathcal{A}}_{\theta}$ is non-explosive, by assumption \ref{assumption:non-explosiveness}. Then,
    \begin{equation*}
        \E[][Z_T^{\theta}] \geqslant \lim_{n \to \infty} \E[][Z_T^{\theta}\mathds{1}_{\{\tau^n = T\}}] = \lim_{n \to \infty} \mathbb{Q}^n[\tau^n = T] = 1,
    \end{equation*}
    which is the desired result.

\subsection{Proof of Lemma \ref{lem:parabolic_pde_existence}}
\label{subsec:proof_parabolic_pde_existence}

We use the notations from \parencite{friedman_partial_1983} for local Hölder norms $|f|_{\alpha}$, $|f|_{2+ \alpha}$,\dots For $r > 0$, let $B_r \subset \mathbb{R}^d$ be the open ball of length $d$ for the Euclidean norm. Let $T > 0$. Let $\phi:\mathbb{R}^d \to [0, 1]$ be a $C^{\infty}$ function with support included in $B_{2}$ and taking the value 1 in $\bar{B}_1$. Moreover, suppose that for any $y \in \mathbb{R}^d$, $r \in (0, \infty) \to \phi(ry)$ is non-decreasing. For $n \in \mathbb{N}^*$, and $y\in \mathbb{R}^d$ define $\phi_n(y) \defeq \phi(ny)$. By \textcite[Chapter 3, Theorem 9]{friedman_partial_1983}, there exists a unique $u_n \in C([0,T] \times \bar{B}_{2n}) \cap C^{1,2}((0,T) \times B_{2n})$ solving \eqref{eq:linear_pde} and such that $u_n(T, y) = f(y)\phi_n(y)$ for all $y \in \bar{B}_{2n}$ and $u_n(t,y) = 0$ on $\partial B_{2n}$. Let $K \subset \mathbb{R}^d$ be a compact set, $y \in K$, and $t \in [0,T]$. Let $m<n$ be two positive integers such that $K \subset B_m$. By the Feynman-Kac formula,
    \begin{align*}
        0 \leqslant u_n(t, y) - u_m(t, y) &=
        \mathbb{E}\Big[e^{-\int_{t}^T c(y + \Sigma(W_u-W_t))\diff u} (f\phi_n)(y + \Sigma(W_T-W_t))\mathds{1}_{\{\forall u \in [t, T],x+\Sigma(W_u-W_t) \in B_{2n}\}}\\
        &\phantom{=\mathbb{E}\Big[}
        -e^{-\int_{t}^T c(y + \Sigma(W_u-W_t))\diff u} (f\phi_m)(y + \Sigma(W_T-W_t))\mathds{1}_{\{\forall u \in [t, T],x+\Sigma(W_u-W_t) \in B_{2m}\}}
        \Big]\\
        &\leqslant
        \mathbb{E}\Big[f(y + \Sigma(W_T-W_t))\mathds{1}_{\{\exists u \in [t, T],x+\Sigma(W_u-W_t) \notin B_{m}\}} \Big].
    \end{align*}
    By the Cauchy-Schwarz inequality,
    \begin{equation*}
        |u_n(t, y) - u_m(t, y)| \leqslant
        \sup_{(t', y') \in [0,T] \times K}\mathbb{E}\big[f(y' + \Sigma(W_T-W_{t'}))^2\big]^{\frac{1}{2}}
        \Proba[][\exists u \in [t, T],x+\Sigma(W_u-W_t) \notin B_{m}]^{\frac{1}{2}}.
    \end{equation*}
    The first factor is finite thanks to the assumption $f\in \mathcal{C}_e$. The second one can be bounded by a function depending only on $\Sigma$, $T$, and the distance between $K$ and $\mathbb{R}^d \setminus B_m$ having limit $0$ as $m\to \infty$. Thus, $(u_n)$ is a Cauchy sequence for the sup-norm on $[0,T] \times K$. We conclude that there exists a continuous function $u$ on $[0,T] \times \mathbb{R}^d$ such that $(u_n)$ converges uniformly to $u$ on every compact. Moreover, \eqref{eq:linear_pde_final} is trivially verified.

    By \textcite[Chapter 3, Theorem 5]{friedman_partial_1983}, for every $n_0 \in \mathbb{N}^*$, and $\alpha \in (0,1)$ such that $c$ is $\alpha$-Hölder continuous on $B_{2n_0}$, there exists a constant $C > 0$ such that,
    \begin{equation*}
        |u_n|_{2+\alpha,(0,T)\times B_{2n_0}} \leqslant C|u_n|_{\infty,(0,T)\times B_{2n_0}},\quad n \geqslant n_0.
    \end{equation*}
    The right-hand side of the last inequality in bounded in $n$, hence, by the Arzelà-Ascoli theorem, for all $\epsilon > 0$,  a subsequence of $(u_n)$, $(\partial_t u_n)$ and $(\nabla^2_y u_n)$ converges uniformly in $[0, T - \epsilon] \times B_{n_0}$ to $g^{\epsilon}$, $\partial_t g^{\epsilon}$ and $\nabla^2_y g^{\epsilon}$ respectively, for some $C^{1,2}$ function $g^{\epsilon}$. Since it is the case of every $u_n$, $g^{\epsilon}$ solves \eqref{eq:linear_pde} for $u$. By uniqueness of the limit, $g^{\epsilon}=u$. This being valid for all $n_0$ and $\epsilon$, we have that $u \in C([0,T] \times \mathbb{R}^d) \cap C^{1,2}((0,T) \times \mathbb{R}^d)$ and solves the desired PDE.

    Let $A > 0$ and $B \in [1,2)$ such that for all $y \in \mathbb{R}^d$, $|f(y)| \leqslant A e^{|y|^B}$. Let $(t,y) \in [0,T] \times \mathbb{R}^d$ and $n \in \mathbb{N}^*$. Then, by Feynman-Kac,
    \begin{equation*}
        u_n(t,y) \leqslant A \mathbb{E}\Big[e^{|y + \Sigma (W_T - W_t)|^B} \Big] \leqslant A \mathbb{E}\Big[ e^{2^{B-1}|\Sigma(W_T - W_t)|^B} \Big] e^{2^{B-1}|y|^B}.
    \end{equation*}
    Passing to the limit, this inequality also holds with $u$ instead of $u_n$, thus $u(t,\cdot) \in \mathcal{C}_e$.

    Uniqueness is a direct consequence of the Feynman-Kac representation of such a solution.

\subsection{Proof of Theorem \ref{thm:PDE_characterization_lik}}
\label{subsec:proof_PDE_characterization_lik}

Existence and uniqueness of $u(\omega)$ for every $\omega$ in a subset $\Omega'$ of $\Omega$ having probability one is a direct consequence (using backwards induction) of Lemma \ref{lem:parabolic_pde_existence}, its subsequent remark, and the fact that if $f \in \mathcal{C}_e$, then $\Lambda^k_{\theta}(\cdot -p, x)f \in \mathcal{C}_e$ for any $(k, x, p)$, by Assumption \ref{assumption:intensities_regularity}.
    
    By the independence of $S^{\theta}$ and $(X, P, N)$ under $\mathbb{P}$, we have that
    \begin{equation*}
        \E[][Z_T^{\theta}| \mathcal{F}_T^{obs}]
        = \int_{\mathbb{R}^d} \psi(y, X, P, N) m_{S_0}(\diff y)\quad \mathbb{P}-a.s.,
    \end{equation*}
    where
    \begin{equation*}
        \begin{split}
            \psi(y,X(\omega), P(\omega), N(\omega))
        = \mathbb{E}\bigg[ \exp\bigg(&\sum_{k \in \mathbb{M}} \int_{(0,T]} \ln\big(\Lambda^k_{\theta}(X_{u-}(\omega), S^{\theta}_u - P_{u-}(\omega))\big) \diff N^k_{u}(\omega) \\
        &+ \sum_{k \in \mathbb{M}} \int_0^{T} \big(1-\Lambda_{\theta}^k(X_{u-}(\omega), S^{\theta}_u - P_{u-}(\omega))\big)\mathds{1}_{\mathcal{X}^k_+}(X_{u-}(\omega)) \diff u \bigg)\bigg| S_0 = y\bigg]
        \end{split}
    \end{equation*}
    for $\omega$ in some measurable subset $\Omega'' \subset \Omega'$ having probability one and such that on $\Omega''$, $(X,P,N)$ has a finite number of jumps. For $\omega \in \Omega''$ and $(t,y) \in [0,T] \times \mathbb{R}$, we have, by the Feynman-Kac formula and backwards induction on the jump times,
    \begin{equation*}
        \begin{split}
            u(\omega)(t,y)
        = \mathbb{E}\bigg[ \exp\bigg(&\sum_{k \in \mathbb{M}} \int_{(t,T]} \ln\big(\Lambda^k_{\theta}(X_{u-}(\omega), y + \Sigma_{\theta}(W_u - W_t) - P_{u-}(\omega))\big) \diff N^k_{u}(\omega) \\
        &+ \sum_{k \in \mathbb{M}} \int_t^{T} \big(1-\Lambda_{\theta}^k(X_{u-}(\omega), y + \Sigma_{\theta}(W_u - W_t) - P_{u-}(\omega))\big)\mathds{1}_{\mathcal{X}^k_+}(X_{u-}(\omega)) \diff u \bigg)\bigg].
        \end{split}
    \end{equation*}
    We deduce the desired result at $t=0$.

\printbibliography

\end{document}